\documentclass[review,onefignum,onetabnum]{siamart171218}

\usepackage[applemac]{inputenc}
\usepackage{lipsum}
\usepackage{amsfonts}
\usepackage{graphicx}
\usepackage{float}
\usepackage{epstopdf}
\usepackage{algorithmic}
\usepackage[algo2e]{algorithm2e} 
\usepackage{scrextend}
\usepackage{standalone}
\usepackage{verbatim}
\usepackage{hyperref}

\ifpdf
  \DeclareGraphicsExtensions{.eps,.pdf,.png,.jpg}
\else
  \DeclareGraphicsExtensions{.eps}
\fi

\newcommand{\ground}{\textsf{Ground}}

\newsiamremark{remark}{Remark}
\newsiamremark{hypothesis}{Hypothesis}
\crefname{hypothesis}{Hypothesis}{Hypotheses}
\newsiamthm{claim}{Claim}

\headers{Computational Complexity of Biased DLA}{N. Bitar, E. Goles, P. Montealegre}

\title{Computational Complexity of Biased Diffusion-Limited Aggregation \thanks{Submitted to the editors \today.
}}

\author{Nicolas Bitar \thanks{Departamento de Ingenier\'ia Matem\'atica, Universidad de Chile, Santiago, Chile
  (\email{nbitar@dim.uchile.cl}).}
\and Eric Goles \thanks{Facultad de Ingenier\'ia y Ciencias, Universidad Adolfo Ib\'a\~nez, Santiago, Chile
  (\email{eric.chacc@uai.cl}, \email{p.montealegre@uai.cl}).}
\and Pedro Montealegre \footnotemark[3]}

\usepackage{amsopn}

\ifpdf
\hypersetup{
  pdftitle={Computational Complexity of Biased Diffusion-Limited Aggregation},
  pdfauthor={Nicolas Bitar, Eric Goles, Pedro Montealegre}
}
\fi

\usepackage{cite}
\usepackage{hyperref}
\usepackage{float}  
\usepackage{enumerate}
\usepackage{stmaryrd}

\usepackage{latexsym}
\usepackage{amsmath,amsfonts,amssymb}
\usepackage{bm}
\usepackage{tikz}
\usepackage{algorithm}
\usepackage{algorithmic}
\usepackage{color, colortbl}
\usepackage{float}

  \usepackage{array}
\usepackage{framed}

\newcommand{\PP}{{\bf P}}
\newcommand{\Lspace}{{\bf L}}
\newcommand{\PC}{{\bf P}-Complete}
\newcommand{\NLC}{{\bf NL}-Complete}
\newcommand{\NC}{{\bf NC}}
\newcommand{\NP}{{\bf NP}}
\newcommand{\NL}{{\bf NL}}
\newcommand{\Pred}{\textsc{1-Prediction}}
 \newcommand{\DPred}{\textsc{DLA-Prediction}}
  \newcommand{\DReal}{\textsc{DLA-Realization}}
  \newcommand{\BPred}{\textsc{BD-Prediction}}
  
  \newcommand{\True}{\textsc{True}}
   \newcommand{\False}{\textsc{False}}

  \newcommand{ \LDEReach}{\textsc{LDE-Reachability}}
  
 \newcommand{\col}{\textnormal{col}}
  \newcommand{\num}{\textnormal{num}}
  \newcommand{\pos}{\textnormal{pos}}
  \newcommand{\row}{\textnormal{row}}

\newcommand{\cO}{\mathcal{O}}

\newcommand{\defproblemu}[3]{
  \vspace{1mm}
\noindent\fbox{
  \begin{minipage}{0.95\textwidth}
  #1 \\
  {\bf{Input:}} #2  \\
  {\bf{Question:}} #3
  \end{minipage}
  }
  \vspace{1mm}
}

\begin{document}
\raggedbottom

\maketitle

\begin{abstract}
Diffusion-Limited Aggregation (DLA) is a cluster-growth model that consists in a set of particles that are sequentially aggregated over a two-dimensional grid. In this paper, we introduce a \emph{biased} version of the DLA model, in which particles are limited to move in a subset of  possible directions. We denote by $k$-DLA the model where the particles move only in $k$ possible directions. We study the biased DLA model from the perspective of Computational Complexity, defining two decision problems The first problem is \textsc{Prediction}, whose input is a site of the grid $c$ and a sequence $S$ of walks, representing the trajectories of a set of particles. The question is whether a particle stops at site $c$ when sequence $S$ is realized. The second problem is \textsc{Realization}, where the input is a set of positions of the grid, $P$. The question is whether there exists a sequence $S$ that realizes $P$, i.e. all particles of $S$ exactly occupy the positions in $P$. Our aim is to classify the \textsc{Prediciton} and \textsc{Realization} problems for the different versions of DLA. We first show that \textsc{Prediction} is {\PC} for 2-DLA (thus for 3-DLA). Later, we show that Prediction can be solved much more efficiently for 1-DLA. In fact, we show that in that case the problem is \NLC. With respect to \textsc{Realization}, we show that restricted to 2-DLA the problem is in \PP, while in the 1-DLA case, the problem is in \Lspace.

\end{abstract}



\section{Introduction}
\label{intro}

Diffusion-Limited Aggregation (DLA) is a kinetic model for cluster growth, first described by Witten and Sander \cite{PhysRevLett.47.1400}, which consists of an idealization of the way dendrites or dust particles form, where the rate-limiting step is the diffusion of matter to the cluster.  The original DLA model consists of a series of particles that are thrown one by one from the top edge of a d-dimensional grid, where $d\geq 2$. The sites in the grid can either be occupied or empty. Initially all the sites in the grid are empty, except for the bottom line which begins and remains occupied. Each particle follows a random walk in the grid, starting from a random position in the top edge, until it neighbors an occupied site, or the particle escapes from the top edge or one of the lateral edges. In case the particle finds itself neighboring an occupied site, the current position of the particle becomes occupied and the next particle is thrown.  The set of occupied sites is called a \emph{cluster}. \\

\begin{figure}[h]
	\begin{center}
	\includegraphics[width=0.6\textwidth]{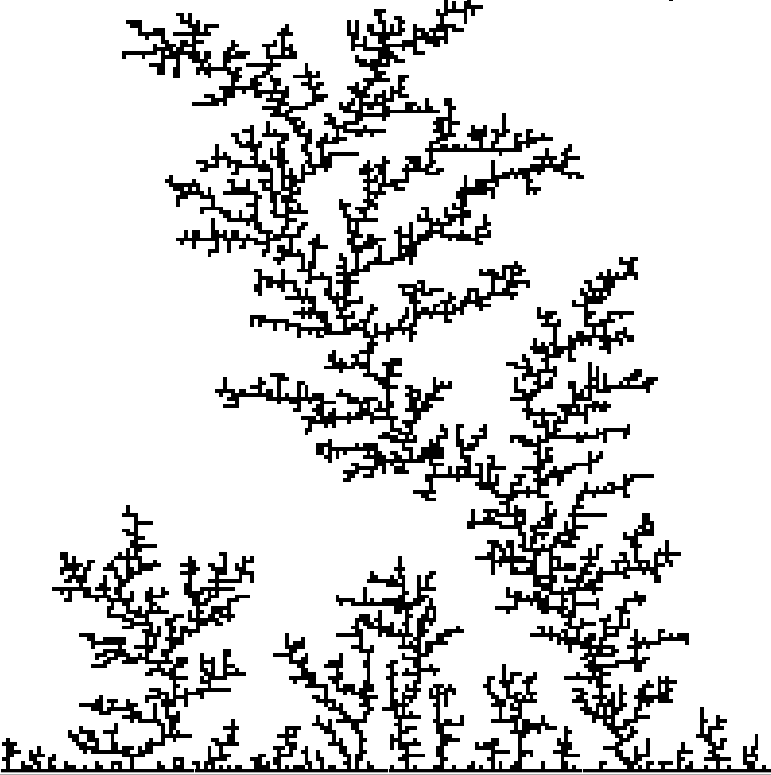}
	\label{fig:1} 
	\caption{A realization of the dynamics for a $200 \times 200$ grid.}
	\end{center}     
\end{figure}

Clusters generated by the dynamics are highly intricate and fractal-like (see Figure \ref{fig:1}); they have been shown to exhibit the properties of scale invariance and multifractality \cite{meakin1983diffusion, halsey1986scaling}.  DLA clusters have been observed to appear in electrodeposition, dielectrics and ion beam microscopy \cite{brady1984fractal, nittmann1985fractal, niemeyer1984fractal}.  Nevertheless, perhaps the fundamental aspect of DLA is its profound connection to Hele-Shaw flow: it has been shown that DLA is its stochastic counterpart
\cite{chan1988simulating, koza1991equivalence}.

In this article we study \emph{restricted versions of DLA}, which consist in the limitation of the directions a particle is allowed to move within the grid.  We ask what would be the consequences of restricting the particles movement in terms of computational complexity. More precisely, we consider four models, parameterized by their number of allowed movement directions, $k\in \{1,2,3,4\}$. The $4$-DLA model is simply the two-dimensional DLA model, i.e. when the particles can move in the $4$ cardinal directions. The $3$-DLA model  is the DLA model when the particles can only move into the South, East or West direction. In the $2$-DLA model, the directions are restricted to the South and East. Finally, in the $1$-DLA model, the particles can only move downwards. \\ 

Even though the particles have restricted movements, it is possible to notice that the fractal-like structures are still present in the clusters obtained by the restricted DLA (see Figure \ref{fig:dynex}).\\

It is interesting to note that, in fact, the $1$-DLA model is a particular case of another computational model created to described processes in statistical physics, that of the Ballistic Deposition Model. In this model, there is a graph and a set of particles that are thrown into the vertices at some fixed initial height $h$. The particle's height decreases one unit at a time, until it reaches the bottom (height $0$ or meets another particle, i.e. there is a particle in an adjacent vertex at the same height, or in the same vertex just behind. The $1$-DLA model corresponds to the Ballistic Deposition model when the graph is an undirected path. \\

\begin{figure}[H]
	\begin{center}
	\includegraphics[width = 5cm
]{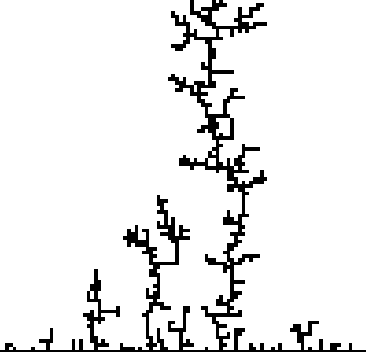}\hspace{0.5cm}
	\includegraphics[width = 5cm ]{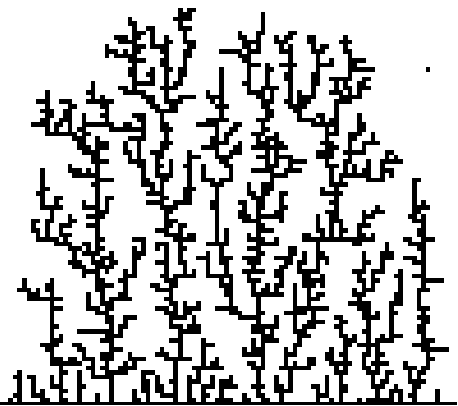}\vspace{0.5cm}
	\includegraphics[width = 5cm]{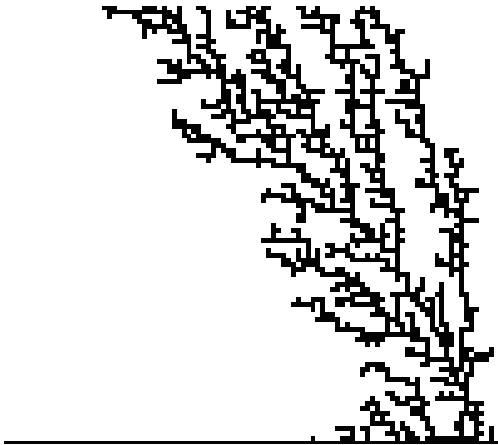}\hspace{0.5cm}
	\includegraphics[width = 5cm]{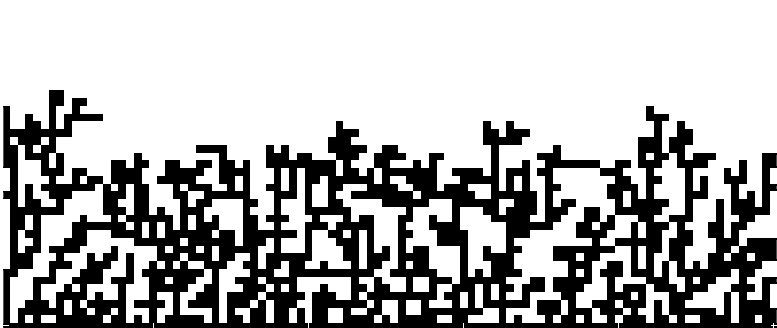}
	
	\caption{DLA simulation for four (top left), three (top right), two (bottom left) and one (bottom right) directions, when $N=100$. }
	\label{fig:dynex} 
	\end{center}     
\end{figure}

Due to the generated cluster's properties, theoretical approaches to the DLA model are usually in the realm of fractal analysis, renormalization techniques and conformal representations \cite{barra2002iterated}. In this article we consider a perhaps unusual approach to study the DLA model (and its restricted versions), related with its computational capabilities, the difficulty of simulating their dynamics, and the possibility of characterizing the patterns they produce. Machta and Greenlaw studied, within the framework of computational complexity theory, the difficulty of computing whether a given site on the grid becomes occupied after the dynamics have taken place, i.e. all the particles have stuck to the cluster or have been discarded \cite{machta1996computational}. Inspired by their work, we consider two decision problems:\\

\begin{itemize}
\item {\DPred}, which receives a sequence of trajectories for $n$ particles (i.e. the trajectories are deterministic and explicit) and the coordinates of a site in the lattice as input. The question is whether the given coordinate is occupied by a particle after the $n$ particles are thrown. \\
\item {\DReal}, which receives an  $n\times n$ sized pattern within the two-dimensional grid as input, and whose question is whether that pattern can be produced by the DLA model. \\
\end{itemize}

For each $k\in \{1, 2, 3\}$ we call {$k$-\DPred} and $k$-\DReal, respectively,  the problems {\DPred} and {\DReal} restricted to the $k$-DLA model.

The computational complexity of a  problem can be defined as the amount of resources, like time or space, needed to computationally solve it. Intuitively, the complexity of {\DPred} represent \emph{how efficiently} (in terms of computational resources) we are able to simulate the dynamics of DLA. On the other hand, the complexity of {\DReal} represent \emph{how complex} are the patterns produced by DLA model. \\

We consider four fundamental complexity classes: {\Lspace}, {\PP}, {\NL} and {\NP}. 
Classes {\Lspace} and {\PP} contain the problems that can be solved in a \emph{deterministic} Turing machine that use \emph{logarithmic space} and take \emph{polynomial time}, respectively.   On the other hand {\NL} and {\NP} are the classes of problems that can be solved in a \emph{non-deterministic} Turing machine that use \emph{logarithmic space} and take \emph{polynomial time}, respectively. For detailed definitions and characterizations of these classes we recommend the book of Arora and Barak \cite{arora2009computational}. A convention between computer theorists states that {\PP}  is the class of problems that can be solved \emph{efficiently} with respect to computation time. In that context, {\NP} can be characterized as the classes of problems that can be \emph{efficiently verified} with respect to computation time. Similarly, the same conventions hold for {\Lspace} and {\NL} changing computation time for space. \\

It is easy to see that $\Lspace \subseteq \NL \subseteq \PP \subseteq \NP$, though it is unknown if any of these inclusions is proper. Perhaps the most famous conjecture in Computational Complexity theory is \textbf{P} $\neq$ \textbf{NP}. Put simply, this conjecture states that there are some problems whose solution can be efficiently verified but cannot be efficiently found. As mentioned in last paragraph, in this context \emph{efficiently} means polynomial time. Similarly, it is conjectured that  $\Lspace \neq \NL$, where in this case efficiently refers to logarithmic space.  It is also conjectured that $\NL \neq \PP$, meaning that some problems can be computed efficiently with respect to computation time, but cannot be verified (or computed) efficiently with respect to space \cite{arora2009computational}. \\

The problems in $\PP$ that are the most likely to not belong to $\NL$ (hence not in $\Lspace$) are the {\PC} problems \cite{Greenlaw:1995}. A problem is {\PC} if any other problem in $\PP$ can be reduced to it via a \emph{log-space reduction}, i.e. a function calculable in logarithmic space that takes \emph{yes}-instances of one problem into the other. In a nutshell, it is  unlikely that some {\PC} problem belongs to {\NL}, because in that case we would have that $\NL = \PP$. Similarly a problem is {\NLC} if any problem in {\NL} can be reduced to it via a {\Lspace} reduction. {\NLC} problems are problems in {\NL} that are the most likely to not belong to {\Lspace}, because if some {\NLC} problem were to belong in {\Lspace}, it would imply that {\NL = \Lspace} \cite{arora2009computational}. \\

One {\PC} problem is the \textsc{Circuit Value Problem (CVP)}, which consists in, given a Boolean circuit and an truth-assignment of its input gates, compute the output value of a given gate. Roughly, this problem is unlikely to be solvable (or verifiable) in logarithmic space because there is no better algorithm than simply sequentially computes truth values of each gate of the Boolean circuit, keeping in memory the values of all gates already evaluated. One {\NLC} problem is {\sc Reachability}, which consists in, given a directed graph $G$ and two vertices $s$ and $t$, decide if there is a directed path between $s$ and $t$. Roughly, this problem is unlikely to be solvable in logarithmic space because there is no way to remember all paths starting from $s$ that do not reach $t$. \\

Within this context, it was shown by Machta and Greenlaw \cite{machta1996computational}, that {\DPred} is \PC.  The proof of this fact consists of reducing to it a version of the Circuit Value Problem, which is known to be {\PC}\cite{Greenlaw:1995}.   Within this proof, we noticed that the gadgets used to simulate the circuits rely heavily on the fact that in the DLA model, particles are free to move in any of the four cardinal directions.  On the other hand, in the context of the study of the Ballistic deposition model it was proven by Matcha and Greenlaw \cite{machta1994parallel} that  $1$-{\DPred} is in $\NC$. The class $\NC$ is a complexity class that contains $\NL$ and is contained in $\PP$.




\subsection{Our results}

We begin our study of the complexity of biased DLA analyzing the complexity of the {\DPred} problem.  By extending the results of Machta and Greenlaw to the $2$-DLA model, we show that $2$-{\DPred} is {\PC}. This result is obtained following essentially the same gadgets used for the non-restricted case, but carefully constructing them using only two directions. More precisely, the construction of Machta and Greenlaw consists in a representation of an instance of {\sc Circuit Value Problem} as a sequence of particle throws, which final positions represent the input Boolean circuit with its gates evaluated on the given input. A gate evaluated \emph{true} is represented by a path of particles, while the \emph{false} signals are represented by the lack of such path. Since the construction of the circuit must be done in logarithmic space, the sequence of particles must be defined without knowing the actual output of the gates. Therefore, the trajectory of the particles considers that, if they are do not stick on a given position, they escape through the top edge of the grid. Since this escape movement is not possible in the $2$-DLA (nor $3$-DLA) model (because particles can not move upwards), we modify the circuit construction to build the gates in a specific way, in order to give the particles enough space to escape through the rightmost edge or deposit at the bottom of the grid without disrupting the ongoing evaluation of the circuit.

The fact that $2$-{\DPred} is {\PC} directly implies that  $3$-{\DPred} is also \PC, settling the prediction problem for these two biased versions of the model.

We then study the $1$-DLA model. Despite of what one might guess, the dynamics of the DLA model restricted to one direction are far from trivial (see Section \ref{sec:restrictdir} for examples of the patterns produced by this model).  Indeed, we begin our study showing that this dynamics can simulate simple sorting algorithms like \emph{Bead-Sort}. Then, we improve the result of Matcha and Greenlaw by showing that $1$-{\DPred} is in $\NL$. This is in fact an improvement, because they showed that $1$-{\DPred} is in $\NC^2$ \cite{machta1994parallel}, and $\NL$ is a sub-class of $\NC^2$ \cite{Greenlaw:1995}.  Our result holds for the Ballistic Deposition model, i.e., when the graph is not restricted to a path but is an arbitrary graph given in the input (we call that problem {\BPred}).
We finish our study of the prediction problem showing that the complexity of {\BPred} \emph{cannot be improved}. Indeed, we show that {\BPred} is {\NLC}.

After this, we study the {\DReal} problem. We observe that $k$-{\DReal} is in {\NP} for all $k\in \{1,2, 3, 4\}$. Moreover, the non-deterministic aspect of the $\NP$ algorithm solving {\DReal} only needs to obtain the order of the sequence on which the particles are placed on the grid, rather than obtaining both the order and the trajectory that each particle follows. In fact, the trajectories can be computed in polynomial time, given the order in which the particles are placed in the grid. 

We then show $1$-{\DReal} can be solved much more efficiently. In fact, we give a characterization of the patterns that the $1$-DLA model can produce. Our characterization is based on a planar directed acyclic graph (PDAG) that represents the possible ways in which the particles are able to stick. Each occupied cell of the grid is represented by a node of the PDAG, plus a unique sink vertex that represents the \emph{ground}. We show that a pattern can be constructed by $1$-DLA if and only if there is a directed path from every vertex to the ground.  We use our characterization to show that $1$-{\DReal} is in $\Lspace$, using a result of Allender \emph{et al.}  \cite{allender2006grid} solving \textsc{Reachability} in log-space, when the input graph is a single sink PDAG.

Finally, we give an efficient algorithm to solve the realization problem in the $2$-DLA model,  showing that $2$-{\DReal} is in \PP. Our algorithm uses the fact that in the $2$-DLA model, the particles are placed into the grid in a very specific way. By establishing which particles force the order in which they are placed. We use this fact to efficiently compute the order in which the particles are thrown, obtaining a polynomial-time algorithm. 

\subsection{Related work}

For dynamical properties of the restricted versions of DLA, including Ballistic Deposition, we refer the reader to \cite{asselah2016diffusion,  mansour2019ballistic, penrose2008growth, meakin1986ballistic}.

Some problems of similar characteristics have been studied in this context, such as the Ising Model, Eden Growth, Internal DLA and Mandelbrot Percolation, to name a few \cite{machta1996computational, moore2000internal}. On the other hand, the problem of Sandpile Prediction is an example where increasing the degrees of freedom increases the computational complexity of the prediction problem. In particular, when the dimension is greater than 3, the prediction problem is \PC; but when the dimension is 1, the problem is in $\NC$ \cite{moore1999computational}.

Another example of a complexity dichotomy that depends on the topology of the system is the Bootstrap Percolation model \cite{chalupa}. In this model, a set of cells in a $d$-dimensional grid are initially infected, in consecutive rounds, healthy sites that have more than the half of their neighbors infected become infected. In this model, the prediction problem consists in determining wether a given site becomes infected at some point in the evolution of the system. In \cite{GMT} it is shown that this prediction problem is {\PC} in three or more dimensions, while in two dimensions it is in {\NC}. Other problems related to Bootstrap percolation involve the maximum time that the dynamics takes before converging to a fixed point \cite{Benevides2015}.

\subsection{Structure of the article}

The first section formally introduces the different computational complexity classes, along side problems of known complexity used throughout the article. Next, the dynamics for the general case of DLA are presented, in addition to the formal definition of the two associated prediction problems that are discussed: Prediction and Realization. The third section focuses on the proof that DLA restricted to 2 or 3 directions is {\PC} and with the presentation of the non-deterministic log-space algorithm for the generalized version of the one direction DLA problem, Ballistic Deposition. The last section talks about the results concerning the Realization problem, where the one-directional case is shown to be solvable in {\Lspace}, and the two-directional case has a polynomial algorithm characterizing figures obtained from the dynamics.

\section{Preliminaries}

\subsection{Complexity Classes and Circuit Value Problem}

In this subsection we will define the main background concepts in computational complexity required in this article. For a more complete and formal presentation we refer to the books of Arora and Barak  \cite{arora2009computational} and  Greenlaw et al.  \cite{Greenlaw:1995}.  We assume that the reader is familiar with the basic concepts dealing with computational complexity.  As we mentioned in the introduction, in this paper we will  only consider complexity classes into which we classify the prediction problems.  {\PP} is the class of problems solvable in a  Turing machine that runs in polynomial time in the size of the input. More formally, if $n$ is the size of the input, then a problem is polynomial time solvable if it can be solved in time $n^{\cO(1)}$ in a deterministic Turing machine. 

A logarithmic-space Turing machine consists in a Turing machine with three tapes: a read-only \emph{input tape}, a write-only \emph{output-tape} and a read-write \emph{work-tape}. The Turing machine is allowed to move as much as it likes on the input tape, but can only use $\cO(\log n)$ cells of the work-tape (where $n$ is the size of the input). Moreover, once the machine writes something in the output-tape, it moves to the next cell and can not return. {\Lspace} is the class of problems solvable in a logarithmic-space Turing machine. 

A non-deterministic Turing machine is a Turing machine whose transition function does not necessarily output a single state but one over a set of possible states. A computation of the non-deterministic Turing machine considers all possible outcomes of the transition function. The machine is required to stop in every possible computation thread, and we say that the machine \emph{accepts} if at least one thread finishes on an accepting state. A non-deterministic Turing machine is said to run in \emph{polynomial-time} if every computation thread stops in a number of steps that is polynomial in the size of the input.  {\NP} is the class of problems solvable in polynomial time in a non-deterministic Turing machine. A non-deterministic Turing machine is said to run in \emph{logarithmic-space} if every thread of the machine uses only logarithmic space in the work-tape. {\NL} is the class solvable in logarithmic space in a non-deterministic Turing machine. 


 
 A problem $\mathcal{L}$ is $\PP$-Complete if it belongs to $\PP$ and any other problem in $\PP$ can be reduced to $\mathcal{L}$ via a logarithmic-space (many-to-one or Turing) reduction. A {\PC} problem belongs to {\NL} implies that ${\PP} ={\NL}$.

One well-known {\PC} problem is the \textsc{Planar-NOR-Circuit-Value-Problem}. A \emph{Planar NOR Boolean Circuit} is a planar directed acyclic graph $C$, where each vertex of $C$ has two incoming and two outgoing edges, except for some vertices that have no incoming edges (called \emph{inputs} of $C$) and others that have no outgoing edges (called \emph{outputs} of $C$). 

Each vertex of $C$ has a Boolean value ({\True} or {\False}). A truth assignment of the inputs of $C$, called $I$, is an assignment of values of the input gates of $C$. The value of a non-input gate $v$ is the NOR function (the negation of the disjunction) of the value of the two incoming neighbors of $v$. A truth assignment $I$ of $C$ defines a a truth-value of the output gates according to the truth values of the preceding gates.  We call $C(I)$ the truth values of the output vertices of $C$ when the input gates are assigned $I$. 

The \textsc{Planar-NOR-Circuit-Value Problem} is defined as follows.
\vspace{2 mm}

\defproblemu{\textsc{Planar-NOR-Circuit-Value Problem (PNORCVP)}}
   {A NOR Boolean Circuit $C$ of size $n$, a truth-assignment $I$ of $C$ and $g$ an output gate of $C$.}{Is $g$ {\True} in $C(I)$?}
 \vspace{2 mm}
 
 In \cite{Greenlaw:1995} it is shown that this problem is \PC.

\begin{proposition}[\cite{Greenlaw:1995}]\label{prop: NORCVPisPC} \textsc{PNORCVP} is \PC.
 \end{proposition}

 A problem $\mathcal{L}$ is {\NLC} if $\mathcal{L}$ belongs to $\NL$ and any other problem in $\NL$ can be reduced to $\mathcal{L}$ via a (many-to-one or Turing) reduction computable in logarithmic space. A {\NLC} problem belongs to {\Lspace} implies that ${\NL} ={\Lspace}$.
 
 One {\NL} problem is \textsc{Reachability}  \cite{arora2009computational}. An instance of \textsc{Reachability} is a directed graph $G$ and two vertices $s$ and $t$. The instance is accepted if there is a directed path from $s$ to $t$ in $G$. \textsc{Reachability} is {\NLC} because the computation  of a non-deterministic log-space Turing machine (which is the set of possible states of the machine plus contents of the working tape) can be represented by a directed graph of polynomial size, and the difficulty of finding a directed path in that graph is the difficulty of finding a sequence of transitions from the initial state to an accepting state. 
 
 For our reductions we will need  specific variant of \textsc{Reachability}, that we call \textsc{Layered DAG Exact Reachability} (\LDEReach). In this problem, the input graph $G$ is a directed acyclic graph (DAG), which is \emph{layered}, meaning that vertices of a layer only receive inputs of a previous layer and only output to a next layer (but vertices with in-degree zero are not necessarily in the first layer). Also, besides $G$ and $s$ and $t$, the input considers a positive integer $k\leq |G|$, where $|G|$ is the number of vertices of the input graph. The question is whether there exists a path of length exactly $k$ connecting vertices $s$ and $t$.  We show that this restricted version is also {\NL}-Complete. 

\begin{proposition}\label{prop:LDEReach}
{\LDEReach} is {\NLC}.
\end{proposition}

\begin{proof}
Let us first consider the problem \textsc{Exact-Reachability}. This problem receives as input a directed graph $G$, two vertices $s$ and $t$ and a positive integer $k\leq |G|$, and the question is whether there exists a directed path of length exactly $k$ connecting vertices $s$ and $t$. Is easy to see that \textsc{Exact-Reachability} is {\NL}-Complete. Indeed, it belongs to  {\NL} because an algorithm can simply nondeterministically choose the right vertices to follow in a directed path of length exactly $k$ from $s$ to $t$.  The verification of such path can be performed using $\cO(\log n)$ simply verifying the adjacency of the vertices in the sequence, and keeping a counter of the length of the path, that uses $\cO(\log n)$ space because $k\leq |G|$).  Observe also that  \textsc{Exact-Reachability} is {\NLC} because, if we have an algorithm $\mathcal{A}$ solving   \textsc{Exact-Reachability},  we can solve \textsc{Reachability} running $\mathcal{A}$ for $k\in \{0, \dots, |G|\}$. 

Observe now that  {\LDEReach} is in {\NL} for the same reasons than \textsc{Exact-Reachability}. We now show that, {\LDEReach} is {\NLC} reducing \textsc{Exact-Reachability} to it. Let $(G,s,t,k)$ be an instance of \textsc{Exact-Reachability}. Consider the instance $(G',s',t',k)$ in {\LDEReach} defined as follows. The set of vertices of $G'$ is a set of $k$ copies of $V(G)$, the set of vertices of $G$. We enumerate the copies from $1$ to $k$, and call them $V_1, \dots, V_k$. Then the set of vertices of $G'$ is $V(G') = V_1\cup \dots \cup V_k$. If $v$ is a vertex of $G$, we call $v_i$ the copy of $v$ that belongs to $V_i$.  There are no edges in $G'$ between vertices in the same copy of $V$. Moreover, if $u, v$ are two adjacent vertices in $G$, we add, for each $i\in \{1, \dots, k-1\}$, a directed edge from the $i$-th copy of $u$ to the $(i+1)$-th copy $v$ in $G'$, formally: $$(u,v) \in E(G) \iff (u_i, v_{i+1}) \in E(G'), \forall i \in \{1, \dots, k\}.$$ Finally, $s'=s_1$ and $t' = t_k$. By construction we obtain that $G'$ is layered, and moreover $(G,s,t,k)$ is a \emph{yes}-instance of \textsc{Exact-Reachability}  if and only if $(G',s',t',k)$ is a \emph{yes}-instance of {\LDEReach}. 

Finally, we remark that we can build the instance $(G', s', t', k)$ in log-space from $(G, s, t, k)$. Indeed, the algorithm has to simply make a counter $j$ from $1$ to $k$, and sequentially connect the vertices of the $j$-th copy of $V(G)$ with the vertices of the $j+1$ copy of $V(G)$.  We deduce that {\LDEReach} is {\NLC}.
\end{proof}

\subsection{The DLA Model and its Restricted Counterpart}\label{sec:restrictdir} 

The dynamics of the computational model of DLA are the following: We begin with a sequence of particles which will undergo a random walk starting from a position at the top edge of a $N\times N$ lattice. The sequence specifies the order in which the particles are released. Each particle moves until it neighbors an occupied site at which point it sticks to its position, growing the cluster. We begin with an occupied bottom edge of the lattice. If the particles does not stick to the cluster and leaves the lattice (exiting through the top or lateral edges), it is discarded. A new particle begins its random walk as soon as the previous particle sticks to the cluster or is discarded. This process goes on until we run out of particles in our sequence. 

To study this model from a computational perspective, it is convenient to consider a \emph{determinisitic} version, where the sequence of sites visited by each released particle is predefined. The prediction problem presented by Machta and Greenlaw \cite{machta1996computational} for a $d$-dimensional DLA is:

\vspace{2 mm}
\defproblemu{\DPred}
   {Three positive integers: $N$, $M$, $L$, a site $p$ in the two-dimensional lattice of size $N^2$, a list of random bits specifying $M$ particle trajectories of length $L$ defined by a site on the top edge of the lattice together with a list of directions of motion.}{Is site $p$ occupied after the particles have been thrown into the lattice?}
 \vspace{2 mm}
 
 For this prediction problem, it is shown in \cite{machta1996computational} that $\DPred$ is \PC. The proof consists of reducing a \PC \ variant of the Circuit Value Problem to the prediction problem. Their construction relies heavily on the fact that the particles can move in four directions (Up, Down, Left or Right).

The question we would like to answer is: what happens to the computational complexity of the prediction problem as we restrict the number of directions the particles are allowed to move along. Instead of the four permitted directions, we restrict the particle to move in three (left, right and downwards), two (right and downwards) and one (only downwards) directions.

We call the different directions by $d_1 =$ Down, $d_2 =$ Right, $d_3 =$ Left and $d_4 =$ Up. From this, we define the following class of prediction problems for $k\in \{1,2,3,4\}$:\\

\vspace{2 mm}
\defproblemu{\textsc{$k$-\DPred}}
  {Three positive integers: $N$, $M$, $L$, a site $p$ in the $N\times N$ lattice, a list of random bits specifying $M$ particle trajectories of length $L$ defined by a site on the top edge of the lattice together with a list of directions of motion, where the allowed directions of motions are $\{d_1, ... , d_k\}$.}{
  Is site $p$ occupied after the particles have been thrown into the lattice?}
 \vspace{2 mm}

Where $\DPred$ is the same as \textsc{4-\DPred}.\\

In addition, we ask for the computational complexity of determining whether a given pattern or figure is obtainable through the different biased dynamics. To do this, we codify a pattern on the two-dimensional grid as a binary matrix, where  0 represents an unoccupied site, and 1 represents an occupied one (an example of this can be found in Section 4.2). We define the computational problem as follows:

\vspace{2 mm}
\defproblemu{\textsc{$k$-\DReal}}
  {A 0-1 matrix $M$ codifying a pattern on the two-dimensional grid.}{
  Does there exist a sequence of particle throws that can move only on the allowed directions of motions, $\{d_1, ... , d_k\}$, whose end figure is represented by $M$?}
 \vspace{2 mm}
 
 Intuitively, this problem is concerned with understanding the complexity of the figures that can be obtained through the different versions of the DLA dynamics, by looking at the computational resources that are required to understand their structure.
 
 To our knowledge, this problem has not been studied from this angle before.
 
\section{{\DPred}} 

\subsection{Two and three directions}

In this section we show that the {2-\DPred} problem is \PC. This result directly implies that the $3$-\DPred \ problem is also \PC.

\begin{theorem}
\textsc{2-\DPred} is \PC.
\end{theorem}

\begin{proof}
We assume without loss of generality that the two directions in which the particles move are Down and Right. Our proof is an adaptation of the gadgets given by Machta and Greenlaw in \cite{machta1996computational} for the 4-directional case. The proof consists on using these directions of motion to simulate the evaluation of an instance of the \textsc{Planar-NOR-Circuit-Value Problem}.

In the construction of \cite{machta1996computational}, the different parts of the circuits are simulated by a series of gadgets, representing the different parts of the input Boolean circuit.  The gadgets are constructed as a series of particle walks, and the truth values are represented by prescribed locations occupied (or not) by a particle. When a gadget simulates the truth value {\False}, some locations remain unoccupied, meaning that some particles visit these locations but do not get stuck in them. This is realized specifying that, for some particles, the walk consideres a trajectory in two phases. First, the particle follows a given trajectory in order to reach a prescribed destination. Then, if the particle is not sticked, the walk continues in the same trajectory backwards. 

In this context, the difference between our constructions and that of \cite{machta1996computational} is that, in our limited version, the particles are not allowed to return through the same trajectory (in the case that they are not sticked in their destination). Therefore, our main challenge consists in designing a way to discard these particles. 

What we must first tackle is the way in which the information is transmitted through the different gadgets. For this purpose, we create wires that transmit the respective truth values. A coordinate occupied by a particle will represent \True, and an empty one will represent \False. Through this representation, a wire carries the value {\True} if a stack of particles grows along it, while wires that carry the value {\False} remain unoccupied. 

The wire is realized by having a pre-assigned particle for each site that makes it up.  See Figure \ref{fig:wire} for a representation of the gadget and Figure~\ref{fig:wire_dyn} for an example its dynamic. The path of each particle begins at the top of the lattice, and moves straight-down to its assigned location in the lattice. If the value the wire is transmitting is \True, the particle will stick at the site. If the value is \False, the particle will not stick. The particle is then instructed to move two positions to the right, and then to move indefinitely downwards to be discarded by means of getting stuck to the bottom, effectively transmitting the value of the wire. We note that the particle only realizes this trajectory in the case that it has not stuck, that is, the value transmitted by the wire is \False.

\begin{figure}[h]
	\begin{center}
	\tikzset{every picture/.style={line width=0.75pt}} 

\begin{tikzpicture}[x=0.75pt,y=0.75pt,yscale=-1,xscale=1]

\draw  [color={rgb, 255:red, 0; green, 0; blue, 0 }  ,draw opacity=1 ][fill={rgb, 255:red, 255; green, 255; blue, 255 }  ,fill opacity=1 ][dash pattern={on 0.84pt off 2.51pt}] (300,130) -- (320,130) -- (320,150) -- (300,150) -- cycle ;
\draw  [color={rgb, 255:red, 0; green, 0; blue, 0 }  ,draw opacity=1 ][fill={rgb, 255:red, 255; green, 255; blue, 255 }  ,fill opacity=1 ][dash pattern={on 0.84pt off 2.51pt}] (300,130) -- (320,130) -- (320,150) -- (300,150) -- cycle ;
\draw  [color={rgb, 255:red, 0; green, 0; blue, 0 }  ,draw opacity=1 ][fill={rgb, 255:red, 255; green, 255; blue, 255 }  ,fill opacity=1 ][dash pattern={on 0.84pt off 2.51pt}] (300,110) -- (320,110) -- (320,130) -- (300,130) -- cycle ;
\draw  [color={rgb, 255:red, 0; green, 0; blue, 0 }  ,draw opacity=1 ][fill={rgb, 255:red, 255; green, 255; blue, 255 }  ,fill opacity=1 ][dash pattern={on 0.84pt off 2.51pt}] (300,150) -- (320,150) -- (320,170) -- (300,170) -- cycle ;
\draw  [color={rgb, 255:red, 255; green, 255; blue, 255 }  ,draw opacity=1 ][fill={rgb, 255:red, 155; green, 155; blue, 155 }  ,fill opacity=0.37 ] (260,110) -- (280,110) -- (280,130) -- (260,130) -- cycle ;
\draw  [color={rgb, 255:red, 255; green, 255; blue, 255 }  ,draw opacity=1 ][fill={rgb, 255:red, 155; green, 155; blue, 155 }  ,fill opacity=0.37 ] (260,130) -- (280,130) -- (280,150) -- (260,150) -- cycle ;
\draw    (270,170) -- (270,190) ;
\draw  [dash pattern={on 0.84pt off 2.51pt}]  (270,190) -- (270,210) ;
\draw  [color={rgb, 255:red, 0; green, 0; blue, 0 }  ,draw opacity=1 ][fill={rgb, 255:red, 255; green, 255; blue, 255 }  ,fill opacity=1 ][dash pattern={on 0.84pt off 2.51pt}] (280,130) -- (300,130) -- (300,150) -- (280,150) -- cycle ;
\draw  [color={rgb, 255:red, 255; green, 255; blue, 255 }  ,draw opacity=1 ][fill={rgb, 255:red, 155; green, 155; blue, 155 }  ,fill opacity=0.37 ] (260,90) -- (280,90) -- (280,110) -- (260,110) -- cycle ;
\draw  [color={rgb, 255:red, 255; green, 255; blue, 255 }  ,draw opacity=1 ][fill={rgb, 255:red, 155; green, 155; blue, 155 }  ,fill opacity=0.37 ] (260,150) -- (280,150) -- (280,170) -- (260,170) -- cycle ;
\draw [color={rgb, 255:red, 208; green, 2; blue, 27 }  ,draw opacity=0.7 ][line width=1.5]  [dash pattern={on 1.69pt off 2.76pt}]  (264,60) -- (264,140) -- (310,140) -- (310,156) ;
\draw [shift={(310,160)}, rotate = 270] [fill={rgb, 255:red, 208; green, 2; blue, 27 }  ,fill opacity=0.7 ][line width=0.08]  [draw opacity=0] (6.97,-3.35) -- (0,0) -- (6.97,3.35) -- cycle    ;
\draw [color={rgb, 255:red, 208; green, 2; blue, 27 }  ,draw opacity=0.7 ][line width=1.5]  [dash pattern={on 1.69pt off 2.76pt}]  (270,60) -- (270,120) -- (310,120) -- (310,136) ;
\draw [shift={(310,140)}, rotate = 270] [fill={rgb, 255:red, 208; green, 2; blue, 27 }  ,fill opacity=0.7 ][line width=0.08]  [draw opacity=0] (6.97,-3.35) -- (0,0) -- (6.97,3.35) -- cycle    ;
\draw [color={rgb, 255:red, 208; green, 2; blue, 27 }  ,draw opacity=0.7 ][line width=1.5]  [dash pattern={on 1.69pt off 2.76pt}]  (276,60) -- (276,100) -- (310,100) -- (310,116) ;
\draw [shift={(310,120)}, rotate = 270] [fill={rgb, 255:red, 208; green, 2; blue, 27 }  ,fill opacity=0.7 ][line width=0.08]  [draw opacity=0] (6.97,-3.35) -- (0,0) -- (6.97,3.35) -- cycle    ;
\draw    (240,160) -- (257,160) ;
\draw [shift={(260,160)}, rotate = 180] [fill={rgb, 255:red, 0; green, 0; blue, 0 }  ][line width=0.08]  [draw opacity=0] (8.93,-4.29) -- (0,0) -- (8.93,4.29) -- cycle    ;
\draw    (340,160) -- (323,160) ;
\draw [shift={(320,160)}, rotate = 360] [fill={rgb, 255:red, 0; green, 0; blue, 0 }  ][line width=0.08]  [draw opacity=0] (8.93,-4.29) -- (0,0) -- (8.93,4.29) -- cycle    ;
\draw    (240,100) -- (257,100) ;
\draw [shift={(260,100)}, rotate = 180] [fill={rgb, 255:red, 0; green, 0; blue, 0 }  ][line width=0.08]  [draw opacity=0] (8.93,-4.29) -- (0,0) -- (8.93,4.29) -- cycle    ;

\draw (262,113.4) node [anchor=north west][inner sep=0.75pt]  [font=\small]  {$\ \ $};
\draw (168,152) node [anchor=north west][inner sep=0.75pt]   [align=left] {{\small Input signal}};
\draw (344,142) node [anchor=north west][inner sep=0.75pt]   [align=left] {{\small Positions for}\\{\small discarded particles}};
\draw (157,91) node [anchor=north west][inner sep=0.75pt]   [align=left] {{\small Output signal}};

\end{tikzpicture}
	\label{fig:wire} 
	\caption{Gadget for the simulation of wires. A series of particles are thrown from the column aligned with the input position, following the trajectories depicted with red pointed arrows. }
	\end{center}     
\end{figure}

\begin{figure}[h]
	\begin{center}
	\includestandalone[scale=0.8]{wire_dyn}
	\label{fig:wire_dyn} 
	\caption{Dynamics of the wire when the input is \True~(top) and when the input is \False~(bottom). In the bottom case, the discarded particles are not necessarily stuck in that position and eventually fall further down.}
	\end{center}     
\end{figure}

We must discard the particles transmitting the value {\False} this way because we cannot make use of the upwards direction to do so (through the top of the lattice, as was the case for the original 4-directional construction). This means, that when we finally put the circuit together, each wire must be isolated by a distance of at least four columns from the next, to permit discarding particles. This separation will later guarantee that the discarded particles do not interfere with the evaluation of the circuit: transmitting a {\False} value of length $n$ corresponds to discard a stack of particles, will only reach a height of $n-1$. In a similar way, this construction can be adapted to allow wires that transmit information that turn a signal in a right or left direction, and to multiply signal as well. 


We now explain the simulation of the NOR gate. It receives two inputs from the preceding layer. Each of these inputs are grown as mentioned before to the sites $a$ and $b$, as shown in Figure~\ref{fig:nor}.  It is important to remember that both input cables are separated by a distance of four columns, to allow for the discarding  of particles. In addition, we grow a \emph{power cable} to function the gate (it is a wire that always carries the {\True} value). Same as before, site $c$ must be at least four columns away from site $b$. Once everything is in place, the gate is evaluated as follows (see Figure~\ref{fig:nor_dyn}). : A particle makes the journey $a\rightarrow b\rightarrow c$. If input 1 is {\True}, then the particle will stick at $a$. The same goes for input 2 and site $b$. If both inputs are {\False}, the particle will then stick at site $c$. In any of the 3 cases, a new wire is grown starting from site $d$. This results in the simulation of a NOR gate.

\begin{figure}[h]
	\begin{center}
	\tikzset{every picture/.style={line width=0.75pt}} 

\begin{tikzpicture}[x=0.75pt,y=0.75pt,yscale=-1,xscale=1]

\draw    (480,30) -- (480,40) ;
\draw  [color={rgb, 255:red, 255; green, 255; blue, 255 }  ,draw opacity=1 ][fill={rgb, 255:red, 0; green, 0; blue, 0 }  ,fill opacity=1 ] (470,100) -- (490,100) -- (490,120) -- (470,120) -- cycle ;
\draw  [color={rgb, 255:red, 255; green, 255; blue, 255 }  ,draw opacity=1 ][fill={rgb, 255:red, 0; green, 0; blue, 0 }  ,fill opacity=1 ] (490,100) -- (510,100) -- (510,120) -- (490,120) -- cycle ;
\draw  [color={rgb, 255:red, 255; green, 255; blue, 255 }  ,draw opacity=1 ][fill={rgb, 255:red, 0; green, 0; blue, 0 }  ,fill opacity=1 ] (510,100) -- (530,100) -- (530,120) -- (510,120) -- cycle ;
\draw    (530,110) -- (550,110) ;
\draw  [dash pattern={on 0.84pt off 2.51pt}]  (550,110) -- (570,110) ;
\draw  [dash pattern={on 0.84pt off 2.51pt}]  (480,20) -- (480,30) ;
\draw  [color={rgb, 255:red, 0; green, 0; blue, 0 }  ,draw opacity=1 ][fill={rgb, 255:red, 255; green, 255; blue, 255 }  ,fill opacity=1 ][dash pattern={on 0.84pt off 2.51pt}] (430,80) -- (450,80) -- (450,100) -- (430,100) -- cycle ;
\draw  [color={rgb, 255:red, 0; green, 0; blue, 0 }  ,draw opacity=1 ][fill={rgb, 255:red, 255; green, 255; blue, 255 }  ,fill opacity=1 ][dash pattern={on 0.84pt off 2.51pt}] (410,80) -- (430,80) -- (430,100) -- (410,100) -- cycle ;
\draw  [color={rgb, 255:red, 255; green, 255; blue, 255 }  ,draw opacity=1 ][fill={rgb, 255:red, 155; green, 155; blue, 155 }  ,fill opacity=0.37 ] (370,80) -- (390,80) -- (390,100) -- (370,100) -- cycle ;
\draw  [color={rgb, 255:red, 255; green, 255; blue, 255 }  ,draw opacity=1 ][fill={rgb, 255:red, 155; green, 155; blue, 155 }  ,fill opacity=0.37 ] (370,100) -- (390,100) -- (390,120) -- (370,120) -- cycle ;
\draw  [color={rgb, 255:red, 255; green, 255; blue, 255 }  ,draw opacity=1 ][fill={rgb, 255:red, 155; green, 155; blue, 155 }  ,fill opacity=0.37 ] (370,120) -- (390,120) -- (390,140) -- (370,140) -- cycle ;
\draw  [color={rgb, 255:red, 255; green, 255; blue, 255 }  ,draw opacity=1 ][fill={rgb, 255:red, 155; green, 155; blue, 155 }  ,fill opacity=0.37 ] (370,140) -- (390,140) -- (390,160) -- (370,160) -- cycle ;
\draw    (380,160) -- (380,170) ;
\draw  [dash pattern={on 0.84pt off 2.51pt}]  (380,170) -- (380,180) ;
\draw  [color={rgb, 255:red, 0; green, 0; blue, 0 }  ,draw opacity=1 ][fill={rgb, 255:red, 255; green, 255; blue, 255 }  ,fill opacity=1 ][dash pattern={on 0.84pt off 2.51pt}] (350,80) -- (370,80) -- (370,100) -- (350,100) -- cycle ;
\draw  [color={rgb, 255:red, 0; green, 0; blue, 0 }  ,draw opacity=1 ][fill={rgb, 255:red, 255; green, 255; blue, 255 }  ,fill opacity=1 ][dash pattern={on 0.84pt off 2.51pt}] (330,80) -- (350,80) -- (350,100) -- (330,100) -- cycle ;
\draw  [color={rgb, 255:red, 0; green, 0; blue, 0 }  ,draw opacity=1 ][fill={rgb, 255:red, 255; green, 255; blue, 255 }  ,fill opacity=1 ][dash pattern={on 0.84pt off 2.51pt}] (310,80) -- (330,80) -- (330,100) -- (310,100) -- cycle ;
\draw  [color={rgb, 255:red, 255; green, 255; blue, 255 }  ,draw opacity=1 ][fill={rgb, 255:red, 155; green, 155; blue, 155 }  ,fill opacity=0.37 ] (270,80) -- (290,80) -- (290,100) -- (270,100) -- cycle ;
\draw  [color={rgb, 255:red, 255; green, 255; blue, 255 }  ,draw opacity=1 ][fill={rgb, 255:red, 155; green, 155; blue, 155 }  ,fill opacity=0.37 ] (270,100) -- (290,100) -- (290,120) -- (270,120) -- cycle ;
\draw  [color={rgb, 255:red, 255; green, 255; blue, 255 }  ,draw opacity=1 ][fill={rgb, 255:red, 155; green, 155; blue, 155 }  ,fill opacity=0.37 ] (270,120) -- (290,120) -- (290,140) -- (270,140) -- cycle ;
\draw  [color={rgb, 255:red, 255; green, 255; blue, 255 }  ,draw opacity=1 ][fill={rgb, 255:red, 155; green, 155; blue, 155 }  ,fill opacity=0.37 ] (270,140) -- (290,140) -- (290,160) -- (270,160) -- cycle ;
\draw    (280,160) -- (280,170) ;
\draw  [dash pattern={on 0.84pt off 2.51pt}]  (280,170) -- (280,180) ;
\draw  [color={rgb, 255:red, 255; green, 255; blue, 255 }  ,draw opacity=1 ][fill={rgb, 255:red, 155; green, 155; blue, 155 }  ,fill opacity=0.37 ] (470,80) -- (490,80) -- (490,100) -- (470,100) -- cycle ;
\draw  [color={rgb, 255:red, 255; green, 255; blue, 255 }  ,draw opacity=1 ][fill={rgb, 255:red, 155; green, 155; blue, 155 }  ,fill opacity=0.37 ] (470,60) -- (490,60) -- (490,80) -- (470,80) -- cycle ;
\draw  [color={rgb, 255:red, 0; green, 0; blue, 0 }  ,draw opacity=1 ][fill={rgb, 255:red, 255; green, 255; blue, 255 }  ,fill opacity=1 ][dash pattern={on 0.84pt off 2.51pt}] (290,80) -- (310,80) -- (310,100) -- (290,100) -- cycle ;
\draw  [color={rgb, 255:red, 0; green, 0; blue, 0 }  ,draw opacity=1 ][fill={rgb, 255:red, 255; green, 255; blue, 255 }  ,fill opacity=1 ][dash pattern={on 0.84pt off 2.51pt}] (450,80) -- (470,80) -- (470,100) -- (450,100) -- cycle ;
\draw  [color={rgb, 255:red, 0; green, 0; blue, 0 }  ,draw opacity=1 ][fill={rgb, 255:red, 255; green, 255; blue, 255 }  ,fill opacity=1 ][dash pattern={on 0.84pt off 2.51pt}] (390,80) -- (410,80) -- (410,100) -- (390,100) -- cycle ;
\draw  [color={rgb, 255:red, 255; green, 255; blue, 255 }  ,draw opacity=1 ][fill={rgb, 255:red, 155; green, 155; blue, 155 }  ,fill opacity=0.37 ] (470,40) -- (490,40) -- (490,60) -- (470,60) -- cycle ;

\draw (372,83.4) node [anchor=north west][inner sep=0.75pt]  [font=\small]  {$\ b\ $};
\draw (272,83.4) node [anchor=north west][inner sep=0.75pt]  [font=\small]  {$\ a\ $};
\draw (472,83.4) node [anchor=north west][inner sep=0.75pt]  [font=\small]  {$\ c\ $};
\draw (472,63.4) node [anchor=north west][inner sep=0.75pt]  [font=\small]  {$\ d\ $};
\draw (261,182) node [anchor=north west][inner sep=0.75pt]   [align=left] {{\small Input 1}};
\draw (361,183) node [anchor=north west][inner sep=0.75pt]   [align=left] {{\small Input 2}};
\draw (529,122) node [anchor=north west][inner sep=0.75pt]   [align=left] {{\small Power}};
\draw (426,21) node [anchor=north west][inner sep=0.75pt]   [align=left] {{\small Output}};

\end{tikzpicture}
	\label{fig:nor} 
	\caption{Gadget for the simulation of a NOR gate. Both inputs are grown until they reach the sites directly below sites $a$ and $b$ respectively. Two successive particles then follow the path $a\rightarrow b\rightarrow c$ following the dotted line, stopping according to the inputs. The output is then grown from site $c$.}
	\end{center}     
\end{figure}

\begin{figure}[h]
	\begin{center}
	\includestandalone[scale=0.6]{NOR_dinamica}
	\label{fig:nor_dyn} 
	\caption{Examples of the dynamics of NOR gate when one of the two inputs is \True~(first two lines), and when both inputs are \False~(bottom). In the three lines in the right column is depicted in a red pointed line the trajectory of the particle that follows  $a\rightarrow b\rightarrow c$. Notice that in the first two lines the particle is fixed before reaching position $c$, while in the third case the particle reaches position $c$. In the right column is depicted the trajectory of the wire that grows from $d$, which is only fixed in that position in the third case.}
	\end{center}     
\end{figure}

The power cable is simply a path of occupied cells, growing from the rightmost part of the circuit. Here we find a second difference in our construction with respect the construction of \cite{machta1996computational}. In the latter, the power cable snakes around the configuration. Starting from the first layer, the power cable grows from right to left,  then it grows through the second layer, and continues to grow left to right, and so on. The gates are constructed following the power cable. In our case, the gates also grow following the power cable. The difference is that we grow a new branch of the power cable for each layer of the circuit. The reason for this change is given by the restriction of the movement directions of the particles. Indeed, in our case, the gates of the same layer are evaluated always from right to left. 

Once a NOR gate is evaluated, the power cable continues to grow in order to allow the evaluation of the next gate. In Figure \ref{fig:NOR_power} it is shown how to grow the power cable in order to cross the information over the NOR gate, independently of its evaluation.

\begin{figure}[h]
	\begin{center}
	\includestandalone[scale=0.7]{NOR_power}
	\label{fig:NOR_power} 
	\caption{Growing the power cable after the evaluation of the NOR gate. One particle is thrown on every column in the order given by the numbers in the corresponding positions. The trajectories are straightforward, except for particles reaching positions 12, 13 and 14, which require a turn. As necessarily after the evaluation of the NOR gate one of $a$, $b$ or $d$ is occupied, after throwing this 15 particles, a particle is fixed in the last position.  }
	\end{center}     
\end{figure}

A given circuit might have gates with outputs in different layers, implying that in the simulation, we may found wires that cross a layer, possibly crossing power cables. A second gate is defined to cope with this problem, which we call \emph{single input OR gate}, following the nomenclature of \cite{machta1996computational}. The construction will be described with the help of Figure \ref{fig:or}. The input wire is grown up to site $p$. There, two particles make specific trajectories: the first visits $a\rightarrow b$, while the second one $c\rightarrow d$. After these particles have completed their trajectories, the power cable is grown starting at the site left of $d$. If the input is true, then the first of the walks will stop at $a$, and the second at $c$. The wire is grown from $p$, and the power cable as mentioned before, crossing the two of them. If the input is negative, then the walks will end at $b$ and $d$ respectively. As before, the wire and the power cable are grown, crossing them. We re-state the importance of the distance between the cables germinated at $p$, $c$ and $a$, for the discarding of particles.

\begin{figure}[H]
	\begin{center}
	    \tikzset{every picture/.style={line width=0.75pt}} 

\begin{tikzpicture}[x=0.75pt,y=0.75pt,yscale=-1,xscale=1]

\draw  [color={rgb, 255:red, 255; green, 255; blue, 255 }  ,draw opacity=1 ][fill={rgb, 255:red, 155; green, 155; blue, 155 }  ,fill opacity=0.37 ] (351,90) -- (371,90) -- (371,110) -- (351,110) -- cycle ;
\draw  [color={rgb, 255:red, 255; green, 255; blue, 255 }  ,draw opacity=1 ][fill={rgb, 255:red, 155; green, 155; blue, 155 }  ,fill opacity=0.37 ] (371,90) -- (391,90) -- (391,110) -- (371,110) -- cycle ;
\draw  [color={rgb, 255:red, 255; green, 255; blue, 255 }  ,draw opacity=1 ][fill={rgb, 255:red, 155; green, 155; blue, 155 }  ,fill opacity=0.37 ] (391,90) -- (411,90) -- (411,110) -- (391,110) -- cycle ;
\draw  [color={rgb, 255:red, 255; green, 255; blue, 255 }  ,draw opacity=1 ][fill={rgb, 255:red, 155; green, 155; blue, 155 }  ,fill opacity=0.37 ] (431,70) -- (451,70) -- (451,90) -- (431,90) -- cycle ;
\draw  [color={rgb, 255:red, 255; green, 255; blue, 255 }  ,draw opacity=1 ][fill={rgb, 255:red, 155; green, 155; blue, 155 }  ,fill opacity=0.37 ] (411,90) -- (431,90) -- (431,110) -- (411,110) -- cycle ;
\draw  [color={rgb, 255:red, 255; green, 255; blue, 255 }  ,draw opacity=1 ][fill={rgb, 255:red, 155; green, 155; blue, 155 }  ,fill opacity=0.37 ] (411,70) -- (431,70) -- (431,90) -- (411,90) -- cycle ;
\draw  [color={rgb, 255:red, 255; green, 255; blue, 255 }  ,draw opacity=1 ][fill={rgb, 255:red, 155; green, 155; blue, 155 }  ,fill opacity=0.37 ] (391,70) -- (411,70) -- (411,90) -- (391,90) -- cycle ;
\draw  [color={rgb, 255:red, 255; green, 255; blue, 255 }  ,draw opacity=1 ][fill={rgb, 255:red, 155; green, 155; blue, 155 }  ,fill opacity=0.37 ] (431,90) -- (451,90) -- (451,110) -- (431,110) -- cycle ;
\draw  [color={rgb, 255:red, 255; green, 255; blue, 255 }  ,draw opacity=1 ][fill={rgb, 255:red, 0; green, 0; blue, 0 }  ,fill opacity=1 ] (512,90) -- (532,90) -- (532,110) -- (512,110) -- cycle ;
\draw  [color={rgb, 255:red, 255; green, 255; blue, 255 }  ,draw opacity=1 ][fill={rgb, 255:red, 155; green, 155; blue, 155 }  ,fill opacity=0.37 ] (411,110) -- (431,110) -- (431,130) -- (411,130) -- cycle ;
\draw  [color={rgb, 255:red, 255; green, 255; blue, 255 }  ,draw opacity=1 ][fill={rgb, 255:red, 0; green, 0; blue, 0 }  ,fill opacity=1 ] (491,90) -- (511,90) -- (511,110) -- (491,110) -- cycle ;
\draw    (532,100) -- (541,100) ;
\draw    (420,130) -- (420,140) ;
\draw  [dash pattern={on 0.84pt off 2.51pt}]  (542,100) -- (552,100) ;
\draw  [dash pattern={on 0.84pt off 2.51pt}]  (420,140) -- (420,150) ;
\draw    (322,100) -- (332,100) ;
\draw  [dash pattern={on 0.84pt off 2.51pt}]  (312,100) -- (322,100) ;
\draw    (420,60) -- (420,70) ;
\draw  [dash pattern={on 0.84pt off 2.51pt}]  (420,50) -- (420,60) ;
\draw  [color={rgb, 255:red, 255; green, 255; blue, 255 }  ,draw opacity=1 ][fill={rgb, 255:red, 155; green, 155; blue, 155 }  ,fill opacity=0.37 ] (451,90) -- (471,90) -- (471,110) -- (451,110) -- cycle ;
\draw  [color={rgb, 255:red, 255; green, 255; blue, 255 }  ,draw opacity=1 ][fill={rgb, 255:red, 155; green, 155; blue, 155 }  ,fill opacity=0.37 ] (332,90) -- (352,90) -- (352,110) -- (332,110) -- cycle ;
\draw  [color={rgb, 255:red, 255; green, 255; blue, 255 }  ,draw opacity=1 ][fill={rgb, 255:red, 155; green, 155; blue, 155 }  ,fill opacity=0.37 ] (471,90) -- (491,90) -- (491,110) -- (471,110) -- cycle ;

\draw (432,93.4) node [anchor=north west][inner sep=0.75pt]  [font=\small]  {$\ b\ $};
\draw (433,73.4) node [anchor=north west][inner sep=0.75pt]  [font=\small]  {$\ a\ $};
\draw (412,72.4) node [anchor=north west][inner sep=0.75pt]  [font=\small]  {$\ p\ $};
\draw (393,73.4) node [anchor=north west][inner sep=0.75pt]  [font=\small]  {$\ c\ $};
\draw (413,93.4) node [anchor=north west][inner sep=0.75pt]  [font=\small]  {$\ d\ $};
\draw (424,141) node [anchor=north west][inner sep=0.75pt]   [align=left] {{\small Signal In}};
\draw (501,111) node [anchor=north west][inner sep=0.75pt]   [align=left] {{\small Power In}};
\draw (303,111) node [anchor=north west][inner sep=0.75pt]   [align=left] {{\small Power Out}};
\draw (424,41) node [anchor=north west][inner sep=0.75pt]   [align=left] {{\small Signal Out}};

\end{tikzpicture}
	\label{fig:or} 
	\caption{Gadget for the simulation of a OR gate. The input is grown up to site $p$. A first particle then follows the trajectory $a\rightarrow b$, stopping according to the truth value of the gadget's input. Then, a second particle makes the trajectory $c\rightarrow d$, also stopping according to the truth value of the input. To finalize, the power wire is grown starting from site $d$, and the output is grown from site $p$. This OR gate with a single input, in fact, simulates the crossing of the input wire with the power wire. }
	\end{center}     
\end{figure}

\begin{figure}[H]
	\begin{center}
	\includestandalone[scale=0.6]{OR_F}

	\label{fig:orV} 
	\caption{Dynamic associated to the OR gate when the input is \True~(top) and \False~(bottom).}
	\end{center}     
\end{figure}

In the following, we give two examples of simple circuits in order to illustrate the resulting reduction. 

\medskip
{\bf Example.}
Let us look at a brief example of how the simulation works. Let us take a circuit consisting of a single NOR gate with inputs labeled $x$ and $y$. In order to simplify the description, let us represent each particle in our simulation by a tuple $p\in \{1, ..., N\}\times\{D, R\}^*$, where the first coordinate represents the column into which the particle is released and the second is a word that codes the trajectory of the particle with $D$ representing a downwards movement, and $R$ a movement to the right. Following the description of the simulation of a NOR gate given above, the \emph{skeleton} of the construction of this gate is given in Figure~\ref{fig:skeleton}.

\begin{figure}[h]
\begin{center}
    \tikzset{every picture/.style={line width=0.75pt}} 

\begin{tikzpicture}[x=0.75pt,y=0.75pt,yscale=-1,xscale=1]

\draw  [color={rgb, 255:red, 0; green, 0; blue, 0 }  ,draw opacity=1 ][fill={rgb, 255:red, 255; green, 255; blue, 255 }  ,fill opacity=1 ][dash pattern={on 0.84pt off 2.51pt}] (270,150) -- (290,150) -- (290,170) -- (270,170) -- cycle ;
\draw  [color={rgb, 255:red, 0; green, 0; blue, 0 }  ,draw opacity=1 ][fill={rgb, 255:red, 255; green, 255; blue, 255 }  ,fill opacity=1 ][dash pattern={on 0.84pt off 2.51pt}] (270,170) -- (290,170) -- (290,190) -- (270,190) -- cycle ;
\draw  [color={rgb, 255:red, 0; green, 0; blue, 0 }  ,draw opacity=1 ][fill={rgb, 255:red, 255; green, 255; blue, 255 }  ,fill opacity=1 ][dash pattern={on 0.84pt off 2.51pt}] (270,190) -- (290,190) -- (290,210) -- (270,210) -- cycle ;
\draw  [color={rgb, 255:red, 0; green, 0; blue, 0 }  ,draw opacity=1 ][fill={rgb, 255:red, 255; green, 255; blue, 255 }  ,fill opacity=1 ][dash pattern={on 0.84pt off 2.51pt}] (270,230) -- (290,230) -- (290,250) -- (270,250) -- cycle ;
\draw  [color={rgb, 255:red, 0; green, 0; blue, 0 }  ,draw opacity=1 ][fill={rgb, 255:red, 255; green, 255; blue, 255 }  ,fill opacity=1 ][dash pattern={on 0.84pt off 2.51pt}] (270,210) -- (290,210) -- (290,230) -- (270,230) -- cycle ;
\draw  [color={rgb, 255:red, 0; green, 0; blue, 0 }  ,draw opacity=1 ][fill={rgb, 255:red, 255; green, 255; blue, 255 }  ,fill opacity=1 ][dash pattern={on 0.84pt off 2.51pt}] (370,150) -- (390,150) -- (390,170) -- (370,170) -- cycle ;
\draw  [color={rgb, 255:red, 0; green, 0; blue, 0 }  ,draw opacity=1 ][fill={rgb, 255:red, 255; green, 255; blue, 255 }  ,fill opacity=1 ][dash pattern={on 0.84pt off 2.51pt}] (370,170) -- (390,170) -- (390,190) -- (370,190) -- cycle ;
\draw  [color={rgb, 255:red, 0; green, 0; blue, 0 }  ,draw opacity=1 ][fill={rgb, 255:red, 255; green, 255; blue, 255 }  ,fill opacity=1 ][dash pattern={on 0.84pt off 2.51pt}] (370,190) -- (390,190) -- (390,210) -- (370,210) -- cycle ;
\draw  [color={rgb, 255:red, 0; green, 0; blue, 0 }  ,draw opacity=1 ][fill={rgb, 255:red, 255; green, 255; blue, 255 }  ,fill opacity=1 ][dash pattern={on 0.84pt off 2.51pt}] (370,230) -- (390,230) -- (390,250) -- (370,250) -- cycle ;
\draw  [color={rgb, 255:red, 0; green, 0; blue, 0 }  ,draw opacity=1 ][fill={rgb, 255:red, 255; green, 255; blue, 255 }  ,fill opacity=1 ][dash pattern={on 0.84pt off 2.51pt}] (370,210) -- (390,210) -- (390,230) -- (370,230) -- cycle ;
\draw  [color={rgb, 255:red, 0; green, 0; blue, 0 }  ,draw opacity=1 ][fill={rgb, 255:red, 255; green, 255; blue, 255 }  ,fill opacity=1 ][dash pattern={on 0.84pt off 2.51pt}] (510,150) -- (530,150) -- (530,170) -- (510,170) -- cycle ;
\draw  [color={rgb, 255:red, 0; green, 0; blue, 0 }  ,draw opacity=1 ][fill={rgb, 255:red, 255; green, 255; blue, 255 }  ,fill opacity=1 ][dash pattern={on 0.84pt off 2.51pt}] (470,150) -- (490,150) -- (490,170) -- (470,170) -- cycle ;
\draw  [color={rgb, 255:red, 0; green, 0; blue, 0 }  ,draw opacity=1 ][fill={rgb, 255:red, 255; green, 255; blue, 255 }  ,fill opacity=1 ][dash pattern={on 0.84pt off 2.51pt}] (510,170) -- (530,170) -- (530,190) -- (510,190) -- cycle ;
\draw  [color={rgb, 255:red, 0; green, 0; blue, 0 }  ,draw opacity=1 ][fill={rgb, 255:red, 255; green, 255; blue, 255 }  ,fill opacity=1 ][dash pattern={on 0.84pt off 2.51pt}] (510,190) -- (530,190) -- (530,210) -- (510,210) -- cycle ;
\draw  [color={rgb, 255:red, 0; green, 0; blue, 0 }  ,draw opacity=1 ][fill={rgb, 255:red, 255; green, 255; blue, 255 }  ,fill opacity=1 ][dash pattern={on 0.84pt off 2.51pt}] (510,230) -- (530,230) -- (530,250) -- (510,250) -- cycle ;
\draw  [color={rgb, 255:red, 0; green, 0; blue, 0 }  ,draw opacity=1 ][fill={rgb, 255:red, 255; green, 255; blue, 255 }  ,fill opacity=1 ][dash pattern={on 0.84pt off 2.51pt}] (510,210) -- (530,210) -- (530,230) -- (510,230) -- cycle ;
\draw  [color={rgb, 255:red, 0; green, 0; blue, 0 }  ,draw opacity=1 ][fill={rgb, 255:red, 255; green, 255; blue, 255 }  ,fill opacity=1 ][dash pattern={on 0.84pt off 2.51pt}] (490,150) -- (510,150) -- (510,170) -- (490,170) -- cycle ;
\draw  [color={rgb, 255:red, 0; green, 0; blue, 0 }  ,draw opacity=1 ][fill={rgb, 255:red, 255; green, 255; blue, 255 }  ,fill opacity=1 ][dash pattern={on 0.84pt off 2.51pt}] (470,110) -- (490,110) -- (490,130) -- (470,130) -- cycle ;
\draw  [color={rgb, 255:red, 0; green, 0; blue, 0 }  ,draw opacity=1 ][fill={rgb, 255:red, 255; green, 255; blue, 255 }  ,fill opacity=1 ][dash pattern={on 0.84pt off 2.51pt}] (470,90) -- (490,90) -- (490,110) -- (470,110) -- cycle ;
\draw  [color={rgb, 255:red, 0; green, 0; blue, 0 }  ,draw opacity=1 ][fill={rgb, 255:red, 255; green, 255; blue, 255 }  ,fill opacity=1 ][dash pattern={on 0.84pt off 2.51pt}] (470,130) -- (490,130) -- (490,150) -- (470,150) -- cycle ;
\draw  [color={rgb, 255:red, 0; green, 0; blue, 0 }  ,draw opacity=1 ][fill={rgb, 255:red, 255; green, 255; blue, 255 }  ,fill opacity=1 ][dash pattern={on 0.84pt off 2.51pt}] (470,70) -- (490,70) -- (490,90) -- (470,90) -- cycle ;
\draw    (230,250) -- (570,250) ;
\draw    (270,245) -- (270,255) ;
\draw    (290,245) -- (290,255) ;
\draw    (310,245) -- (310,255) ;
\draw    (330,245) -- (330,255) ;
\draw    (350,245) -- (350,255) ;
\draw    (230,240) -- (230,250) ;
\draw    (250,245) -- (250,255) ;
\draw    (430,245) -- (430,255) ;
\draw    (450,245) -- (450,255) ;
\draw    (470,245) -- (470,255) ;
\draw    (490,245) -- (490,255) ;
\draw    (510,245) -- (510,255) ;
\draw    (370,245) -- (370,255) ;
\draw    (390,245) -- (390,255) ;
\draw    (410,245) -- (410,255) ;
\draw    (530,245) -- (530,255) ;
\draw    (550,245) -- (550,255) ;
\draw    (570,245) -- (570,255) ;
\draw  [dash pattern={on 4.5pt off 4.5pt}] (250,90) -- (560,90) -- (560,170) -- (250,170) -- cycle ;
\draw  [color={rgb, 255:red, 0; green, 0; blue, 0 }  ,draw opacity=1 ][fill={rgb, 255:red, 255; green, 255; blue, 255 }  ,fill opacity=1 ][dash pattern={on 0.84pt off 2.51pt}] (370,130) -- (390,130) -- (390,150) -- (370,150) -- cycle ;
\draw  [color={rgb, 255:red, 0; green, 0; blue, 0 }  ,draw opacity=1 ][fill={rgb, 255:red, 255; green, 255; blue, 255 }  ,fill opacity=1 ][dash pattern={on 0.84pt off 2.51pt}] (270,130) -- (290,130) -- (290,150) -- (270,150) -- cycle ;

\draw (274,252.4) node [anchor=north west][inner sep=0.75pt]    {$x$};
\draw (375,252.4) node [anchor=north west][inner sep=0.75pt]    {$y$};
\draw (498,254) node [anchor=north west][inner sep=0.75pt]   [align=left] {Power};
\draw (295,80) node  [font=\small] [align=left] {\begin{minipage}[lt]{61.2pt}\setlength\topsep{0pt}
NOR gate
\end{minipage}};
\draw (461,52) node [anchor=north west][inner sep=0.75pt]   [align=left] {Output};

\end{tikzpicture}
    \caption{Skeleton of the construction of a circuit consisting in only one NOR gate.}
    \label{fig:skeleton}
    \end{center}
\end{figure}

Observe that everything is contained in a 19 by 19 grid. We begin by initializing the truth values of the variables $x$ and $y$. If $x = \True$, we add the particles $(5, D^{19})$ and $(5, D^{19})$, else we add nothing. Similarly, if $y = \True$ we add the particles $(10, D^{19})$ and $(10, D^{19})$, else we add nothing. 

For the transmission of the truth values through wires we assign a position in the wire to each particle. For the position $t = (i,j)$ in the lattice, we add the particle $p_t$ given by: $$p_t = (\underbrace{i, D^{j}}_{\text{Position}}\circ\underbrace{R^2D^{N-j}}_{\text{Discard}}).$$
For example,  if the variable $x$  has truth value \True, we add particles for every position $t_j = (5, 18-j)$ with $j\in\{1, 2, 3, 4\}$. As we can see in Figure~\ref{fig:x1yx0} , $p_{t_1}$ will come down up to position $t_1$ and stick because of the particle at $(5,18)$. Let us analyze the case when $x = \False$. In this case we add the same particles corresponding to positions $t_j$, but when they come down, they have nothing to attach to, so they are discarded, as shown in Figure~\ref{fig:x1yx0}.
\begin{figure}[h]
\begin{center}
    \tikzset{every picture/.style={line width=0.75pt}} 

\begin{tikzpicture}[x=0.75pt,y=0.75pt,yscale=-1,xscale=1]

\draw  [color={rgb, 255:red, 255; green, 255; blue, 255 }  ,draw opacity=1 ][fill={rgb, 255:red, 0; green, 0; blue, 0 }  ,fill opacity=1 ] (200,280) -- (220,280) -- (220,300) -- (200,300) -- cycle ;
\draw  [color={rgb, 255:red, 255; green, 255; blue, 255 }  ,draw opacity=1 ][fill={rgb, 255:red, 0; green, 0; blue, 0 }  ,fill opacity=1 ] (200,260) -- (220,260) -- (220,280) -- (200,280) -- cycle ;
\draw    (100,300) -- (309.5,300) (120,294.5) -- (120,305.5)(140,294.5) -- (140,305.5)(160,294.5) -- (160,305.5)(180,294.5) -- (180,305.5)(200,294.5) -- (200,305.5)(220,294.5) -- (220,305.5)(240,294.5) -- (240,305.5)(260,294.5) -- (260,305.5)(280,294.5) -- (280,305.5)(300,294.5) -- (300,305.5) ;
\draw  [draw opacity=0][fill={rgb, 255:red, 255; green, 255; blue, 255 }  ,fill opacity=1 ] (300.67,290) -- (320.51,290) -- (320.51,310) -- (300.67,310) -- cycle ;
\draw  [draw opacity=0][fill={rgb, 255:red, 255; green, 255; blue, 255 }  ,fill opacity=1 ] (99.67,291.33) -- (119.5,291.33) -- (119.5,311.33) -- (99.67,311.33) -- cycle ;
\draw  [color={rgb, 255:red, 0; green, 0; blue, 0 }  ,draw opacity=1 ][fill={rgb, 255:red, 255; green, 255; blue, 255 }  ,fill opacity=1 ][dash pattern={on 0.84pt off 2.51pt}] (200,240) -- (220,240) -- (220,260) -- (200,260) -- cycle ;
\draw  [color={rgb, 255:red, 0; green, 0; blue, 0 }  ,draw opacity=1 ][fill={rgb, 255:red, 255; green, 255; blue, 255 }  ,fill opacity=1 ][dash pattern={on 0.84pt off 2.51pt}] (200,220) -- (220,220) -- (220,240) -- (200,240) -- cycle ;
\draw  [color={rgb, 255:red, 0; green, 0; blue, 0 }  ,draw opacity=1 ][fill={rgb, 255:red, 255; green, 255; blue, 255 }  ,fill opacity=1 ][dash pattern={on 0.84pt off 2.51pt}] (200,200) -- (220,200) -- (220,220) -- (200,220) -- cycle ;
\draw  [color={rgb, 255:red, 0; green, 0; blue, 0 }  ,draw opacity=1 ][fill={rgb, 255:red, 255; green, 255; blue, 255 }  ,fill opacity=1 ][dash pattern={on 0.84pt off 2.51pt}] (200,180) -- (220,180) -- (220,200) -- (200,200) -- cycle ;

\draw (210,310) node    {$x$};
\draw (210,250) node  [font=\small]  {$t_{1}$};
\draw (210,230) node  [font=\small]  {$t_{2}$};
\draw (210,210) node  [font=\small]  {$t_{3}$};
\draw (210,190) node  [font=\small]  {$t_{4}$};

\end{tikzpicture}~\tikzset{every picture/.style={line width=0.75pt}} 

\begin{tikzpicture}[x=0.75pt,y=0.75pt,yscale=-1,xscale=1]

\draw    (100,300) -- (309.5,300) (120,294.5) -- (120,305.5)(140,294.5) -- (140,305.5)(160,294.5) -- (160,305.5)(180,294.5) -- (180,305.5)(200,294.5) -- (200,305.5)(220,294.5) -- (220,305.5)(240,294.5) -- (240,305.5)(260,294.5) -- (260,305.5)(280,294.5) -- (280,305.5)(300,294.5) -- (300,305.5) ;
\draw  [draw opacity=0][fill={rgb, 255:red, 255; green, 255; blue, 255 }  ,fill opacity=1 ] (300.67,290) -- (320.51,290) -- (320.51,310) -- (300.67,310) -- cycle ;
\draw  [draw opacity=0][fill={rgb, 255:red, 255; green, 255; blue, 255 }  ,fill opacity=1 ] (99.67,291.33) -- (119.5,291.33) -- (119.5,311.33) -- (99.67,311.33) -- cycle ;
\draw  [color={rgb, 255:red, 0; green, 0; blue, 0 }  ,draw opacity=1 ][fill={rgb, 255:red, 255; green, 255; blue, 255 }  ,fill opacity=1 ][dash pattern={on 0.84pt off 2.51pt}] (200,240) -- (220,240) -- (220,260) -- (200,260) -- cycle ;
\draw  [color={rgb, 255:red, 0; green, 0; blue, 0 }  ,draw opacity=1 ][fill={rgb, 255:red, 255; green, 255; blue, 255 }  ,fill opacity=1 ][dash pattern={on 0.84pt off 2.51pt}] (200,220) -- (220,220) -- (220,240) -- (200,240) -- cycle ;
\draw  [color={rgb, 255:red, 0; green, 0; blue, 0 }  ,draw opacity=1 ][fill={rgb, 255:red, 255; green, 255; blue, 255 }  ,fill opacity=1 ][dash pattern={on 0.84pt off 2.51pt}] (200,200) -- (220,200) -- (220,220) -- (200,220) -- cycle ;
\draw  [color={rgb, 255:red, 0; green, 0; blue, 0 }  ,draw opacity=1 ][fill={rgb, 255:red, 255; green, 255; blue, 255 }  ,fill opacity=1 ][dash pattern={on 0.84pt off 2.51pt}] (200,180) -- (220,180) -- (220,200) -- (200,200) -- cycle ;
\draw    (220,250) -- (250,250) ;
\draw    (250,250) -- (250,276.75) ;
\draw [shift={(250,279.75)}, rotate = 270] [fill={rgb, 255:red, 0; green, 0; blue, 0 }  ][line width=0.08]  [draw opacity=0] (8.93,-4.29) -- (0,0) -- (8.93,4.29) -- cycle    ;

\draw (210,310) node    {$x$};
\draw (210,250) node  [font=\small]  {$t_{1}$};
\draw (210,230) node  [font=\small]  {$t_{2}$};
\draw (210,210) node  [font=\small]  {$t_{3}$};
\draw (210,190) node  [font=\small]  {$t_{4}$};

\end{tikzpicture}
\end{center}
\caption{Evaluation of the wire transmitting the value of gate $x$, when $x = \True$ (right) and $x=\False$ (left).}
\label{fig:x1yx0}
\end{figure}

For the power cable, we simple add a particle for each position on the cable, without the suffix $R^2D^{N-j}$ because they are never discarded. Adding everything up, if Figure \ref{fig:example1dyn} we show the state after all the particles have been released in the case where $x = \True$ and $y =\False$. Now, we move to the evaluation of the gate. For this purpose we add a particle $P = (3, D^{13}R^{11})$ as was described above for the execution of the NOR gate. In our case, because position $(5,14)$ will be occupied by the transmitted value from $x=1$, $P$ will stick at $(5,13)$. Finally, we transmit the value for the output wire (0 in our particular case) through the same means described before.

\begin{figure}[h]
\begin{center}
    \includestandalone[width=\textwidth]{final}
\end{center}
\caption{Left: the evaluation of the wires transmitting values $x = \True$ and $y = \False$, as well as the power cable. Right: the evaluation of the NOR gate for that case.}
\label{fig:example1dyn}
\end{figure}

Therefore, our instance of the {2-\DPred} problem will be all the particles mentioned along with the output site, which becomes occupied if and only if the corresponding circuit (in this case just one gate) outputs $\True$.




Let us now consider the simulation of a slightly more complex circuit, consisting in three inputs and four NOR gates, as shown Figure \ref{fig:ex2}. In the same figure, we give a schema of the simulation, including the power cables and the single input OR gates. Then, in Figure \ref{fig:skeletonex2}  we give a skeleton of the construction, where we illustrate some of the positions that the particles visit in their trajectories. Finally, in Figure \ref{fig:iteration} we show the relative order in which the particles are thrown for in the different layers. 
\begin{figure}[h]
\begin{center}
    \includestandalone[scale=0.5]{circ2}
\end{center}
\caption{A second example of a circuit with tree inputs and two NOR gates. In left is depicted a planar NOR circuit with three inputs and four NOR gates. In right we have a schema of the reduction, where in red is represented the power cables. }
\label{fig:ex2}
\end{figure}
\begin{figure}
\begin{center}
    \includestandalone[scale=0.5]{skeleton2}
\end{center}
\caption{Schema of the positions of the particles in a simulation of the circuit given in Figure \ref{fig:ex2}. For simplicity, we omit the positions of the discarded particles.  }
\label{fig:skeletonex2}
\end{figure}

\begin{figure}
\begin{center}
    \includestandalone[width=0.6\textwidth]{iteracion1}\\ \vspace{0.3cm}
    \includestandalone[width=0.6\textwidth]{iteracion2}\\ \vspace{0.3cm}
    \includestandalone[width=0.6\textwidth]{iteracion3}
\end{center}
\caption{Schema of the order in which the particles are thrown in a simulation of the circuit given in Figure \ref{fig:ex2}, for the gates in the first layer (top left), the second layer (top right) and finally the third layer (bottom). The color code the following order: first red,  then orange, then yellow, then green and finally gray.  Particles of the same color are thrown from left to right, and/or respecting the order given in the definition gates.}
\label{fig:iteration}
\end{figure}






\clearpage
The reduction described above shows that the an instance of the \textsc{Planar-NOR-Circuit-Value Problem} is correctly evaluated by the growth of the 2-DLA cluster. A planar embedding of a planar circuit can be computed in \NC. From such an embedding, we can compute the positions of the corresponding gadgets (inputs, gates and wires). Furthermore, the size of each gate is constant. We deduce that this is a~ \NC~reduction.The key point is that the choice of paths for the walks is independent of the evaluation of the circuit. The full layout of the walks is given globally by the planar layout of the original circuit  All calculations required to compute these walks can be performed in~\NC. 

\end{proof}

\begin{corollary}
\textsc{3-\DPred} is \PC.

\end{corollary}

\begin{proof}
The particles in this case are allowed to move downwards, to the left and to the right. Thus, the proof of the $\bf{P}$-Completeness is straightforward. Because we already showed that given two directions, the prediction problem is \PC \ we can just ignore one of the lateral directions, and execute the same constructions shown in the previous theorem. 
\end{proof}

Because the aforementioned proofs rely only on the use of two dimensions, both results are directly extended to the dynamics in an arbitrary number of dimensions:

\begin{corollary}
\textsc{2-\DPred} and \textsc{3-\DPred} in $\mathbb{Z}^{d}$, are \PC.
\end{corollary}

\subsection{One Direction}
\label{sec:pred}

 By restricting the directions in which we allow particles to move, our problem statement simplifies. Because particles are only permitted to fall, there is no need to specify the whole trajectory of the particles, just the column of the $N\times N$ lattice we are throwing it down. Therefore, as a first method for representing the given behavior, we describe our input as a sequence of particle drops:  $S= a_{1}a_{2} \ ...\ a_{n-1}a_{n}$ where each $a_{i}\in[N]$ represents the column where the $i$-th particle is dropped, and $[n]$ denotes the set $\{1, \dots, n\}$ for each integer $n$.

This one-dimensional case is actually a particular instance of a more general model called \textit{Ballistic Deposition} (BD), first introduced by Vold and Sutherland to model colloidal aggregation \cite{vold1963computer, sutherland1966comments}. The growth model takes the substrate to be an undirected graph $G = (V,E)$, where each vertex defines a column through a ``height" function $h:V\rightarrow\mathbb{N}$, which represents the highest particle at the vertex. In addition, a sequence of particle throws is given by a list of vertices $S = v_1 v_2 \  ...  \ v_{n-1} v_{n}$, where a particle gets stuck at a height determined by its vertex, and all vertices neighboring it. It is easy to see that the one-dimensional DLA problem on a $N\times N$ square lattice is the special case when $G = ([N], \{(i,i+1): i\in[N]\})$.

Our prediction problem is as follows:

\vspace{2 mm}
\defproblemu{\textsc{BD Prediction}}
  { A graph $G = (V,E)$, a sequence $S$ of particle throws and a site $t = (h,v)\in \mathbb{N}\times V$, where $v$ is a vertex and $h$ a specified height for it.}{Is site $t$ occupied after the particles have been thrown into the graph?}
 \vspace{2 mm}

This problem was shown to be in $\NC^2$ by Matcha and Greenlaw using a Minimum-Weight Path parallel algorithm \cite{machta1994parallel}. We improve this result to show that the problem is in fact \textbf{\textsc{NL}}-Complete.




\begin{figure}[h]
	\begin{center}
	\includegraphics[width=0.75\textwidth]{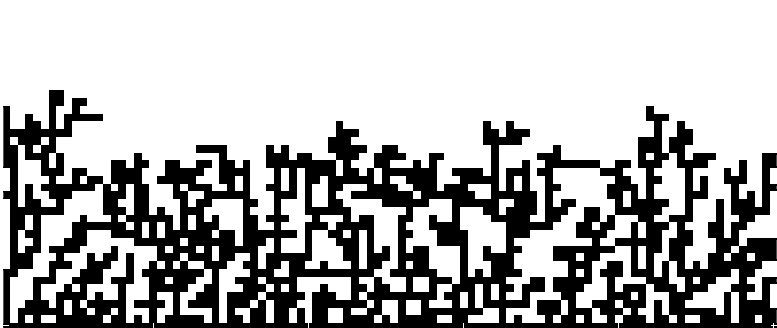}
	\caption{A realization of the one-directional dynamics of the system on a one-dimensional strip. }
	\end{center}     
\end{figure}

\subsubsection{Computational Capabilities of the Dynamics}

As a first look at the computational capabilities of the model, and sticking to the 1-DLA version of Ballistic Deposition, we show that we can sort natural numbers, simulating The Bead-Sort model described by Arulanandham et al. \cite{arulanandham2002bead}. This model consists of sorting natural numbers through gravity: numbers are represented by beads on rods, like an abacus, and are let loose to be subjected to gravity. As shown in \cite{arulanandham2002bead}, this process effectively sorts any given set of natural numbers. It is reasonable to think that because of the dynamics and constraints of our model (one direction of movement for the particles), the same sorting method can be applied within our model, which is in fact the case. 

\begin{lemma}
Bead-Sort can be simulated.
\end{lemma}

\begin{proof}
Let  $A$ be a set of $n$ positive natural numbers, with $m$ being the biggest number in the set. We create a $2m\times 2m$ lattice where we will be throwing the particles. Here the $k$-th rod from the Bead-Sort model is represented by row $2k-2$ on our lattice. Now, for each number $a\in A$ we create the sequence $S_a =$ 2 4 6 $\dotsc  2a$. The total sequence of launches $S$ is created by concatenating all $S_a$ for $a\in A$. We note that because of the commutativity of our model, the order in which the concatenation is made is not relevant. Thus, throwing sequence $S$ into our lattice, effectively simulates the Bead-Sort algorithm. 
\end{proof}

Let us give an example using the set $A = \{7, 4, 1, 10\}$. Following the proof, we must simulate 1-DLA on a $20\times 20$ square lattice, and create the sequences $S_1 =$ 2, $S_4 = $2 4 6 8, $S_7 =$ 2 4 6 8 10 12 14, and $S_{10} =$ 2 4 6 8 10 12 14 16 18 20. By releasing the sequence $S = S_1 \circ S_4 \circ S_7 \circ S_{10}$ into the lattice we obtain Figure \ref{fig:sort}, which is ordered increasingly, effectively sorting set $A$.

\begin{figure}[H]
\centering
\begin{tikzpicture}[x=0.75pt,y=0.75pt,yscale=-1,xscale=1]

\draw  [draw opacity=0] (527.5,298) -- (126.5,298) -- (126.5,158) -- (527.5,158) -- cycle ; \draw  [color={rgb, 255:red, 128; green, 128; blue, 128 }  ,draw opacity=1 ] (527.5,298) -- (527.5,158)(507.5,298) -- (507.5,158)(487.5,298) -- (487.5,158)(467.5,298) -- (467.5,158)(447.5,298) -- (447.5,158)(427.5,298) -- (427.5,158)(407.5,298) -- (407.5,158)(387.5,298) -- (387.5,158)(367.5,298) -- (367.5,158)(347.5,298) -- (347.5,158)(327.5,298) -- (327.5,158)(307.5,298) -- (307.5,158)(287.5,298) -- (287.5,158)(267.5,298) -- (267.5,158)(247.5,298) -- (247.5,158)(227.5,298) -- (227.5,158)(207.5,298) -- (207.5,158)(187.5,298) -- (187.5,158)(167.5,298) -- (167.5,158)(147.5,298) -- (147.5,158)(127.5,298) -- (127.5,158) ; \draw  [color={rgb, 255:red, 128; green, 128; blue, 128 }  ,draw opacity=1 ] (527.5,298) -- (126.5,298)(527.5,278) -- (126.5,278)(527.5,258) -- (126.5,258)(527.5,238) -- (126.5,238)(527.5,218) -- (126.5,218)(527.5,198) -- (126.5,198)(527.5,178) -- (126.5,178)(527.5,158) -- (126.5,158) ; \draw  [color={rgb, 255:red, 128; green, 128; blue, 128 }  ,draw opacity=1 ]  ;
\draw  [fill={rgb, 255:red, 0; green, 0; blue, 0 }  ,fill opacity=1 ] (147.5,278) -- (167.5,278) -- (167.5,298) -- (147.5,298) -- cycle ;
\draw  [fill={rgb, 255:red, 0; green, 0; blue, 0 }  ,fill opacity=1 ] (147.5,258) -- (167.5,258) -- (167.5,278) -- (147.5,278) -- cycle ;
\draw  [fill={rgb, 255:red, 0; green, 0; blue, 0 }  ,fill opacity=1 ] (187.5,278) -- (207.5,278) -- (207.5,298) -- (187.5,298) -- cycle ;
\draw  [fill={rgb, 255:red, 0; green, 0; blue, 0 }  ,fill opacity=1 ] (187.5,258) -- (207.5,258) -- (207.5,278) -- (187.5,278) -- cycle ;
\draw  [fill={rgb, 255:red, 0; green, 0; blue, 0 }  ,fill opacity=1 ] (187.5,238) -- (207.5,238) -- (207.5,258) -- (187.5,258) -- cycle ;
\draw  [fill={rgb, 255:red, 0; green, 0; blue, 0 }  ,fill opacity=1 ] (147.5,238) -- (167.5,238) -- (167.5,258) -- (147.5,258) -- cycle ;
\draw  [fill={rgb, 255:red, 0; green, 0; blue, 0 }  ,fill opacity=1 ] (147.5,218) -- (167.5,218) -- (167.5,238) -- (147.5,238) -- cycle ;
\draw  [fill={rgb, 255:red, 0; green, 0; blue, 0 }  ,fill opacity=1 ] (227.5,278) -- (247.5,278) -- (247.5,298) -- (227.5,298) -- cycle ;
\draw  [fill={rgb, 255:red, 0; green, 0; blue, 0 }  ,fill opacity=1 ] (267.5,278) -- (287.5,278) -- (287.5,298) -- (267.5,298) -- cycle ;
\draw  [fill={rgb, 255:red, 0; green, 0; blue, 0 }  ,fill opacity=1 ] (227.5,258) -- (247.5,258) -- (247.5,278) -- (227.5,278) -- cycle ;
\draw  [fill={rgb, 255:red, 0; green, 0; blue, 0 }  ,fill opacity=1 ] (267.5,258) -- (287.5,258) -- (287.5,278) -- (267.5,278) -- cycle ;
\draw  [fill={rgb, 255:red, 0; green, 0; blue, 0 }  ,fill opacity=1 ] (307.5,278) -- (327.5,278) -- (327.5,298) -- (307.5,298) -- cycle ;
\draw  [fill={rgb, 255:red, 0; green, 0; blue, 0 }  ,fill opacity=1 ] (387.5,278) -- (407.5,278) -- (407.5,298) -- (387.5,298) -- cycle ;
\draw  [fill={rgb, 255:red, 0; green, 0; blue, 0 }  ,fill opacity=1 ] (347.5,278) -- (367.5,278) -- (367.5,298) -- (347.5,298) -- cycle ;
\draw  [fill={rgb, 255:red, 0; green, 0; blue, 0 }  ,fill opacity=1 ] (307.5,258) -- (327.5,258) -- (327.5,278) -- (307.5,278) -- cycle ;
\draw  [fill={rgb, 255:red, 0; green, 0; blue, 0 }  ,fill opacity=1 ] (347.5,258) -- (367.5,258) -- (367.5,278) -- (347.5,278) -- cycle ;
\draw  [fill={rgb, 255:red, 0; green, 0; blue, 0 }  ,fill opacity=1 ] (387.5,258) -- (407.5,258) -- (407.5,278) -- (387.5,278) -- cycle ;
\draw  [fill={rgb, 255:red, 0; green, 0; blue, 0 }  ,fill opacity=1 ] (427.5,278) -- (447.5,278) -- (447.5,298) -- (427.5,298) -- cycle ;
\draw  [fill={rgb, 255:red, 0; green, 0; blue, 0 }  ,fill opacity=1 ] (467.5,278) -- (487.5,278) -- (487.5,298) -- (467.5,298) -- cycle ;
\draw  [fill={rgb, 255:red, 0; green, 0; blue, 0 }  ,fill opacity=1 ] (507.5,278) -- (527.5,278) -- (527.5,298) -- (507.5,298) -- cycle ;
\draw  [fill={rgb, 255:red, 0; green, 0; blue, 0 }  ,fill opacity=1 ] (227.5,238) -- (247.5,238) -- (247.5,258) -- (227.5,258) -- cycle ;
\draw  [fill={rgb, 255:red, 0; green, 0; blue, 0 }  ,fill opacity=1 ] (267.5,238) -- (287.5,238) -- (287.5,258) -- (267.5,258) -- cycle ;

\draw (102,288) node   {$10\ \rightarrow $};
\draw (103,268) node   {$7\ \rightarrow $};
\draw (103,248) node   {$4\ \rightarrow $};
\draw (102,229) node   {$1\ \rightarrow $};
\end{tikzpicture}

\caption{Figure obtained by trowing sequence $S$ into the lattice. Set $A$ is ordered decreasingly from the first row onwards.}
 \label{fig:sort}
\end{figure}


A key aspect of the present model is the commutativity of throws through non-consecutive vertices of the graph. Given a figure, $\mathcal{F}$, we define $\varphi_v(\mathcal{F})$ as the figure that results after throwing a particle through vertex $v$. We notice, that because of the dynamics of our model, the point at which a given particle freezes is determined uniquely by the state of the vertex it has been dropped in, and the state of its neighbors. Therefore, given $v\in V$: $$(\varphi_v\circ\varphi_u)(\mathcal{F}) = (\varphi_u\circ\varphi_v)(\mathcal{F}) \ \ \forall u\notin N_G(v).$$
We will later use this fact to create a better algorithm for the prediction problem.\\


\subsubsection{Non-deterministic Log-space Algorithm} 

There is a critical aspect of the dynamics that is exploitable to create an algorithm: if two particles are thrown on non-adjacent vertices, their relative order in the input sequence is reversible. By using this, our aim is to shuffle the sequence into another sequence with the same final configuration, but ordered in a way that allows us to quickly solve the prediction problem. Specifically, if we are able to reorganize the input sequence into one that releases particles according to the height they will ultimately end at, the remaining step to solve the problem is checking amongst the particles for the target height, if the target vertex appears.\\

Let $S$ be the input sequence of \textsc{BD Prediction}. Formally, a sequence $S = s_1 \dots s_n$ is composed of the vertices onto which each particle will be released. From $S$, we are able define a sequence of particles $p_1, \dots, p_n$ that represents the same realization as the input sequence. We define a \emph{particle}  $p$, as a triple $(V(p), \num(p), \pos(p)) \in V\times [n]\times[n]$, where the first coordinate, $V(p)$, denotes the vertex onto which the particle is thrown, the second coordinate, $\num(p)$, is an integer representing the number of particles thrown onto vertex $V(p)$ before $p$, and the third, $\pos(p)$ is the position of the particle within sequence $S$. The particle description of $S$ is easily obtained by setting $V(p_i) = s_i$, $\num(p_i) = |\{ j \in [n]: s_j = s_i \wedge j\leq i \}|$, and $\pos(p_i) = i$. 

Let us call $\mathbb{P} = \{ p_1, \dots, p_n\}$ the \emph{set of particles} of $S$. 

To further break down the problem, we define the following sets:
$$A(p) := \{q \in \mathbb{P}: \pos(q)<\pos(p)\},$$
$$N(p) := \{q \in A(p):  V(q)\in N_{G}(V(p))\cup\{V(p)\}\},$$
$$N^=(p) := \{q \in N(p): V(q) = V(p)\}.$$

In words, $A(p)$ denotes the set of particles thrown before $p$, $N(p)$ denotes the set of particles that are thrown before $p$ on vertices adjacent to $p$, and $N^=(p)$ denotes the subset of particles that belong to $N(p)$ and are in the same vertex as $p$. 

For a particle $p\in \mathbb{P}$, the \emph{row of $p$}, denoted $\row(p)$ is the height at which the particle ends up at after the dynamics have taken place. In other words, $\row(p) = h(V(p))$ after releasing the sequence $S' = s_{1} \dots s_{\pos(p)}.$   Relative to this definition, we call $N^r$ the set of particles thrown before $p$ in vertices adjacent than $p$ that stick at row $r$, formally:
$$N^r(p) = \{q \in N(p): \row(q) = r\}.$$

We translate the dynamics into this new notation in the following lemma:

\begin{lemma}\label{lem:rows}
Let $p \in \mathbb{P}$ be a particle and let 
$$r =\left\{  \begin{array}{ccc} 1 & \textrm{ if } & N(p) = \emptyset, \\ \max \{\row(q): q \in N(p)\} & \textrm{ if } & N(p) \neq \emptyset. \end{array} \right.$$
 Then, $$\row(p) =\left\{  \begin{array}{ccc} r+1 & \textrm{ if } & N^r(p) \cap N^=(p) \neq \emptyset,\\  r & \textrm{ if } & N^r(p) \cap N^=(p) = \emptyset.   \end{array} \right.$$
\end{lemma}

Explicitly, if the particle is the first of its neighbors to be thrown ($N(p) = \emptyset$), its row is 1. If not, that is, if its neighbors are higher than the vertex it is thrown in ($N^r(p) \cap N^=(p) = \emptyset$), the particle sticks at their height. Lastly, if the vertex that the particle is thrown in is higher that its neighbors ($N^r(p) \cap N^=(p) \neq \emptyset$), the particle sticks one row higher than the last particle in the vertex.

\begin{proof}

Let $p$ be a particle such that $N(p) = \emptyset$. This implies that $p$ is the first particle thrown through vertex $V(p)$ and its adjacent vertices. From the commutativity property, we deduce that $\row(p)  = 1$. On the other hand, $r = 1$ and $N^r(p) \cap N^=(p) = \emptyset$, so $\row(p) = r$. 

Suppose now that $N(p) \neq \emptyset$, and let $q$ be a particle in $N^r(p)$. Observe that $\row(q) = r$, and $\row(u) \leq r$ for all $u \in N(p)$. Then, when $p$ is thrown, the first particle that it encounters is $q$. 
Suppose that we can pick $q$ such that $V(q) = V(p)$ (i.e. $N^r(p) \cap N^=(p) \neq \emptyset$). Since $V(q) = V(p)$, we deduce that $\row(p) = \row(q) + 1 = r+1$. On the other hand, if $N^r(p) \cap N^=(p) = \emptyset$ then the coordinate $(V(p), r)$ is empty when $p$ is thrown, but some of $(u, r)$, for $u\in N_{G}(V(p))$, are occupied. We deduce that $\row(p) = \row(q) = r$. 
\end{proof}

We create a weighted graph that codifies the dependence of the particles to each other.
Let $G_S$ be a weighted directed graph defined from $S$ as follows:  the vertex set of $G_S$ is the set of particles $\mathbb{P}$ plus one more vertex $g$, called the \emph{ground} vertex. The edges of $G_S$ have weights, given by the weight function $W$ defined as:

$$W(p,q) = \left\{ \begin{array}{cl} 1 &\textrm{ if }~  (p=g) \wedge (N(q) = \emptyset), \\ 1 &\textrm{ if }~ p\in N^=(q), \\ 0 & \textrm{ if }~ p\in N(q) \setminus N^=(q),\\  -\infty & \textrm{ otherwise.}\end{array} \right.$$

Observe that if we keep only the edges with weight different than $-\infty$, the obtained graph has no directed cycle, i.e. it is a directed acyclic graph. Moreover, the set of incoming edges of vertex $p$ is $N(p)$ if $N(p)\neq \emptyset$, and $\{g\}$ otherwise. For $p\in \mathbb{P}$, we call $\tilde{\omega}_{gp}$ the longest (maximum weight) path from $g$ to $p$ in $G_S$.

\begin{theorem}\label{theo:grafodla}
For every $p\in \mathbb{P}$, $\row(p) = \tilde{\omega}_{gp}$.
\end{theorem}

\begin{proof}
We reason by induction on $\pos(p)$. Let $p\in \mathbb{P}$ be the particle such that $\pos(p)=1$. Observe that $p$ has only one incoming edge, which comes from $g$, and $W(g,p)=1$. Then $\tilde{\omega}_{gp} = 1 = \row(p)$.

Suppose now that $\row (p) = \tilde\omega_{gp}$ for every particle $q$ such that $\pos(q) \leq k$ and let $p$ be a particle with $\pos(p) = k+1$. If $N(p) = \emptyset$, then, like in the base case, the only incoming edge of $p$ is $g$, and from Lemma \ref{lem:rows} we deduce that $\row(p) = 1 = \tilde{\omega}_{gp}$. 

Suppose now that $N(p)$ is different than $\emptyset$.   Let $q$ be a particle in $N(p)$ such that $\row(q)$ is maximum, i.e. $\row(q) = \max\{ \row(u): u\in N(p)\} $. Observe that, from induction hypothesis and the choice of $q$, $\tilde\omega_{gq} \geq \tilde\omega_{gu}$ for all $u\in N(p) \setminus \{q\}$. Moreover, $\tilde\omega_{gp} \leq \tilde\omega_{gq} + 1$.

Suppose that $q$ can be chosen to be such that $V(q) = V(p)$.  Lemma \ref{lem:rows} then implies that $\row(p) = \row(q) + 1$. On the other hand, the path from $g$ to $p$ that passes through $q$ is of weight $\tilde\omega_{gq} + 1$. We deduce that $\tilde\omega_{gp} = \tilde\omega_{gq}+1 = \row(q)+1 = \row(p)$.

Suppose now that for all $u \in N^=(p)$, $\row(u)$ is strictly smaller than $\row(q)$. In this case, Lemma \ref{lem:rows} implies that $\row(p) = \row(q)$. On the other hand, the path from $g$ to $p$ that passes through $q$ is of weight $\tilde\omega_{gq}$, which is greater or equal than $\tilde\omega_{gu}$, for all $u\in N(p) \setminus N^=(p)$ and strictly greater than $\tilde\omega_{gu}$, for all $u \in N^=(p)$. We deduce that $\tilde\omega_{gp} = \tilde\omega_{gq} = \row(q)=\row(q)$.
\end{proof}
\vspace{1mm}
We now present an \textbf{NL}-algorithm for the prediction problem.

Given a site $(h, v)$, we want to non-deterministically obtain a path through graph $G_{S}$ that will guarantee the site will be occupied by a particle. Each step of our algorithm we will non-deterministically guess a pair consisting of the next particle, and the corresponding weight of the transition between the last particle and the new one, such that the sum of the weights is the maximum weight to the final particle.\\
We say a particle $p$ is valid for an input sequence $S$,  if $p\in\mathbb{P}$.  The following log-space algorithm verifies if a particle is valid.

\begin{algorithm}[H]

\KwIn{A sequence $S$ and a particle $p \in V\times[n]\times[n]$}
 \KwOut{Accept if particle $p$ is valid.}
Check if $V(p) = s_{\pos(p)}$\;

Sum $\leftarrow 0$

 \For{$i\in\{1, ..., \pos(p)\}$}{
  \If{$s_{i} = V(p)$}{
  	Sum $\leftarrow $ Sum $+ \ 1$\;
   }
   }
  \eIf{$\textnormal{Sum} = \num(p)$}{
    Accept\;
   }{Reject.\;}
 
 \caption{}

\end{algorithm}

At each step of the for loop of the algorithm, we must remember the value of Sum, the particles vertex, $V(p)$, and the current index of the iteration. This amounts to using $\cO(\log(n))$ space.

We also present a log-space algorithm to determine wether the obtained transition weight corresponds to the value of the weight function

\begin{algorithm}[H]

\KwIn{A sequence $S$, two valid particles $p,q \in V\times[n]\times[n]$, and $w\in\{0,1\}$}
 \KwOut{Accept if $W(p,q) = w$.}
 \If{$p = g$}{
  	\For{$i<\pos(q)$}{
		Check that $s_i$ is not adjacent to $V(q)$\;
	}
	Accept\;
   }
 \If{$w = 1$}{
 	Check that $\pos(p) < \pos(q)$ and $V(p) = V(q)$\;
	
	Accept\;
  }
  
   \If{$w = 0$}{
 	Check that $\pos(p) < \pos(q)$\;
	
	Check that $V(p)\neq V(q)$\;
	
	Check that $V(p)$ is adjacent to $V(q)$
	
	Accept\;
  }

Reject\;
 
 \caption{}

\end{algorithm}

For this algorithm, the only case where information needs to be stored is when $p = g$. For this instance, each iteration of the for loop must remember the index of the iteration, and the vertex $V(q)$. Therefore, this algorithm uses $\cO(\log(n))$ space.

	

Combining these two subroutines, we are now ready to present the main algorithm.
\begin{algorithm}[H]
\label{algo:main}

\KwIn{A graph $G = (V, E)$, a sequence $S$ and a site $t=(x,v) \in \mathbb{N}\times V$}
 \KwOut{Accept if a particle occupies site $t$ and reject otherwise.}

 Non-deterministically obtain $m$, the number of particles, and $p_1$. Write them down.\;

Check if $p_{1}$ is valid and that $W(g, p_1) = 1$\;

Sum $\leftarrow 1$\;

Write down Sum.

 \For{$j\in\{2, ..., m\}$}{
 Non-deterministically obtain particle $p_{j}$ and the transition weight $w_{j-1, j}$, and write them down.
 
 Check if $p_{j}$ is valid.\;
 
 Check if $W(p_{j-1},p_{j}) = w_{j-1, j}$\;
 
 Sum $\leftarrow $ Sum $ + \ w_{j-1, j}$\;
 
 Erase $p_{j-1}$ and $w_{j-1, j}$\;
 }
  \eIf{$\textnormal{Sum} = x$ and $V(p_m) = v$}{
   Accept\;
   }{Reject\;}
 \caption{\textbf{NL} algorithm for \textsc{BD Prediction}}

\end{algorithm}

At the $j$-th step of this algorithm, we must retain the following information: the sum of the weights so far, particles $p_{j-1}$ and $p_j$, the current weight $w_{j-1,j}$. This means that around $4\log(n) = \cO(\log(n))$ space is used on the tape.

\begin{proposition}
{\BPred} is in \textsc{\textbf{NL}}.
\end{proposition}

\begin{proof}
Let us show that Algorithm \ref{algo:main} decides \textsc{BD Prediction}.  Let $S$ be an input sequence and $P = (h,v)$ an input site. \\

If the release of $S$ on to the underlying graph results on site $P$ being occupied, by Theorem \ref{theo:grafodla} we know that there exists a particle $q\in\mathbb{P}$ such that $h = \row(q) = \tilde{\omega}_{gq}$ and $V(q) = v$. Let $C = g  \ p_1  \ p_2 \ ... \ p_{m-1} \ p_{m}$ be the maximum weight path of weight  $\tilde{\omega}_{gq}$, where $p_{m} = q$. Then, for the $j$-th non-deterministic choice the algorithm makes, it obtains the pair: $p_{j}$ and $W(p_{j-1}, p_{j})$.\\

If the algorithm accepts for $S$ and $P$, we will obtain a sequence of particles such that the sum of the transition weights is exactly $h$ and that $V(p_{m}) = v$. This means that the weight from the ground to the last particle will indeed be $h = \tilde{\omega}_{g p_{m}}$. Due to Theorem \ref{theo:grafodla}, this means that $\row(p_{m}) = \tilde{\omega}_{g p_{m}} = h$. Therefore, particle $p_m$ indeed occupies site $P$.
\end{proof}






\begin{theorem}
{\BPred} is \textsc{\textbf{NL}}-Complete.
\end{theorem}

\begin{proof}
Given proposition \ref{prop:LDEReach}, in order to show that the problem is \textbf{\textsc{{NL}}}-Hard, we will reduce an instance of {\LDEReach} to the Ballistic Deposition problem. \\

Let $(G, s, t, k)$ be an instance of {\LDEReach}, where $m$ is the number of layers of $G$, and let $i\in [m]$ be the index such that $s\in V_{i}$. The idea is to throw two particles for each vertex in all layers from $i$ to $i+k$. This way, the height at which the particles freeze will increase with each layer. Formally, for every $i < j\leq i+k$ and every $u\in V_{j}$ we create the sequence $S_{u} = u u$ (two particles are thrown in vertex $u$). Concatenating these sequences, we obtain a sequence of throws on the whole layer $S_j = \bigcirc_{u\in V_{j}} S_{u}$. We note that due to the structure of the graph and the commutativity of the dynamics, the order in which these sequences are concatenated does not matter because no two vertices in the same layer are adjacent.  

Finally, our input sequence will be the concatenation of the sequences associated to every layer from the $i$-th layer to the $i+k$-th one, $S = S_{s} \circ \bigcirc_{j=i+1}^{i+k} S_j$. The order in which these sequences are concatenated is important, and must be done in increasing order according to their index. At any point in this process the only information retained is the vertex for which we are currently creating the sequence $S_{u}$. This only requires $\log(n)$ space to store, making this process a log-space reduction.

By defining the site $P = (k+1, t)$, we create an instance of \BPred: $(G, S, P)$. Let us prove that this is indeed a reduction.\\

If $(G, s, t, k)\in \LDEReach$, then there exists a directed path $C = v_{0} \ ... \ v_{k}$ in $G$, where $s = v_0$ and $t = v_{k}$.  Because of the layered structure of the graph, if $s\in V_{i}$ then $v_{j}\in V_{i+j}$. Then, because by construction, for every vertex on $C$ two particles will be dropped, the height of the last particle dropped in $v_j$ will be $j+1$, meaning that the last particle dropped on $v_k = t$ will have a height of $k+1$. This means that site $P$ will in fact be occupied, and therefore $(G, S, P)\in\BPred$.\\

If $(G, S, P)\in\BPred$, site $P = (k+1, t)$ is occupied after the sequence of particles, $S$, has been released onto $G$. Due to our construction, if site $P$ is occupied, site $(k, t)$ must also be occupied by a particle. Let $l\in [m]$ be such that $t\in V_{l}$. Because only two particles are thrown at each vertex, for the latter site to be occupied there must exist a vertex $v_{1}\in V_{l-1}$ adjacent to $t$ such that sites $(k, v_{1})$ and $(k-1, v_{1})$ are occupied.

Iterating this process, we obtain a sequence $v_1 \ ... \ v_{k-1}$ such that $v_{i}$ is adjacent to $v_{i-1}$ and sites $(k - i + 1, v_{i})$ and $(k-i, v_{i})$ are occupied, for every $i\in\{2, ..., k-1\}$.  By virtue of the construction of $S$, the only posible way in which site $(1, v_{k-1})$ is occupied is that $s = v_{k-1}$. This proves that $(G, s, t, k)\in \LDEReach$, concluding our proof.
\end{proof}

\section{Shape Characterizations and Realization}

Although the figures obtained by simulating the dynamics are complex and fractal-like, not every shape is obtainable as an end product. This naturally leads to the problem of characterizing the figures which are obtainable through the dynamics, for the different restrictions of the DLA model.

A crucial observation is the fact that not all connected shapes are realizable. Figure 
\ref{infact} shows shapes that are not achievable for each of the restricted versions.

\begin{figure}[H]
	\label{infact}
	\begin{center}
	\includegraphics[scale=0.5]{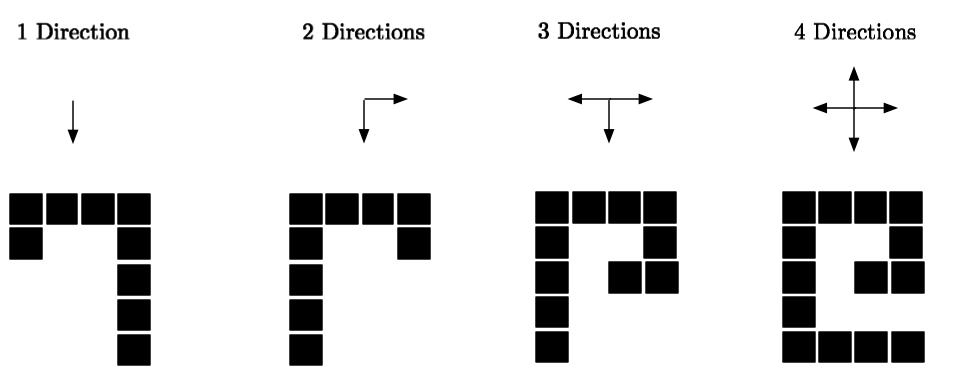}
	\label{fig:inf} 
	\caption{Four non-constructible figures with the respective maximum number of directions where it is not constructible. The last figure is not constructible even with four directions.}
	\end{center}     
\end{figure}

Given a fixed number of allowed directions, determining whether a given figure is realizable is an \textsc{\textbf{NP}} problem, where the non-deterministic choices of the algorithm are the order in which the particles are released. 
Having the order, it is possible to compute the trajectories that each particle takes in polynomial time, by finding a path that does not neighbor the already constructed cluster, from the place at which it is released to its final destination. \\

We give better algorithms for the shape characterization problems for both 1-DLA and 2-DLA, showing that the former belongs to the class of log-space solvable problems, $\Lspace$, and the latter to $\PP$. 

\subsection{One Direction}
To characterize shapes created by the one-directional dynamics, given a sequence of drops $S$, and its corresponding shape $\mathcal{F}(S)\in\{0,1\}^{m\times n}$, we construct a planar directed acyclic graph (DAG), $G = (V,E)$, that takes into account the ways in which a shape can be constructed with the dynamics. We construct $G$ in two steps.\\
We begin by constructing $G_a = (V_a, E_a)$, where:
\begin{equation*}
\begin{aligned}
	V_a  &=\{ij : \mathcal{F}(S)_{ij} = 1\} \cup \{g\} ,\\
	 E_1 &= \{(ij, i j +1) : \mathcal{F}(S)_{ij} = \mathcal{F}(S)_{i j + 1} = 1\}, \\
	E_2 &=\{(ij, i j-1) : \mathcal{F}(S)_{ij} = \mathcal{F}(S)_{i j - 1} = 1\},\\
	E_3 &=  \{(ij, i+1 j) : \mathcal{F}(S)_{ij} = \mathcal{F}(S)_{i+1 j} = 1\},\\
	E_4 &= \{(nj, g) : \mathcal{F}(S)_{nj} = 1\}, \end{aligned}
\end{equation*}
and $E_a = E_1 \cup E_2 \cup E_3 \cup E_4$.
The intuition is the following: we create a vertex for each block in the shape and one representing the ground where the initial particles stick. $E_1$ and $E_2$ account for the particles that stick through their sides, $E_3$ accounts for particles falling on top of each other and finally $E_4$ connects the first level to the ground.

For example, given the sequence $S =$ 2 7 7 2 6 3 4 4 4 5 6 3 2 6 2, we depict the obtained shape and corresponding graph in Figure \ref{fig:example}.

 \begin{figure}[H]
	\begin{center}
	\includegraphics[scale = 0.9]{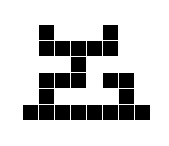}
	\label{fig:3}	
\begin{tikzpicture}[scale=0.1][H]
\tikzstyle{every node}+=[inner sep=0pt]
\draw [black] (44.5,-51.9) circle (3);
\draw (44.5,-51.9) node {$g$};
\draw [black] (18.7,-46.9) circle (3);
\draw (18.7,-46.9) node {$5, 2$};
\draw [black] (66.2,-45.1) circle (3);
\draw (66.2,-45.1) node {$5, 7$};
\draw [black] (18.7,-34.6) circle (3);
\draw (18.7,-34.6) node {$4, 2$};
\draw [black] (29,-34.6) circle (3);
\draw (29,-34.6) node {$4, 3$};
\draw [black] (39.2,-34.6) circle (3);
\draw (39.2,-34.6) node {$4, 4$};
\draw [black] (38.4,-26.6) circle (3);
\draw (38.4,-26.6) node {$3, 4$};
\draw [black] (38.4,-18.1) circle (3);
\draw (38.4,-18.1) node {$2, 4$};
\draw [black] (47.7,-19.1) circle (3);
\draw (47.7,-19.1) node {$2, 5$};
\draw [black] (57.5,-18.9) circle (3);
\draw (57.5,-18.9) node {$2, 6$};
\draw [black] (28.3,-18.9) circle (3);
\draw (28.3,-18.9) node {$2, 3$};
\draw [black] (18.7,-18.9) circle (3);
\draw (18.7,-18.9) node {$2, 2$};
\draw [black] (69.8,-33.3) circle (3);
\draw (69.8,-33.3) node {$4, 7$};
\draw [black] (57.5,-33.3) circle (3);
\draw (57.5,-33.3) node {$4, 6$};
\draw [black] (57.5,-9) circle (3);
\draw (57.5,-9) node {$1, 6$};
\draw [black] (18.7,-9) circle (3);
\draw (18.7,-9) node {$1, 2$};
\draw [black] (63.34,-46) -- (47.36,-51); \fill [black] (47.36,-51) -- (48.28,-51.24) -- (47.98,-50.29);
\draw [black] (21.65,-47.47) -- (41.55,-51.33); \fill [black] (41.55,-51.33) -- (40.86,-50.69) -- (40.67,-51.67);
\draw [black] (68.92,-36.17) -- (67.08,-42.23); \fill [black] (67.08,-42.23) -- (67.79,-41.61) -- (66.83,-41.32);
\draw [black] (60.5,-33.3) -- (66.8,-33.3); \fill [black] (66.8,-33.3) -- (66,-32.8) -- (66,-33.8); \draw [black] (66.8,-33.3) -- (60.5,-33.3); \fill [black] (60.5,-33.3) -- (61.3,-33.8) -- (61.3,-32.8); \draw [black] (18.7,-37.6) -- (18.7,-43.9); \fill [black] (18.7,-43.9) -- (19.2,-43.1) -- (18.2,-43.1); \draw [black] (26,-34.6) -- (21.7,-34.6); \fill [black] (21.7,-34.6) -- (22.5,-35.1) -- (22.5,-34.1); \draw [black] (21.7,-34.6) -- (26,-34.6); \fill [black] (26,-34.6) -- (25.2,-34.1) -- (25.2,-35.1); \draw [black] (36.2,-34.6) -- (32,-34.6); \fill [black] (32,-34.6) -- (32.8,-35.1) -- (32.8,-34.1); \draw [black] (32,-34.6) -- (36.2,-34.6); \fill [black] (36.2,-34.6) -- (35.4,-34.1) -- (35.4,-35.1); \draw [black] (38.7,-29.59) -- (38.9,-31.61); \fill [black] (38.9,-31.61) -- (39.32,-30.77) -- (38.32,-30.87);
\draw [black] (38.4,-21.1) -- (38.4,-23.6); \fill [black] (38.4,-23.6) -- (38.9,-22.8) -- (37.9,-22.8); \draw [black] (35.41,-18.34) -- (31.29,-18.66); \fill [black] (31.29,-18.66) -- (32.13,-19.1) -- (32.05,-18.1);
\draw [black] (31.29,-18.66) -- (35.41,-18.34); \fill [black] (35.41,-18.34) -- (34.57,-17.9) -- (34.65,-18.9);
\draw [black] (41.38,-18.42) -- (44.72,-18.78); \fill [black] (44.72,-18.78) -- (43.98,-18.2) -- (43.87,-19.19);
\draw [black] (44.72,-18.78) -- (41.38,-18.42); \fill [black] (41.38,-18.42) -- (42.12,-19) -- (42.23,-18.01);
\draw [black] (50.7,-19.04) -- (54.5,-18.96); \fill [black] (54.5,-18.96) -- (53.69,-18.48) -- (53.71,-19.48);
\draw [black] (54.5,-18.96) -- (50.7,-19.04); \fill [black] (50.7,-19.04) -- (51.51,-19.52) -- (51.49,-18.52);
\draw [black] (21.7,-18.9) -- (25.3,-18.9); \fill [black] (25.3,-18.9) -- (24.5,-18.4) -- (24.5,-19.4); \draw [black] (25.3,-18.9) -- (21.7,-18.9); \fill [black] (21.7,-18.9) -- (22.5,-19.4) -- (22.5,-18.4); \draw [black] (18.7,-12) -- (18.7,-15.9); \fill [black] (18.7,-15.9) -- (19.2,-15.1) -- (18.2,-15.1); \draw [black] (57.5,-12) -- (57.5,-15.9); \fill [black] (57.5,-15.9) -- (58,-15.1) -- (57,-15.1); 
\end{tikzpicture}
\end{center}
\label{fig:example}
\caption{Shape and graph obtained from sequence $S =$ 2 7 7 2 6 3 4 4 4 5 6 3 2 6 2.}
\end{figure}
 
\begin{lemma}
\label{constr}
A configuration is constructible iff $\forall ij \in V_a\setminus\{g\}$ there exists a directed path between $ij$ and $g$.
\end{lemma}

\begin{proof}
 Given a constructible configuration, by virtue of the definition, there is a sequence $S$ of particle drops that generates the configuration. Therefore, by the dynamics of our system and the construction of the graph, for each node in our graph, there exists a directed path to the \emph{ground}, represented by node $g$.
 
Now, let $G = (V,E)$ be the graph corresponding to a given configuration on the grid. If we have that $\forall ij \in V\setminus\{g\}$ there exists a directed path between $ij$ and $g$, we start by reversing the direction of the arcs in our graph, and running a Breadth-first search-like algorithm starting from node $g$, to determine the minimum distance from $g$ to each of the other nodes. For all nodes with distance $1$ from $g$, $\{i_k j_k\}_{k=1}^{n}$, we create the sequence $S_1 = j_1 ... j_n$. Then for all non-visited nodes, with distance $d$ from $g$,  $\{a_k b_k\}_{k=1}^{m}$, we create the sequence $S_d = b_1 \ ...\  b_m$. Let us show that the sequence $S = S_1 \circ S_2 \circ ... \circ S_M$, with $M = \max\{\text{dist}(ij, p): ij \in V\setminus\{g\}\}$ corresponds to the configuration. We do this by induction over $|S| = n$.

For $n=1$, we have only one occupied state on our configuration, which is the only particle present in $S$. It is straightforward to see that $\mathcal{F}(S)$ actually corresponds to the configuration.

Assuming the sequence is correct for $n$, let us prove that it is also correct for $n+1$. Because of the way $S$ was constructed, the only possibility for the configuration not to be achieved is that $S_{n+1}$ ends up higher or lower in $\mathcal{F}(S)$ on its column that in the given configuration. Due to the dynamics, the only way for a particle to become inmobile is to stick to another particle, this means that either it sticks to a particle in its own column, or to a particle in one of the neighboring ones. By the induction hypothesis, $S\setminus S_{n+1}$ actually corresponds to the configuration of the first $n$ particles. Therefore, if, without loss of generality, $S_{n+1}$ ends up lower, this means that when generating the sequence $S$ a node that was at a smaller distance from $g$ than the particle with which it sticks, was added after all other nodes that at the same distance, which is a contradiction.
\end{proof}

Next, we obtain $G$ from $G_a$ through the following procedure:
\begin{enumerate}

	\item For every $v\in V_{a}\setminus\{g\}$, we create two vertices $v_{1}$ and $v_{2}$, connected by an arc from the former to the latter.
	
	\item For every $e = (v,u)\in E_1$, we create an edge $(v_1, u_1)$. 
	
	\item For every $e = (v,u)\in E_2$, we create an edge $(v_2, u_2)$.
	
	\item For every $e = (v,u)\in E_3$, we create an edge $(v_2, u_1)$.
	
	\item For every $e = (v,g)\in E_4$, we create an edge $(v_2, g)$.

\end{enumerate}

\vspace{4mm}
\begin{center}
\begin{tikzpicture}[x=0.75pt,y=0.75pt,yscale=-1,xscale=1]

\draw   (110.17,207.67) .. controls (110.17,202.33) and (114.49,198) .. (119.83,198) .. controls (125.17,198) and (129.5,202.33) .. (129.5,207.67) .. controls (129.5,213.01) and (125.17,217.33) .. (119.83,217.33) .. controls (114.49,217.33) and (110.17,213.01) .. (110.17,207.67) -- cycle ;
\draw   (169.83,208.33) .. controls (169.83,202.63) and (174.46,198) .. (180.17,198) .. controls (185.87,198) and (190.5,202.63) .. (190.5,208.33) .. controls (190.5,214.04) and (185.87,218.67) .. (180.17,218.67) .. controls (174.46,218.67) and (169.83,214.04) .. (169.83,208.33) -- cycle ;
\draw   (110.5,168.33) .. controls (110.5,162.63) and (115.13,158) .. (120.83,158) .. controls (126.54,158) and (131.17,162.63) .. (131.17,168.33) .. controls (131.17,174.04) and (126.54,178.67) .. (120.83,178.67) .. controls (115.13,178.67) and (110.5,174.04) .. (110.5,168.33) -- cycle ;
\draw   (139.83,259) .. controls (139.83,253.29) and (144.46,248.67) .. (150.17,248.67) .. controls (155.87,248.67) and (160.5,253.29) .. (160.5,259) .. controls (160.5,264.71) and (155.87,269.33) .. (150.17,269.33) .. controls (144.46,269.33) and (139.83,264.71) .. (139.83,259) -- cycle ;
\draw    (120.83,178.67) -- (119.94,196) ;
\draw [shift={(119.83,198)}, rotate = 272.96] [fill={rgb, 255:red, 0; green, 0; blue, 0 }  ][line width=0.75]  [draw opacity=0] (8.93,-4.29) -- (0,0) -- (8.93,4.29) -- cycle    ;

\draw    (119.83,217.33) -- (148.78,247.23) ;
\draw [shift={(150.17,248.67)}, rotate = 225.93] [fill={rgb, 255:red, 0; green, 0; blue, 0 }  ][line width=0.75]  [draw opacity=0] (8.93,-4.29) -- (0,0) -- (8.93,4.29) -- cycle    ;

\draw    (180.17,218.67) -- (151.58,247.25) ;
\draw [shift={(150.17,248.67)}, rotate = 315] [fill={rgb, 255:red, 0; green, 0; blue, 0 }  ][line width=0.75]  [draw opacity=0] (8.93,-4.29) -- (0,0) -- (8.93,4.29) -- cycle    ;

\draw    (131.5,207.7) -- (167.83,208.3) ;
\draw [shift={(169.83,208.33)}, rotate = 180.95] [fill={rgb, 255:red, 0; green, 0; blue, 0 }  ][line width=0.75]  [draw opacity=0] (8.93,-4.29) -- (0,0) -- (8.93,4.29) -- cycle    ;
\draw [shift={(129.5,207.67)}, rotate = 0.95] [fill={rgb, 255:red, 0; green, 0; blue, 0 }  ][line width=0.75]  [draw opacity=0] (8.93,-4.29) -- (0,0) -- (8.93,4.29) -- cycle    ;
\draw   (337.67,113.25) .. controls (337.67,106.48) and (343.15,101) .. (349.92,101) .. controls (356.68,101) and (362.17,106.48) .. (362.17,113.25) .. controls (362.17,120.02) and (356.68,125.5) .. (349.92,125.5) .. controls (343.15,125.5) and (337.67,120.02) .. (337.67,113.25) -- cycle ;
\draw   (337.67,153.25) .. controls (337.67,146.48) and (343.15,141) .. (349.92,141) .. controls (356.68,141) and (362.17,146.48) .. (362.17,153.25) .. controls (362.17,160.02) and (356.68,165.5) .. (349.92,165.5) .. controls (343.15,165.5) and (337.67,160.02) .. (337.67,153.25) -- cycle ;
\draw    (349.92,125.5) -- (349.92,139) ;
\draw [shift={(349.92,141)}, rotate = 270] [fill={rgb, 255:red, 0; green, 0; blue, 0 }  ][line width=0.75]  [draw opacity=0] (8.93,-4.29) -- (0,0) -- (8.93,4.29) -- cycle    ;

\draw   (356.17,192.25) .. controls (356.17,185.48) and (361.65,180) .. (368.42,180) .. controls (375.18,180) and (380.67,185.48) .. (380.67,192.25) .. controls (380.67,199.02) and (375.18,204.5) .. (368.42,204.5) .. controls (361.65,204.5) and (356.17,199.02) .. (356.17,192.25) -- cycle ;
\draw   (356.17,230.25) .. controls (356.17,223.48) and (361.65,218) .. (368.42,218) .. controls (375.18,218) and (380.67,223.48) .. (380.67,230.25) .. controls (380.67,237.02) and (375.18,242.5) .. (368.42,242.5) .. controls (361.65,242.5) and (356.17,237.02) .. (356.17,230.25) -- cycle ;
\draw    (368.42,204.5) -- (368.42,216) ;
\draw [shift={(368.42,218)}, rotate = 270] [fill={rgb, 255:red, 0; green, 0; blue, 0 }  ][line width=0.75]  [draw opacity=0] (8.93,-4.29) -- (0,0) -- (8.93,4.29) -- cycle    ;

\draw   (418.67,192.25) .. controls (418.67,185.48) and (424.15,180) .. (430.92,180) .. controls (437.68,180) and (443.17,185.48) .. (443.17,192.25) .. controls (443.17,199.02) and (437.68,204.5) .. (430.92,204.5) .. controls (424.15,204.5) and (418.67,199.02) .. (418.67,192.25) -- cycle ;
\draw   (418.67,230.25) .. controls (418.67,223.48) and (424.15,218) .. (430.92,218) .. controls (437.68,218) and (443.17,223.48) .. (443.17,230.25) .. controls (443.17,237.02) and (437.68,242.5) .. (430.92,242.5) .. controls (424.15,242.5) and (418.67,237.02) .. (418.67,230.25) -- cycle ;
\draw    (430.92,204.5) -- (430.92,216) ;
\draw [shift={(430.92,218)}, rotate = 270] [fill={rgb, 255:red, 0; green, 0; blue, 0 }  ][line width=0.75]  [draw opacity=0] (8.93,-4.29) -- (0,0) -- (8.93,4.29) -- cycle    ;

\draw    (380.67,192.25) -- (416.67,192.25) ;
\draw [shift={(418.67,192.25)}, rotate = 180] [fill={rgb, 255:red, 0; green, 0; blue, 0 }  ][line width=0.75]  [draw opacity=0] (8.93,-4.29) -- (0,0) -- (8.93,4.29) -- cycle    ;

\draw    (418.67,230.25) -- (382.67,230.25) ;
\draw [shift={(380.67,230.25)}, rotate = 360] [fill={rgb, 255:red, 0; green, 0; blue, 0 }  ][line width=0.75]  [draw opacity=0] (8.93,-4.29) -- (0,0) -- (8.93,4.29) -- cycle    ;

\draw    (349.92,165.5) -- (361.33,181.38) ;
\draw [shift={(362.5,183)}, rotate = 234.28] [fill={rgb, 255:red, 0; green, 0; blue, 0 }  ][line width=0.75]  [draw opacity=0] (8.93,-4.29) -- (0,0) -- (8.93,4.29) -- cycle    ;

\draw   (389.33,279.5) .. controls (389.33,273.79) and (393.96,269.17) .. (399.67,269.17) .. controls (405.37,269.17) and (410,273.79) .. (410,279.5) .. controls (410,285.21) and (405.37,289.83) .. (399.67,289.83) .. controls (393.96,289.83) and (389.33,285.21) .. (389.33,279.5) -- cycle ;
\draw    (368.42,242.5) -- (398.15,267.87) ;
\draw [shift={(399.67,269.17)}, rotate = 220.48] [fill={rgb, 255:red, 0; green, 0; blue, 0 }  ][line width=0.75]  [draw opacity=0] (8.93,-4.29) -- (0,0) -- (8.93,4.29) -- cycle    ;

\draw    (430.92,242.5) -- (401.19,267.87) ;
\draw [shift={(399.67,269.17)}, rotate = 319.52] [fill={rgb, 255:red, 0; green, 0; blue, 0 }  ][line width=0.75]  [draw opacity=0] (8.93,-4.29) -- (0,0) -- (8.93,4.29) -- cycle    ;

\draw    (220.5,213) .. controls (222.17,211.33) and (223.83,211.33) .. (225.5,213) .. controls (227.17,214.67) and (228.83,214.67) .. (230.5,213) .. controls (232.17,211.33) and (233.83,211.33) .. (235.5,213) .. controls (237.17,214.67) and (238.83,214.67) .. (240.5,213) .. controls (242.17,211.33) and (243.83,211.33) .. (245.5,213) .. controls (247.17,214.67) and (248.83,214.67) .. (250.5,213) .. controls (252.17,211.33) and (253.83,211.33) .. (255.5,213) .. controls (257.17,214.67) and (258.83,214.67) .. (260.5,213) .. controls (262.17,211.33) and (263.83,211.33) .. (265.5,213) .. controls (267.17,214.67) and (268.83,214.67) .. (270.5,213) .. controls (272.17,211.33) and (273.83,211.33) .. (275.5,213) .. controls (277.17,214.67) and (278.83,214.67) .. (280.5,213) .. controls (282.17,211.33) and (283.83,211.33) .. (285.5,213) .. controls (287.17,214.67) and (288.83,214.67) .. (290.5,213) .. controls (292.17,211.33) and (293.83,211.33) .. (295.5,213) -- (299.5,213) -- (307.5,213) ;
\draw [shift={(309.5,213)}, rotate = 180] [fill={rgb, 255:red, 0; green, 0; blue, 0 }  ][line width=0.75]  [draw opacity=0] (8.93,-4.29) -- (0,0) -- (8.93,4.29) -- cycle    ;

\draw (121.33,166) node   {$v$};
\draw (119.83,204.67) node   {$u$};
\draw (180.5,205.67) node   {$w$};
\draw (151.17,258) node   {$g$};
\draw (350.42,110.75) node   {$v_{1}$};
\draw (351.33,150.83) node   {$v_{2}$};
\draw (368.92,189.75) node   {$u_{1}$};
\draw (369.83,227.83) node   {$u_{2}$};
\draw (431.42,189.75) node   {$w_{1}$};
\draw (430.33,228.33) node   {$w_{2}$};
\draw (400.67,278) node   {$g$};
\end{tikzpicture}
\end{center}

This procedure turns $G_a$ into a planar DAG.\\

It is straightforward to see that there is a directed path from a vertex $v\in V_{a}$ to $g$ on graph $G_a$ if and only if there is a path from vertex $v_1$ to $g$ on graph $G$.\\

Due to our construction, graph $G$ is what is known as a Multiple Source Single Sink Planar DAG (MSPD): there are multiple vertices with in-degree zero, and one vertex with out-degree zero. Allender et al. showed that the reachability problem on MSPD is in fact log-space solvable \cite{allender2006grid}.

\begin{theorem}[\cite{allender2006grid}, Theorem 5.7]
MSPD reachability is in \Lspace.
\end{theorem}

\begin{proposition}
Determining whether a given figure is a valid configuration for 1-DLA, is in $\Lspace$.
\end{proposition}

The proof of this proposition consists on creating a log-space reduction from our realization problem to MSPD reachability. This stems from the fact that if a problem is log-space reducible to a log-space solvable problem, it is itself log-space solvable (the proof of this fact can be found in \cite{arora2009computational}).

\begin{proof}

Let $M$ be a matrix representing the configuration, using the same convention as the definition shape. We can see that the directed graph is constructible from a shape in $\NC^{1}$. By assigning a processor for each pair of coordinates in the matrix, that is $\mathcal{O}(n^2)$ processors, to construct the nodes and arcs of our graph. Because $\NC^{1}\subseteq\Lspace$, this procedure is realizable in log-space. This procedure creates the MSPD $G$. 

For each vertex $v\in V$, we can solve the MSPD reachability problem for the instance $(G, v, g)$ in log-space. Due to Lemma \ref{constr}, this solves the one-directional realization problem.

\end{proof}

\subsection{Two directions}

To characterize shapes generated by the two-directional dynamics, we follow a similar approach  than we did for one-directional dynamics. We define for each figure a directed graph called the \emph{dependency graph}, which has as vertex set the set of all particles in the input figure $F$. The dependency graph satisfies that, if a particle $u$ is an in-neighbor of a particle $v$, then $u$ must be fixed before $v$ on every realization of the figure. Moreover, we show that the dependency graph characterizes the realizable figures. More precisely, we show that a figure is realizable if and only if its is acyclic. 

Let $M$ be the input matrix of an instance of \textsc{$2$-\DReal}. For simplicity, in this section we assume that $M$ is a matrix of dimensions $[N+1]\times[N]$ where all the coordinates in row $N+1$ have value $1$ (representing the \emph{ground}). Moreover, we assume that $N$ is large enough so $M$ has a border of zeroes. More precisely, every coordinate of $M$ that equals $1$ and that is not in row $N+1$, is in a row and column of $M$ that is greater than $3$, and in a column fewer than $N-3$. More formally, we assume that $M = (m_{ij})$  satisfies:

$$ i \leq 3 \vee j \leq 3 \vee j\geq N-3 \Rightarrow m_{ij} = 0.$$

We will say a coordinate $(i,j)\in [N]\times [N]$ is a \emph{particle} if $M_{ij} = 1$. The set of all particles is called a \emph{figure} $F$. We denote the row and column of a particle $p\in F$ by row$(p)$ and col$(p)$ respectively. Two particles $p$ and $q$ are said to be \emph{adjacent} if $|\col(p)-\col(q)| + |\row(p)-\row(q)| = 1$. We call $n$ the total number of particles, i.e. $n = |F|$.

\begin{definition}
Let $F\subseteq\{0,1\}^{N\times N}$ be a figure. A 2DLA-realization of $F$ is a bijective function $\varphi: F \to [n]$ such that, for every $t\in [n]$:
\begin{itemize}
\item Each particle is fixed to an adjacent particle already fixed, i.e.,  $\varphi^{-1}(t)$ has an adjacent particle in $\varphi^{-1}(\{1, \dots, t-1\})$ or $\row(\varphi^{-1}(t)) =N$.
\item There exists an \emph{available path} for $\varphi^{-1}(t)$, that is to say, a path $P = p_0, \dots p_k$ such that 
\begin{itemize}
\item The path starts in $(0,0)$ and reaches $p_k$, i.e. , $p_0 = (0,0)$ and $p_k = \varphi^{-1}(t)$,
\item The particle is not fixed before its final destination, i.e.,  $0 < \ell < k$, cell $p_\ell$ is not adjacent to any cell in $\varphi^{-1}(\{1, \dots, t-1\})$ and $\row(p_\ell)<N$.
\item The path reaches the cell using only two directions, i.e, for each $0\leq \ell < k$ cell $p_{\ell+1} = (\row(p_\ell)+1, \col(p_\ell)) $ or $p_{\ell+1} = (\row(p_\ell), \col(p_\ell)+1) $
\end{itemize}
\end{itemize}
\end{definition}

Let us better understand how the relative position of particles affect the order in which they must be placed. For a particle $u\in F$, we define the \emph{shadow} of the particle as:
$$S(u) = \{(i,j): \row(u) > i \ \wedge \  \col(u) > j\}$$

\begin{figure}[H]
\begin{center}
\begin{tikzpicture}[x=0.75pt,y=0.75pt,yscale=-1,xscale=1]
\draw  [color={rgb, 255:red, 0; green, 0; blue, 0 }  ,draw opacity=1 ][fill={rgb, 255:red, 255; green, 255; blue, 255 }  ,fill opacity=1 ] (69.57,40.87) -- (89.47,40.87) -- (89.47,59.67) -- (69.57,59.67) -- cycle ;
\draw  [color={rgb, 255:red, 0; green, 0; blue, 0 }  ,draw opacity=1 ][fill={rgb, 255:red, 186; green, 186; blue, 186 }  ,fill opacity=1 ] (89.47,59.67) -- (179.67,59.67) -- (179.67,149) -- (89.47,149) -- cycle ;
\draw  [dash pattern={on 4.5pt off 4.5pt}]  (89.47,40.87) -- (181,40.33) ;
\draw  [dash pattern={on 4.5pt off 4.5pt}]  (69.57,59.67) -- (70.33,149.67) ;
\draw (75,46) node [anchor=north west][inner sep=0.75pt]{$u$};
\draw (116.67,90.57) node [anchor=north west][inner sep=0.75pt]{$S(u)$};
\end{tikzpicture}
\end{center}
\caption{}
\end{figure}

\begin{definition}
Let $F$ be a 2-DLA constructible figure. A realization $\varphi$ of $F$ is said to be canonical if 
$$\forall u,v\in F: \ \ u\in S(v) \implies \varphi(u) < \varphi(v).$$
\end{definition}

The following lemma states that every figure that is realizable admits a canonical realization. 

\begin{lemma}
\label{canon}
$F$ admits a 2-DLA realization if and only if $F$ has a canonical realization.
\end{lemma}

\begin{proof}
If $F$ has a canonical realization it is trivially 2-DLA constructible. 

Let us suppose that $F$ is 2-DLA constructible but has no canonical realization. We will say a pair of particles $u,v$ is a \emph{bad pair} of a realization $\varphi$ if $\varphi(u) > \varphi(v)$ and $u\in S(v)$. Let $\varphi$ be the realization with the least amount of bad pairs. We will show that it is possible to define a realization with a lower amount of bad pairs, giving us a contradiction.

From all possible bad pairs $(v,u)$, let us take one that minimizes $\col(v)$. Fixing such $v$, let us take $u$ as the first cell to be placed in $S(v)$ after $v$. In order to reach a contradiction, we consider two cases. 

\textbf{Case 1:} Let us suppose that $u$ is strictly in the interior of $S(v)$, that is, col$(u)>$ col$(v) +1$ and row$(u) >$ row$(v) + 1$ as see on Figure \ref{interior}.

\begin{figure}[H]
\label{interior}
\begin{center}
\begin{tikzpicture}[x=0.75pt,y=0.75pt,yscale=-1,xscale=1]
\draw  [color={rgb, 255:red, 0; green, 0; blue, 0 }  ,draw opacity=1 ][fill={rgb, 255:red, 255; green, 255; blue, 255 }  ,fill opacity=1 ] (69.57,40.87) -- (89.47,40.87) -- (89.47,59.67) -- (69.57,59.67) -- cycle ;
\draw  [color={rgb, 255:red, 0; green, 0; blue, 0 }  ,draw opacity=1 ][fill={rgb, 255:red, 186; green, 186; blue, 186 }  ,fill opacity=0.71 ] (89.47,59.67) -- (179.67,59.67) -- (179.67,149) -- (89.47,149) -- cycle ;
\draw  [dash pattern={on 4.5pt off 4.5pt}]  (89.97,80.37) -- (181.5,79.83) ;
\draw  [dash pattern={on 4.5pt off 4.5pt}]  (69.57,59.67) -- (70.33,149.67) ;
\draw  [dash pattern={on 4.5pt off 4.5pt}]  (89.47,40.87) -- (181,40.33) ;
\draw  [dash pattern={on 4.5pt off 4.5pt}]  (109.57,60.17) -- (110.33,150.17) ;
\draw  [color={rgb, 255:red, 0; green, 0; blue, 0 }  ,draw opacity=1 ][fill={rgb, 255:red, 255; green, 255; blue, 255 }  ,fill opacity=1 ] (140.57,100.37) -- (160.47,100.37) -- (160.47,119.17) -- (140.57,119.17) -- cycle ;
\draw (75,46) node [anchor=north west][inner sep=0.75pt]    {$v$};
\draw (145,105) node [anchor=north west][inner sep=0.75pt]    {$u$};
\end{tikzpicture}
\end{center}
\caption{Case 1: $u$ is in the interior of $S(v)$.}
\end{figure}

For $w\in F$ let us call $F_w(\varphi)$ the set of particles such that
$$F_w(\varphi) = \{z\in F: \varphi(z)\leq\varphi(w)\}$$

It is clear that in this case, we must have that $u$ must be adjacent to a particle in $F_v(\varphi)$. If not, we would have a contradiction with the fact that $u$ is the fist element to be placed after $v$ and that means it would be placed without a contact point.

We now define $\tilde{\varphi}$ as the realization that equals $\varphi$, except that we place $u$ right before $v$, more formally:

$$\tilde{\varphi}(p) = \begin{cases}\varphi(p) \ \text{ if } \ \varphi(p) < \varphi(v) \\ \varphi(v) \ \text{ if } \ p = u \\ \varphi(p) + 1 \ \text{ if } \ \varphi(v) \leq \varphi(p) < \varphi(u) \\ \varphi(p) \ \text{ if } \ \varphi(p) > \varphi(u) \end{cases}$$

Let us see that $\tilde{\varphi}$ is a realization. Let $w$ be any particle of $F$. Let $p_w$ denote the trajectory the particle takes to reach cell $w$ in the realization given by $\varphi$. Clearly if $\varphi(w)<\varphi(v)$ or $\varphi(w)>\varphi(u)$ then $F_w(\varphi) = F_w(\tilde{\varphi})$ and then $p_w$ is still available as a trajectory  for the new realization $\tilde{\varphi}$. If $w = u$ then  $F_w(\tilde{\varphi}) \subseteq F_w(\varphi) $ and then $P_w$ is also available as a trajectory  for the new realization $\tilde{\varphi}$. Suppose then that $\varphi(v) \leq \varphi(w) < \varphi(u)$. If $u$ is not adjacent to any of the coordinates on $p_w$ (for example when $w=v$), this path is also valid as a trajectory for the new realization $\tilde{\varphi}$. We can observe that this is always the case: if $w\not\in S(v)$, $p_w$ can never be adjacent to $u$. If $w\in S(v)$, because this means that $\varphi(u)>\varphi(w)$, $u$ cannot be adjacent to $p_w$.

Therefore, $\tilde{\varphi}$ is a realization of $F$. Moreover, $\tilde{\varphi}$  has a lower number of bad pairs. Indeed, on one hand $(v,u)$ is no longer a bad pair of $\tilde{\varphi}$. On the other, suppose that $(w_1,w_2)$ is a bad pair of $\tilde{\varphi}$ that is not a bad pair of $\varphi$. Then necessarily $w_1 = u$ and $w_2 \in S(u)$ is such that  $\varphi(v) < \varphi(w_2) < \varphi(u)$. We obtain a contradiction with the fact that $S(u)\subset S(v)$ and $u$ was the first particle in $S(v)$ to be placed after $v$. We deduce that $\tilde{\varphi}$ is a realization of $F$ with a lower number of bad pairs than $\varphi$, which is a contradiction with the choice of $\varphi$. 

\textbf{Case 2:} Let us now suppose that $u$ is in the border of $S(v)$, that is, $\col(u) =\col(v) +1$ or $\row(u) =\row(v) +1$. Without loss of generality let us assume that $\col(u) =\col(v) +1$.

\begin{figure}[H]
\begin{center}
\begin{tikzpicture}[x=0.75pt,y=0.75pt,yscale=-1,xscale=1]
\draw  [color={rgb, 255:red, 0; green, 0; blue, 0 }  ,draw opacity=1 ][fill={rgb, 255:red, 255; green, 255; blue, 255 }  ,fill opacity=1 ] (69.57,40.87) -- (89.47,40.87) -- (89.47,59.67) -- (69.57,59.67) -- cycle ;
\draw  [color={rgb, 255:red, 0; green, 0; blue, 0 }  ,draw opacity=1 ][fill={rgb, 255:red, 186; green, 186; blue, 186 }  ,fill opacity=0.71 ] (89.47,59.67) -- (179.67,59.67) -- (179.67,149) -- (89.47,149) -- cycle ;
\draw  [dash pattern={on 4.5pt off 4.5pt}]  (89.97,80.37) -- (181.5,79.83) ;
\draw  [dash pattern={on 4.5pt off 4.5pt}]  (69.57,59.67) -- (70.33,149.67) ;
\draw  [dash pattern={on 4.5pt off 4.5pt}]  (89.47,40.87) -- (181,40.33) ;
\draw  [dash pattern={on 4.5pt off 4.5pt}]  (109.57,60.17) -- (110.33,149) ;
\draw  [color={rgb, 255:red, 0; green, 0; blue, 0 }  ,draw opacity=1 ][fill={rgb, 255:red, 255; green, 255; blue, 255 }  ,fill opacity=1 ] (89.67,105.17) -- (109.95,105.17) -- (109.95,123.67) -- (89.67,123.67) -- cycle ;
\draw (75,46) node [anchor=north west][inner sep=0.75pt]    {$v$};
\draw (94,110) node [anchor=north west][inner sep=0.75pt]    {$u$};
\end{tikzpicture}
\end{center}
\caption{Case 2: when $u$ is in the border of $v$ ($\col(u) = \col(v)+1$)}
\end{figure}

 Let $p_u$ denote the trajectory the particle takes to reach cell $u$ in the realization given by $\varphi$. We have that $p_u$ must necessarily contain a coordinate $(i,j)$ such that $j = \col(v)$ and $i < \row(v)$.  Let $C(\varphi; u)$ the connected component of $[N]\times[N]\setminus(S(v)\cup p_u)$ that contains $z = (\text{col}(v), N)$.

\begin{figure}[H]
\begin{center}
\begin{tikzpicture}[x=0.75pt,y=0.75pt,yscale=-1,xscale=1]
\draw  [color={rgb, 255:red, 0; green, 0; blue, 0 }  ,draw opacity=0.1 ][fill={rgb, 255:red, 0; green, 0; blue, 0 }  ,fill opacity=0.16 ] (251,171) -- (250.8,209) -- (100.33,210.33) -- (101,91) -- (160.33,90.33) -- (160.33,150.33) -- (200.33,150.33) -- (200.33,171) -- cycle ;
\draw  [color={rgb, 255:red, 0; green, 0; blue, 0 }  ,draw opacity=1 ][fill={rgb, 255:red, 255; green, 255; blue, 255 }  ,fill opacity=1 ] (230.9,100.87) -- (250.8,100.87) -- (250.8,119.67) -- (230.9,119.67) -- cycle ;
\draw  [color={rgb, 255:red, 0; green, 0; blue, 0 }  ,draw opacity=1 ][fill={rgb, 255:red, 230; green, 230; blue, 230 }  ,fill opacity=0.71 ] (250.8,119.67) -- (341,119.67) -- (341,209) -- (250.8,209) -- cycle ;
\draw  [color={rgb, 255:red, 0; green, 0; blue, 0 }  ,draw opacity=1 ][fill={rgb, 255:red, 255; green, 255; blue, 255 }  ,fill opacity=1 ] (251,160.5) -- (271.28,160.5) -- (271.28,179) -- (251,179) -- cycle ;
\draw  [color={rgb, 255:red, 0; green, 0; blue, 0 }  ,draw opacity=1 ][fill={rgb, 255:red, 255; green, 255; blue, 255 }  ,fill opacity=1 ] (230.9,190.2) -- (250.8,190.2) -- (250.8,209) -- (230.9,209) -- cycle ;
\draw    (200.33,171) -- (248,171) ;
\draw [shift={(251,171)}, rotate = 180] [fill={rgb, 255:red, 0; green, 0; blue, 0 }  ][line width=0.08]  [draw opacity=0] (8.93,-4.29) -- (0,0) -- (8.93,4.29) -- cycle    ;
\draw    (200.33,150.33) -- (200.33,171) ;
\draw    (160.33,150.33) -- (200.33,150.33) ;
\draw    (160.33,90.33) -- (160.33,150.33) ;
\draw (235.5,106.5) node [anchor=north west][inner sep=0.75pt]    {$v$};
\draw (256,165.5) node [anchor=north west][inner sep=0.75pt]    {$u$};
\draw (236,195) node [anchor=north west][inner sep=0.75pt]    {$z$};
\draw (165.33,107.73) node [anchor=north west][inner sep=0.75pt]    {$p_{u}( \varphi )$};
\draw (121.33,166.4) node [anchor=north west][inner sep=0.75pt]    {$C( \varphi ;u)$};
\end{tikzpicture}
\end{center}
\caption{Construction of $C(\varphi;u)$}
\end{figure}

Now let us define the set $C(u,v)$ as the set of particles that are placed in $C(\varphi;u)$ after $v$ and before $u$ according to $\varphi$. More precisely:
$$C(u, v) = \{w\in P(C(\varphi; u)): \varphi(v) < \varphi(w) < \varphi(u) \},$$

Finally, let us define the realization $\tilde{\varphi}$ as follows: 
\begin{itemize}
\item First set all elements in $F_v\setminus\{v\}$preserving the order given by $\varphi$.
\item Then, set all the elements of $C(u,v)$ preserving the order given by $\varphi$. 
\item Then, set $u$ followed by $v$.
\item Finally set the rest of the elements of $F$ preserving the order given by $\varphi$. 
\end{itemize}

Once again let us prove that $\tilde{\varphi}$ is in fact a realization of $F$. Let $w\in F$ be a particle. If $\tilde{\varphi}(w) < \varphi(v)$ or $\tilde{\varphi}(w) > \varphi(u)$, then it is clear that $p_w$ is still available in $\tilde{\varphi}$, because $F_w(\tilde{\varphi}) = F_w(\varphi)$. For the particles $w\in C(u,v)$ we have that $F_w(\tilde{\varphi})\subseteq F_w(\varphi)$ and then the path $p_w$ is also available in $\tilde{\varphi}$. The same argument goes for $w = u$.

For the remaining options, let us define
$$X = \{w\in F: \varphi(v)\leq\varphi(w)<\varphi(u)\} \setminus C(u,v).$$
This set is not empty because it contains $v$. Let us take $w\in X$, and suppose that no coordinate in $p_w$ is adjacent to $C(u,v)$. In this case $p_w$ is also available in $\tilde{\varphi}$. Suppose then that $p_w$ intersects $C(u,v)$. Since $w$ is not contained in $C(u,v)$, necessarily $p_w$ intersects $p_u$ (because $p_u$ is the frontier of the component $C(\varphi;u)$). Let $y\in p_u\cap p_{w}$ be the coordinate of the last time $p_u$ intersects $p_{w}$. Then define the trajectory $\tilde{p}_w$ of $w$ in $\tilde{\varphi}$ as $p_1$ and $p_2$, where $p_1$ equals to $p_u$ from the top of the grid up until reaching $y$ and $p_2$ is equal to $p_w$ from $y$ until reaching $w$. We obtain that $\tilde{p}_w$ is an available trajectory for $w$ in the realization $\tilde{\varphi}$. 
We deduce that $\tilde{\varphi}$ is a realization of $F$.

To reach a contradiction, now we show that $\tilde{\varphi}$ has a lower number of bad pairs than $\varphi$. Obviously $(v,u)$ is not a bad pair of $\tilde{\varphi}$. Let us suppose that there is a bad pair $(p,q)$ for $\tilde{\varphi}$, that is not a bad pair for $\varphi$. The only possibility for this to happen is that
$$q\in C(u,v)\cup\{u\} \quad \textrm{and}\quad p\in\{z\in F: z\not\in C(u,v) \ \wedge \ \varphi(v)\leq\varphi(z)<\varphi(v)\}.$$

In this case we would obtain a bad pair for $\varphi$ where $\col(q) < \col(v)$, which is not possible due to the choice of $v$. We conclude that in this case, $\tilde{\varphi}$ is a realization that has strictly less number of bad pairs than $\varphi$, which is a contradiction with the existence of $\varphi$.
\end{proof}

This Lemma allows us to better structure the proof that 2-DLA constructible figures can be characterized through canonical realizations.
Given a particle $p\in F$, we define the particles scope as
$$L(p) = \{(i,j): \row(p)< i \ \wedge \  \col(p) < j\} = \{q\in F: p\in S(q)\}.$$

\begin{figure}[H]
\begin{center}
\begin{tikzpicture}[x=0.75pt,y=0.75pt,yscale=-1,xscale=1]
\draw  [color={rgb, 255:red, 0; green, 0; blue, 0 }  ,draw opacity=1 ][fill={rgb, 255:red, 255; green, 255; blue, 255 }  ,fill opacity=1 ] (240.57,120.37) -- (260.47,120.37) -- (260.47,139.17) -- (240.57,139.17) -- cycle ;
\draw  [color={rgb, 255:red, 0; green, 0; blue, 0 }  ,draw opacity=1 ][fill={rgb, 255:red, 155; green, 155; blue, 155 }  ,fill opacity=0.67 ] (260.47,139.17) -- (350.67,139.17) -- (350.67,228.5) -- (260.47,228.5) -- cycle ;
\draw  [dash pattern={on 4.5pt off 4.5pt}]  (260.47,120.37) -- (352,119.83) ;
\draw  [dash pattern={on 4.5pt off 4.5pt}]  (240.57,139.17) -- (241.33,229.17) ;
\draw  [color={rgb, 255:red, 0; green, 0; blue, 0 }  ,draw opacity=1 ][fill={rgb, 255:red, 155; green, 155; blue, 155 }  ,fill opacity=0.67 ] (150.37,31.03) -- (240.57,31.03) -- (240.57,120.37) -- (150.37,120.37) -- cycle ;
\draw  [dash pattern={on 4.5pt off 4.5pt}]  (259.7,30.37) -- (260.47,120.37) ;
\draw  [dash pattern={on 4.5pt off 4.5pt}]  (149.03,139.7) -- (240.57,139.17) ;
\draw (245,125) node [anchor=north west][inner sep=0.75pt]    {$u$};
\draw (287.17,170.57) node [anchor=north west][inner sep=0.75pt]    {$S( u)$};
\draw (176.17,63.57) node [anchor=north west][inner sep=0.75pt]    {$L( u)$};
\end{tikzpicture}
\end{center}
\caption{The shadow $S(v)$ and scope $L(v)$ of a particle $v$}
\end{figure}

With this at hand, we reinterpret Lemma \ref{canon} as follows: a figure $F$ is 2DLA-Realizable if and only if there exists a canonical realization $\varphi$, that is a realization such that, for every $u\in F$:
$$\varphi(z) < \varphi(u) < \varphi(x), \ \forall z\in S(u), \ \forall x\in L(u).$$

In the following we only focus on canonical realizations. Because of previous result, when constructing a figure and want to add a particle $x = (i,j)$ we can guarantee that there will be an interrupted path from the top of the grid to the coordinate $(i-2,j-2)\in L(x)$.

\begin{figure}[H]
\begin{center}
\begin{tikzpicture}[x=0.75pt,y=0.75pt,yscale=-1,xscale=1]
\draw  [color={rgb, 255:red, 0; green, 0; blue, 0 }  ,draw opacity=1 ][fill={rgb, 255:red, 255; green, 255; blue, 255 }  ,fill opacity=1 ] (240.57,178.37) -- (260.47,178.37) -- (260.47,197.17) -- (240.57,197.17) -- cycle ;
\draw  [color={rgb, 255:red, 255; green, 255; blue, 255 }  ,draw opacity=0 ][fill={rgb, 255:red, 155; green, 155; blue, 155 }  ,fill opacity=0.67 ] (260.47,197.17) -- (300,197.17) -- (300,239.75) -- (260.47,239.75) -- cycle ;
\draw  [dash pattern={on 4.5pt off 4.5pt}]  (260.47,178.37) -- (299,178.25) ;
\draw  [dash pattern={on 4.5pt off 4.5pt}]  (240.57,197.17) -- (240.5,241.75) ;
\draw  [color={rgb, 255:red, 0; green, 0; blue, 0 }  ,draw opacity=0 ][fill={rgb, 255:red, 155; green, 155; blue, 155 }  ,fill opacity=0.67 ] (150.37,89.03) -- (240.57,89.03) -- (240.57,178.37) -- (150.37,178.37) -- cycle ;
\draw  [dash pattern={on 4.5pt off 4.5pt}]  (259.7,88.37) -- (260.47,178.37) ;
\draw  [dash pattern={on 4.5pt off 4.5pt}]  (149.03,197.7) -- (240.57,197.17) ;
\draw  [dash pattern={on 4.5pt off 4.5pt}]  (260.47,197.17) -- (299,197.05) ;
\draw  [dash pattern={on 4.5pt off 4.5pt}]  (260.47,197.17) -- (260.4,241.75) ;
\draw  [dash pattern={on 4.5pt off 4.5pt}]  (239.8,88.37) -- (240.57,178.37) ;
\draw  [dash pattern={on 4.5pt off 4.5pt}]  (149.03,178.9) -- (240.57,178.37) ;
\draw  [color={rgb, 255:red, 0; green, 0; blue, 0 }  ,draw opacity=1 ][fill={rgb, 255:red, 255; green, 255; blue, 255 }  ,fill opacity=0.87 ] (200.77,140.77) -- (220.67,140.77) -- (220.67,159.57) -- (200.77,159.57) -- cycle ;
\draw    (150.37,89.03) .. controls (152.72,89.04) and (153.89,90.23) .. (153.88,92.59) .. controls (153.86,94.95) and (155.03,96.14) .. (157.39,96.15) .. controls (159.75,96.16) and (160.92,97.35) .. (160.9,99.71) .. controls (160.89,102.07) and (162.06,103.26) .. (164.42,103.27) .. controls (166.77,103.28) and (167.94,104.47) .. (167.93,106.82) .. controls (167.91,109.18) and (169.08,110.37) .. (171.44,110.38) .. controls (173.8,110.39) and (174.97,111.58) .. (174.96,113.94) .. controls (174.94,116.3) and (176.11,117.49) .. (178.47,117.5) .. controls (180.83,117.51) and (182,118.7) .. (181.98,121.06) .. controls (181.96,123.42) and (183.13,124.61) .. (185.49,124.62) .. controls (187.85,124.63) and (189.02,125.81) .. (189.01,128.17) .. controls (188.99,130.53) and (190.16,131.72) .. (192.52,131.73) .. controls (194.88,131.74) and (196.05,132.93) .. (196.03,135.29) .. controls (196.01,137.65) and (197.18,138.84) .. (199.54,138.85) -- (202.99,142.34) -- (208.61,148.03) ;
\draw [shift={(210.72,150.17)}, rotate = 225.37] [fill={rgb, 255:red, 0; green, 0; blue, 0 }  ][line width=0.08]  [draw opacity=0] (6.25,-3) -- (0,0) -- (6.25,3) -- cycle    ;

\draw (245,183) node [anchor=north west][inner sep=0.75pt]    {$x$};
\end{tikzpicture}

\end{center}
\caption{We assume that for every occupied site $x=(i,j)$, when a particle is placed in $x$ there is an available path that starts in $(1,1)$ and reaches $(i-2,j-2)$.}
\end{figure}

This means that, to build a realization of a figure, it is sufficient to select one that satisfies that, for each particle $(i,j)$, there is a path from $(i-2, j-2)$ to $(i,j)$ that avoids the collision with other particles already fixed. For a particle $x \in F$, we define $V(x)$ the vicinity of $x$ as the set of sites depicted on Figure~\ref{vecindad}. More formally, if $x=(i,j)$ then $V(x) = V_1(x) \cup V_2(x) \cup V_3(x)  \cup V_4(x)$ where
$$V_1((i,j)) = \{(i-3,j), (i-2,j), (i-1,j), (i-2,j+1), (i-1,j+1)\},$$ 
$$V_2((i,j)) = \{(i,j-3), (i,j-2), (i, j-1), (i+1,j-2), (i+1, j-1)\},$$ 

\begin{figure}[H]
\label{vecindad}
\begin{center}

\tikzset{every picture/.style={line width=0.75pt}} 

\begin{tikzpicture}[x=0.75pt,y=0.75pt,yscale=-1,xscale=1]

\draw   (270,150) -- (350,150) -- (350,170) -- (270,170) -- cycle ;
\draw   (290,150) -- (330,150) -- (330,190) -- (290,190) -- cycle ;
\draw    (310,150) -- (310,190) ;
\draw  [fill={rgb, 255:red, 155; green, 155; blue, 155 }  ,fill opacity=1 ] (330,150) -- (350,150) -- (350,170) -- (330,170) -- cycle ;
\draw   (330,90) -- (350,90) -- (350,170) -- (330,170) -- cycle ;
\draw   (330,110) -- (370,110) -- (370,150) -- (330,150) -- cycle ;
\draw    (330,130) -- (370,130) ;
\draw   (80,150) -- (160,150) -- (160,170) -- (80,170) -- cycle ;
\draw   (100,150) -- (140,150) -- (140,190) -- (100,190) -- cycle ;
\draw    (120,150) -- (120,190) ;
\draw  [fill={rgb, 255:red, 155; green, 155; blue, 155 }  ,fill opacity=1 ] (140,150) -- (160,150) -- (160,170) -- (140,170) -- cycle ;
\draw   (140,90) -- (160,90) -- (160,170) -- (140,170) -- cycle ;
\draw   (140,110) -- (180,110) -- (180,150) -- (140,150) -- cycle ;
\draw    (140,130) -- (180,130) ;
\draw  [fill={rgb, 255:red, 208; green, 2; blue, 27 }  ,fill opacity=0.33 ] (270,150) -- (290,150) -- (290,170) -- (270,170) -- cycle ;
\draw  [fill={rgb, 255:red, 208; green, 2; blue, 27 }  ,fill opacity=0.33 ] (290,150) -- (330,150) -- (330,190) -- (290,190) -- cycle ;
\draw  [fill={rgb, 255:red, 208; green, 2; blue, 27 }  ,fill opacity=0.33 ] (140,90) -- (160,90) -- (160,110) -- (140,110) -- cycle ;
\draw  [fill={rgb, 255:red, 208; green, 2; blue, 27 }  ,fill opacity=0.33 ] (140,110) -- (180,110) -- (180,150) -- (140,150) -- cycle ;

\draw (335,155) node [anchor=north west][inner sep=0.75pt]  [color={rgb, 255:red, 255; green, 255; blue, 255 }  ,opacity=1 ]  {$x$};
\draw (145,155) node [anchor=north west][inner sep=0.75pt]  [color={rgb, 255:red, 255; green, 255; blue, 255 }  ,opacity=1 ]  {$x$};
\draw (311,202.4) node [anchor=north west][inner sep=0.75pt]    {$V_{2}( x)$};
\draw (127,202.4) node [anchor=north west][inner sep=0.75pt]    {$V_{1}( x)$};

\end{tikzpicture}

\end{center}
\caption{Visual representation of $V(x)$ and sets $V_1(x)$ and $V_2(x)$ highlighted in red}
\end{figure}

Due to the nature of the dynamics, we have to consider only four different possible paths a particle can take to reach $x$ from $y = (\row(x) -2, \col(x) - 2)$. We denote these paths as $Q_1, Q_2, Q_3$ and $Q_4$, which are shown in the following figure.

\begin{figure}[H]
\begin{center}

\tikzset{every picture/.style={line width=0.75pt}} 

\begin{tikzpicture}[x=0.75pt,y=0.75pt,yscale=-1,xscale=1]
\draw   (90,150) -- (170,150) -- (170,170) -- (90,170) -- cycle ;
\draw   (110,110) -- (150,110) -- (150,190) -- (110,190) -- cycle ;
\draw    (130,110) -- (130,190) ;
\draw  [fill={rgb, 255:red, 155; green, 155; blue, 155 }  ,fill opacity=1 ] (150,150) -- (170,150) -- (170,170) -- (150,170) -- cycle ;
\draw  [fill={rgb, 255:red, 208; green, 2; blue, 27 }  ,fill opacity=0.33 ] (150,90) -- (170,90) -- (170,110) -- (150,110) -- cycle ;
\draw  [fill={rgb, 255:red, 208; green, 2; blue, 27 }  ,fill opacity=0.33 ] (150,110) -- (190,110) -- (190,150) -- (150,150) -- cycle ;
\draw    (110,130) -- (190,130) ;
\draw  [fill={rgb, 255:red, 155; green, 155; blue, 155 }  ,fill opacity=1 ] (110,110) -- (130,110) -- (130,130) -- (110,130) -- cycle ;
\draw   (310,150) -- (390,150) -- (390,170) -- (310,170) -- cycle ;
\draw   (330,110) -- (370,110) -- (370,190) -- (330,190) -- cycle ;
\draw    (350,110) -- (350,190) ;
\draw  [fill={rgb, 255:red, 155; green, 155; blue, 155 }  ,fill opacity=1 ] (370,150) -- (390,150) -- (390,170) -- (370,170) -- cycle ;
\draw   (370,90) -- (390,90) -- (390,170) -- (370,170) -- cycle ;
\draw   (370,110) -- (410,110) -- (410,150) -- (370,150) -- cycle ;
\draw    (330,130) -- (410,130) ;
\draw  [fill={rgb, 255:red, 155; green, 155; blue, 155 }  ,fill opacity=1 ] (330,110) -- (350,110) -- (350,130) -- (330,130) -- cycle ;
\draw    (170,110) -- (170,150) ;
\draw  [fill={rgb, 255:red, 208; green, 2; blue, 27 }  ,fill opacity=0.33 ] (370,110) -- (390,110) -- (390,150) -- (370,150) -- cycle ;
\draw  [fill={rgb, 255:red, 208; green, 2; blue, 27 }  ,fill opacity=0.33 ] (390,130) -- (410,130) -- (410,150) -- (390,150) -- cycle ;
\draw  [fill={rgb, 255:red, 208; green, 2; blue, 27 }  ,fill opacity=0.33 ] (350,150) -- (370,150) -- (370,170) -- (350,170) -- cycle ;
\draw   (90,300) -- (170,300) -- (170,320) -- (90,320) -- cycle ;
\draw   (110,260) -- (150,260) -- (150,340) -- (110,340) -- cycle ;
\draw    (130,260) -- (130,340) ;
\draw  [fill={rgb, 255:red, 155; green, 155; blue, 155 }  ,fill opacity=1 ] (150,300) -- (170,300) -- (170,320) -- (150,320) -- cycle ;
\draw  [fill={rgb, 255:red, 255; green, 255; blue, 255 }  ,fill opacity=1 ] (150,240) -- (170,240) -- (170,260) -- (150,260) -- cycle ;
\draw  [fill={rgb, 255:red, 255; green, 255; blue, 255 }  ,fill opacity=1 ] (150,260) -- (190,260) -- (190,300) -- (150,300) -- cycle ;
\draw    (110,280) -- (190,280) ;
\draw  [fill={rgb, 255:red, 155; green, 155; blue, 155 }  ,fill opacity=1 ] (110,260) -- (130,260) -- (130,280) -- (110,280) -- cycle ;
\draw   (310,300) -- (390,300) -- (390,320) -- (310,320) -- cycle ;
\draw   (330,260) -- (370,260) -- (370,340) -- (330,340) -- cycle ;
\draw    (350,260) -- (350,340) ;
\draw  [fill={rgb, 255:red, 155; green, 155; blue, 155 }  ,fill opacity=1 ] (370,300) -- (390,300) -- (390,320) -- (370,320) -- cycle ;
\draw   (370,240) -- (390,240) -- (390,320) -- (370,320) -- cycle ;
\draw   (370,260) -- (410,260) -- (410,300) -- (370,300) -- cycle ;
\draw    (330,280) -- (410,280) ;
\draw  [fill={rgb, 255:red, 155; green, 155; blue, 155 }  ,fill opacity=1 ] (330,260) -- (350,260) -- (350,280) -- (330,280) -- cycle ;
\draw    (170,260) -- (170,300) ;
\draw  [fill={rgb, 255:red, 208; green, 2; blue, 27 }  ,fill opacity=0.33 ] (330,300) -- (370,300) -- (370,320) -- (330,320) -- cycle ;
\draw  [fill={rgb, 255:red, 208; green, 2; blue, 27 }  ,fill opacity=0.33 ] (370,280) -- (390,280) -- (390,300) -- (370,300) -- cycle ;
\draw  [fill={rgb, 255:red, 208; green, 2; blue, 27 }  ,fill opacity=0.33 ] (90,300) -- (110,300) -- (110,320) -- (90,320) -- cycle ;
\draw  [fill={rgb, 255:red, 208; green, 2; blue, 27 }  ,fill opacity=0.33 ] (110,300) -- (150,300) -- (150,340) -- (110,340) -- cycle ;
\draw [color={rgb, 255:red, 208; green, 2; blue, 27 }  ,draw opacity=1 ]   (120,280) -- (120,310) -- (147,310) ;
\draw [shift={(150,310)}, rotate = 180] [fill={rgb, 255:red, 208; green, 2; blue, 27 }  ,fill opacity=1 ][line width=0.08]  [draw opacity=0] (6.25,-3) -- (0,0) -- (6.25,3) -- cycle    ;
\draw [color={rgb, 255:red, 208; green, 2; blue, 27 }  ,draw opacity=1 ]   (350,120) -- (360,120) -- (360,140) -- (380,140) -- (380,147) ;
\draw [shift={(380,150)}, rotate = 270] [fill={rgb, 255:red, 208; green, 2; blue, 27 }  ,fill opacity=1 ][line width=0.08]  [draw opacity=0] (6.25,-3) -- (0,0) -- (6.25,3) -- cycle    ;
\draw [color={rgb, 255:red, 208; green, 2; blue, 27 }  ,draw opacity=1 ]   (130,120) -- (160,120) -- (160,147) ;
\draw [shift={(160,150)}, rotate = 270] [fill={rgb, 255:red, 208; green, 2; blue, 27 }  ,fill opacity=1 ][line width=0.08]  [draw opacity=0] (6.25,-3) -- (0,0) -- (6.25,3) -- cycle    ;
\draw  [fill={rgb, 255:red, 208; green, 2; blue, 27 }  ,fill opacity=0.33 ] (350,320) -- (370,320) -- (370,340) -- (350,340) -- cycle ;
\draw [color={rgb, 255:red, 208; green, 2; blue, 27 }  ,draw opacity=1 ]   (340,280) -- (340,290) -- (360,290) -- (360,310) -- (367,310) ;
\draw [shift={(370,310)}, rotate = 180] [fill={rgb, 255:red, 208; green, 2; blue, 27 }  ,fill opacity=1 ][line width=0.08]  [draw opacity=0] (6.25,-3) -- (0,0) -- (6.25,3) -- cycle    ;

\draw (155,155) node [anchor=north west][inner sep=0.75pt]  [color={rgb, 255:red, 255; green, 255; blue, 255 }  ,opacity=1 ]  {$x$};
\draw (115,115) node [anchor=north west][inner sep=0.75pt]  [color={rgb, 255:red, 255; green, 255; blue, 255 }  ,opacity=1 ]  {$y$};
\draw (375,155) node [anchor=north west][inner sep=0.75pt]  [color={rgb, 255:red, 255; green, 255; blue, 255 }  ,opacity=1 ]  {$x$};
\draw (335,115) node [anchor=north west][inner sep=0.75pt]  [color={rgb, 255:red, 255; green, 255; blue, 255 }  ,opacity=1 ]  {$y$};
\draw (155,305) node [anchor=north west][inner sep=0.75pt]  [color={rgb, 255:red, 255; green, 255; blue, 255 }  ,opacity=1 ]  {$x$};
\draw (115,265) node [anchor=north west][inner sep=0.75pt]  [color={rgb, 255:red, 255; green, 255; blue, 255 }  ,opacity=1 ]  {$y$};
\draw (375,305) node [anchor=north west][inner sep=0.75pt]  [color={rgb, 255:red, 255; green, 255; blue, 255 }  ,opacity=1 ]  {$x$};
\draw (335,265) node [anchor=north west][inner sep=0.75pt]  [color={rgb, 255:red, 255; green, 255; blue, 255 }  ,opacity=1 ]  {$y$};
\draw (130,200) node [anchor=north west][inner sep=0.75pt]    {$Q_{1}$};
\draw (350,200) node [anchor=north west][inner sep=0.75pt]    {$Q_{2}$};
\draw (130,350) node [anchor=north west][inner sep=0.75pt]    {$Q_{3}$};
\draw (350,350) node [anchor=north west][inner sep=0.75pt]    {$Q_{4}$};
\end{tikzpicture}
\end{center}
\caption{ A representation of the four choices of paths starting from $y = (i-2,j-2)$ and reaching coordinate $x = (i,j)$.}
\end{figure}

More formally, in Table \ref{table:defcaminos} are given the four possible choices of paths starting from $(i-2, j-2)$ and reaching coordinate $(i,j)$,  with the list of visited cells and the cells adjacent to each path. For $k\in \{1,2,3,4\}$ we call $Q_k((i,j))$ the union of the set of visited and adjacent cells of path $Q_k((i,j))$. Given a realization $\varphi$ of $F$, we say that a path $Q_k$ is \emph{available} if all cells $c$ in the second of and third columns of the $k$-th row of the table are satisfy that $\varphi(c) > \varphi((i,j))$. Otherwise, we say that $Q_k$ is \emph{unavailable}.  

\renewcommand{\arraystretch}{2}

\begin{table}[H]
\label{table:defcaminos}
\begin{center}
\begin{tabular}{c|c|c|}
Path & Sequence of visited cells & Adjacent cells in $V(x)$\\ \hline 
$Q_1((i,j))$ &  \parbox[c]{4.5cm}{\centering $(i-2, j-2), (i-2, j-1), $ \\ $(i-2, j), (i-1, j) $}    & \parbox[c]{4.5cm}{\centering $(i-3,j),(i-2,j+1),$ \\ $(i-1,j+1)  $} \\ \hline
$Q_2((i,j))$ &  \parbox[c]{4.5cm}{\centering $(i-2, j-2), (i-2, j-1), $ \\ $(i-1, j-1), (i-1, j)$} & \parbox[c]{4.5cm}{\centering $(i-2,j),(i-1,j+1),$\\ $(i,j-1)$} \\ \hline
$Q_3((i,j))$ &  \parbox[c]{4.5cm}{\centering $(i-2, j-2), (i-1, j-2), $ \\ $(i, j-2), (i, j-1) $} & \parbox[c]{4.5cm}{\centering $(i,j-3),(i+1,j-1),$\\ $(i+1,j-1)$} \\ \hline
$Q_4((i,j))$ &  \parbox[c]{4.5cm}{\centering $(i-2, j-2), (i-1, j-2), $ \\ $(i-1, j-1), (i, j-1)$} & \parbox[c]{4.5cm}{\centering $(i,j-2),(i+1,j-1),$\\ $(i-1,j)$} \\ \hline
\end{tabular}
\end{center}

\caption{Definition of possible choices of paths starting from $(i-2, j-2)$ and reaching coordinate $(i,j)$,  with the list of visited cells and the cells adjacent to each path. In a realization, we say that a path $Q_i$ is \emph{available} if all cells in the second of and third columns of the $i$-th row are empty.}
\end{table}

\renewcommand{\arraystretch}{1}

\begin{remark} There exist two other possible choices for paths reaching cell $(i,j)$ from cell $(i-2,j-2)$, however, namely the paths $(i-2, j-2), (i-2, j-1), (i-1, j-1), (i, j-1), (i, j)$ and $(i-2, j-2), (i-1, j-2), (i-1, j-1), (i-1, j), (i, j)$. However, these paths require have more available cells than $Q_2$ and $Q_3$, respectively. Therefore, any trajectory that uses any of such paths have $Q_2$ or $Q_4$ available. 
\end{remark}

\subsubsection{The dependency graph}

We are now ready to define the \emph{dependency graph} of a figure. The dependency graph $H$ of a figure $F$ is a directed graph satisfying that, if $u,v \in F$ and $(u,v)\in E(H)$, then $u$ must be fixed before $v$ on every canonical realization of $v$. To precisely define the set of edges, we apply a number of applications of a set of \emph{rules}, described below. 

First, let us denote $\ground$ as the set $\{N+1\}\times [N]$, that is to say, the set of particles in the bottom of the figure. The vertex set of the dependency graph $H$ is $F\cup \ground$. The construction starts with the set of edges $$E_0= (\ground \times F\setminus \ground) \cup \{(u,v) \in F\times F : u \in S(v)\} .$$
We start this way because all the edges in the bottom must be fixed (are fixed from the beginning) before all particles in $F$. Moreover, as we are considering only canonical realizations, and all the particles in $S(v)$ are fixed before $v$. Then, we iteratively add edges to $E_0$ according to two rules defined below, obtaining this way sets $E_1, E_2, \dots$. The process stops when no new edges can be with any of the two rules. The resulting set is the set of edges of $H$.   

For $\ell \geq 0$ and a node $v\in V(H)$, let us consider the sets
$$N(v) = \{(\row(v),\col(v)\pm 1), (\row(v) \pm 1,\col(v)) \} \cap V(H)$$
$$D^-_\ell(v) = \{w \in V(H): (w,v) \in E_\ell \},$$
$$D^+_\ell(v) = \{w \in V(H): (v,w) \in E_\ell \},$$
which are called, respectively, the sets of \emph{adjacent nodes} and the sets of \emph{predecessors} and \emph{successors} of $v$ with respect to edge set $E_\ell$.  Given $E_\ell$, we construct $E_{\ell+1}$ adding edges to $E_\ell$ according to the following two rules.\\

The first rule is used to force that a particle that is not adjacent to the ground must have at least one adjacent particle that is potentially fixed before it. This rule is applied over vertices $v\in F$ such that $D^-_\ell(v) = \emptyset$.  First, let us call $R_\ell(v)$ the set of neighbors of $v$ defined by:

$$R_\ell(v) = N(v) \setminus \left( \bigcup_{w\in D_\ell^+(v)} D_\ell^+(w) \cup D^+_\ell(v)  \right) $$
In words, $R_\ell(v)$ is the set of particles adjacent to $v$ that are not successors of $v$, neither successors of its successors. If $R_\ell(v) =\emptyset$, then all neighbors of $v$ must be fixed after $v$, and therefore the figure is not realizable. To represent that, we add edge $(v,v)$ to $E_{\ell +1}$.  Now suppose that $R_\ell(v)$ is nonempty and suppose that there exists $u \in R_\ell(v)$ such that $R_\ell(v) \setminus (D^+_\ell(u) \cap \{u\}) = \emptyset$. Then, necessarily $u$ must be a predecessor of $v$. We conclude the following definition of {\bf Rule 1}.

\begin{definition}[{\bf Rule 1}]
For each $v\in F$ such that $D_\ell^-(v)=\emptyset$:
\begin{itemize}
\item If $R_\ell(v) = \emptyset$, then add edge $(v,v)$ to $E_{\ell +1}$.
\item If $R_\ell(v) \neq \emptyset$ and there exists $u \in R_\ell(v)$ satisfying that $R_\ell(v) \setminus (D_{\ell}^+(u) \cup \{u\}) = \emptyset$ then we add $(u,v)$ to $E_{\ell +1}$. 
\end{itemize}
 \end{definition}

   \medskip
The second rule states that, when a particle is fixed there must exist an available path reaching $v$ from $(\row(v)-2, \col(v)-2)$. This implies that we must fix $v$ before the cells that block all the available paths of $v$. More precisely, for a particle $v\in F$, we call $B_\ell(v)$ the set of currently unavailable paths of $v$, formally:
$$ B_\ell(v) = \{ k \in \{1,2,3,4\}: Q_k(v) \cap D^-_\ell(v) \neq \emptyset \},$$

 The set of available paths of $v$ is $A_\ell(v) =\{1,2,3,4\}\setminus B_\ell(v)$. Now let us define $\mathcal{I}_\ell(v)$ as the set of particles in the intersection of all available paths, formally $$\mathcal{I}_\ell(v) = \bigcap_{k \in A_\ell(v)} Q_k(v).$$
Then, {\bf Rule 2} states that all particles in $\mathcal{I}_{v}$ must be fixed after $v$. We deduce the following definition.

 \begin{definition}[{\bf Rule 2}]
 
 For each $v\in F$:
 \begin{itemize}
\item If $A_\ell(v) = \emptyset$, then add edge $(v,v)$ to $E_{\ell +1}$.
\item If $u \in \mathcal{I}_{v}$, then we add $(v,u)$ to $E_{\ell +1}$. 
\end{itemize}
 
 \end{definition}

%
%
%
  \medskip

  
 Therefore, the construction of the dependency graph consists in the applications of {\bf Rule 1} and {\bf Rule 2} until no new edges are created. Since a figure contains $n$ particles, there are at most $n^2$ possible edges between them. Therefore, the construction finalizes in at most $n^2$ applications of the rules. Observe also that each rule can be applied in polynomial time. We formalize previous construction in Algorithm \ref{algo:dependency}. 

\begin{algorithm}[H]

\KwIn{A figure $F$ represented by a set of coordinates in $[N]\times [N]$}
 \KwOut{A set of edges $E$}
$ \ground \leftarrow  (\{N+1\}\times [N])$ \tcp{the cells in the bottom}\;
$E_0 = \leftarrow\{(u,v): u \in S(v)\}\cup (\ground \times (F\setminus \ground)) $\: 

\For {$ u = (i,j) \in F$}{
$N(u) \leftarrow \{(i+1,j),(i-1,j),(i,j+1),(i,j-1)\} \cap (F\cup \ground)$\;

$Q_1(u) = \{(i-3, j), (i-2,j), (i-1,j), (i-2,j+1), (i-1, j+1)\}\cap (F\cup \ground)$\;

$Q_2(u) =\{(i-2, j), (i-1, j), (i-1, j+1), (i, j-1)\}\cap (F\cup \ground) $\;

$Q_3(u) = \{(i, j-3), (i,j-2), (i,j-1), (i+1,j-2), (i+1, j-1) \}\cap (F\cup \ground)$\;

$Q_4(u) = \{(i, j-2), (i, j-1), (i+1, j-1), (i-1, j) \}\cap (F\cup \ground)$\;

}

 \For{ $\ell \in \{1, \dots n^2\}$}{
$E_{\ell +1} \leftarrow E_{\ell}$\;

\For{$x \in F$}{
$D^+(x) \leftarrow \{w \in F: (x,w)\in E_{\ell}\}$\;

$D^-(x) \leftarrow \{w \in F: (w,x)\in E_{\ell}\} $\;

$R_\ell(x) \leftarrow N(x) \setminus \left( \bigcup_{w\in D_\ell^+(x)} D_\ell^+(w) \cup D^+_\ell(x)  \right) $\;

$ B_\ell(x) \leftarrow \{ k \in \{1,2,3,4\}: D^-(x)\cap Q_k(x) \neq \emptyset\}$\;

$A_\ell(x) \leftarrow \{1,2,3,4\} \setminus B_\ell(x)$. \;

$ \mathcal{I}_\ell(x) \leftarrow \displaystyle{\bigcap_{k\in A_\ell(x)} Q_k(x) }$ \;
 }
 \For{$v \in F$}{
 
 \tcp{Rule 1}
\If{$D_\ell^{-}(v) =\emptyset$ and $R_\ell(v) = \emptyset$}{
$E_{\ell+1} \leftarrow E_{\ell+1}\cup \{(v,v)\}$\;
 }
\If{$D_\ell^{-}(v) =\emptyset$ and $R_\ell(v) \neq \emptyset$}{
   \For{$u \in R_\ell(v)$}{
\If {$R_\ell(v) \setminus (D^+(u) \cup \{u\}) = \emptyset$}{
 $E_{\ell+1} \leftarrow E_{\ell+1}\cup \{(u,v)\}$\;
  }
  }
} 

  \tcp{Rule 2}
\eIf{$A_\ell(v) = \emptyset$}{
$E_{\ell+1} \leftarrow E_{\ell+1}\cup \{(v,v)\}$\;
 }{
    \For{$u \in F \cup \ground$}{
  \If{$u \in \mathcal{I}_v$}{
 $E_{\ell+1} \leftarrow E_{\ell+1}\cup \{(v,u)\}$\;
 }
 }
 }
 }
 }
 

 \Return $E_{n^2}\setminus (\ground \times F\setminus \ground)) $\;
 \caption{Constructing the edges of the dependency graph of figure $F$}
 \label{algo:dependency}
\end{algorithm}

  \begin{definition}  
The \emph{dependency graph} of a figure $F$ is the graph with vertex set $F$ and set of edges obtained as the output of Algorithm \ref{algo:dependency} on input $F$. 
\end{definition}

In Figure \ref{fig:dependecy_example1} is depicted an example of a figure $F$ and its corresponding dependency graph.

\begin{figure}
\begin{center}
\includegraphics[width=0.4\textwidth]{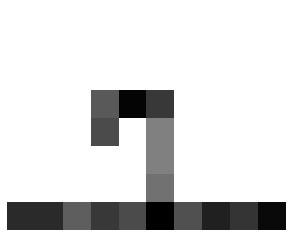}
\includegraphics[width=0.5\textwidth]{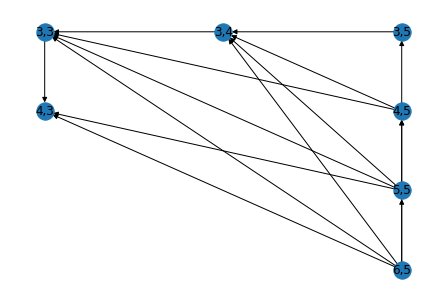}
\includegraphics[width=0.4\textwidth]{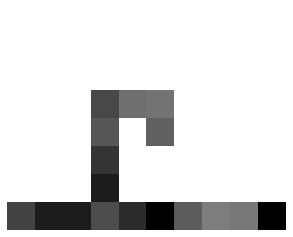}
\includegraphics[width=0.5\textwidth]{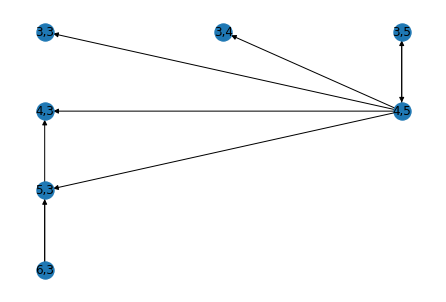}
\includegraphics[width=0.4\textwidth]{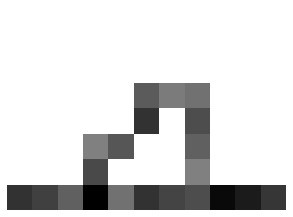}
\includegraphics[width=0.5\textwidth]{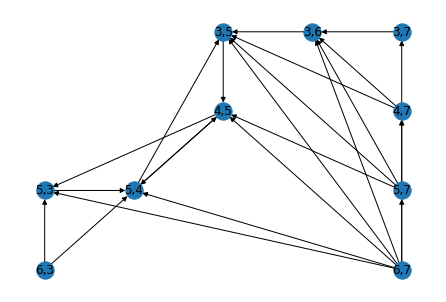}
\end{center}
\caption{Three examples of figures and their corresponding  dependency graphs. Observe that the middle and bottom figures are not 2DLA-realizable, and their dependency graphs are not acyclic.}
\label{fig:dependecy_example1}
\end{figure}

\begin{lemma}
\label{lem:dependency}
Let $H = (F,E)$ be the dependency graph of a 2DLA-realizable figure $F$. For every pair of particles $u,v \in F$, if $(u,v)\in E$ then $u$ is fixed before $v$ on every canonical realization of $F$.
\end{lemma}

\begin{proof}
We prove the result by induction on $\ell$, showing that every edge in $E_{\ell}$ have the above property. The base case $\ell = 0$ is trivial because $E_0$ contains only the edges $(u,v)$ such that $u \in S(v)$. Now suppose that for a given $\ell\geq 0$ all the edges in $(u,v)$ satisfy that $u$ is fixed before $v$ on every canonical realization of $F$. If $E_{\ell +1} = E_{\ell}$ then $E(H) = E_{\ell}$ and we are done. Therefore, suppose that $E_{\ell + 1}$ contains an edge $(u,v)$ not in $E_{\ell}$ and assume that there exists a canonical realization $\varphi$ that fixes $v$ before $u$. We will reach a contradiction  by showing that $(u,v)$ can not created with {\bf Rule 1} or {\bf Rule 2}.

First, as $v$ can fixed before $u$ in $\varphi$, there must exist a neighbor $w \in N(v)$ that is fixed before $v$ and $u$. From the induction hypothesis, this $w$ must be contained in $R_\ell(v)$ and moreover, $w \notin D_\ell^+(u)\cup \{u\} \cup D_\ell^+(v)$. Then $(u,v)$ was not created according to {\bf Rule 1}.  Second, since $u$ can be fixed after $v$ in $\varphi$, there must exist an available path $k\in A_\ell(u)$ such that $v \notin Q_k(u)$. We deduce that $(u,v)$ can not be created according to {\bf Rule 2}. We deduce that, if $(u,v)$ is an edge of $E_{\ell+1}$ then $u$ must be fixed before $v$ on every canonical realization of $F$.  
\end{proof}

Our approach consists in constructing a canonical realization of $F$ by removing one by one the particles of $F$, applying the following recursive algorithm. First, if $|F|= 1$, i.e., the input figure consists in a single particle, the algorithm returns a realization consisting in this single particle. If $|F|>1$, we compute the set of edges $E$ of the dependency graph of $F$. Then, we look for a \emph{removable particle} $u$, defined as follows.

\begin{definition}\label{def:dependency} A removable particle is a particle satisfying the following conditions.

\begin{itemize}
\item $C(u) = \{ k \in \{1,2,3,4\}: (F \cup \ground) \cap Q_k(u) = \emptyset\} \neq \emptyset$, meaning that there is at least one available path  $u$  not touching any particle in $F \cup \ground$. 
\item $D^+(u) =  \{w \in F: (t,w)\in E\} = \emptyset$, meaning that no other particle needs to be fixed after $u$. 
\item The dependency graph of $F\setminus \{u\}$ is acyclic.
\end{itemize}
\end{definition}

If no such particle $u$ is found, the algorithm rejects. Otherwise, we recursively run the algorithm in $F\setminus \{u\}$. Finally, if the recursive call of the algorithm does not reject, the output is a realization $\tilde{\varphi}$ of $F\setminus \{u\}$. The realization of $F$ is obtained by fixing $u$ after all the particles in $F\setminus\{u\}$ according to $\tilde{\varphi}$, obtaining this way a canonical realization $\varphi$. The pseudocode of algorithm is written in Algorithm \ref{algo:construction}.

\begin{algorithm}[h]

\KwIn{A figure $F \subseteq [N]\times[N] $ of size $n$ with an acyclic dependency graph.}
 \KwOut{A realization of $F$.}
 
\eIf{$|F| = \{u\}$}{
\Return a list $\varphi$ containing $u$ as a single element.\;
}{
Compute $E$ as the output of Algorithm \ref{algo:dependency} on input $F$\;

\For{$x\in F$}{

$D^+(x) \leftarrow \{w \in F: (x,w)\in E\}$\;

$ C_x \leftarrow \{ k \in \{1,2,3,4\}: F\cap Q_k(x) \neq \emptyset\}$\;

 } 
 $\alpha \leftarrow \{u \in F:  C_x \neq \emptyset\}$\;
 
$\beta \leftarrow \{u \in F:  D^+(u) = \emptyset\}$\;

 $\gamma \leftarrow \alpha\cap \beta$\;
 
 \For{$u \in \gamma$}{
 Compute $E'$ as the output of Algorithm \ref{algo:dependency} on input $F\setminus \{u\}$\;
 
\If{ $H=(F\setminus\{u\}, E')$ contains a cycle}{
 $\gamma \leftarrow \gamma \setminus \{u\}$}
}
 Pick any $u\in \gamma$\;
 
 $\tilde{\varphi} \leftarrow$ output of Algorithm \ref{algo:construction} on input $F\setminus \{u\}$\;
 
 $\varphi \leftarrow $ list $ \tilde{\varphi}$  adding $u$ in the end. \;
 
\Return $\varphi$ \;
}

 \caption{Constructing the realization of $F$}
 \label{algo:construction}
\end{algorithm}

\begin{lemma}\label{lem:tech1} Let $F$ be a figure with an acyclic  dependency graph $H$, then $H$ contains a removable particle.
\end{lemma}

\begin{proof}

Let $F$ be a figure such that its corresponding dependency graph $H$ is acyclic. Let $W$ be the set of particles with out-degree zero in $H$, formally $W = \{v \in F: D^+(v) = \emptyset\}$. This set is non-empty as $H$ is acyclic. Now let us pick the particle in $u \in W$ such that $\col(u)$ is minimum, and over all possible choices, we pick the one such that $\row(u)$ is maximum. In other words, $u$ is the bottom-left most particle in $W$.

We are going to show that $u$ is a removable particle, by showing a series of claims. We define before the following four sets relatively to $u$:

$$S_{NE}(u) = \{w \in F: \row(w)\leq \row(u) \textrm{ and } \col(w) \geq \col(u)\}\setminus \{u\}$$
$$S_{NW}(u) = \{w \in F: \row(w)< \row(u) \textrm{ and } \col(w) < \col(u)\}$$
$$S_{SE}(u) = \{w \in F: \row(w)> \row(u) \textrm{ and } \col(w) > \col(u)\}$$
$$S_{SW}(u) = \{w \in F: \row(w)\geq \row(u) \textrm{ and } \col(w) \leq \col(u)\} \setminus \{u\}$$

Observe that $S_{NW}(u)$ is the scope of $u$ and $S_{SE}(u)$ is the shadow of $u$. Observe that the scope of $u$ is empty, because $u$ has out degree zero in $H$. 

\medskip

\noindent {\bf Claim 1:} All nodes in $S_{SW}(u)$ are ancestors of $u$ in $H$. Formally, for every $w \in S_{SW}(u)$ there is a $w,u$-directed path in $H$. \\

\textit{Proof of Claim 1}. Let $w_0$ an arbitrary particle in $S_{SW}(u)$. From the choice of $u$, we know that $w_0$ has at least one out-neighbor in $H$. Let $P = w_0, w_1, \dots, w_k$ be a longest directed path in $H$ starting from $w_0$. Then $w_k$ has out degree zero, implying that it is not contained in $S_{SW}(u)$. Let $w_\ell$, with $0<\ell \leq k$ be the first node in $P$ not contained in $S_{SW}(u)$. If $w_\ell = u$ we are done. If $w_\ell$ belongs to $S_{SE}(u)$ then $(w_ell,u) \in E(H)$ and  $P' = w_0, \dots, w_\ell, u$ is a directed path in $H$. Finally, suppose that $w_\ell$ belongs to $S_{NE}$. Then necessarily $(w_{\ell-1}, w_\ell)$ is created by {\bf Rule 2} and $w_{\ell}$ belongs to $V_1(w_{\ell-1})$. Observe that $u$ is also contained in $V_1(w_{\ell-1})$ and then $(w_{\ell-1}, u)$ is also an edge created by {\bf Rule 2}. We conclude that $P'' = w_0, \dots, w_{\ell-1}, u$ is a directed path of $H$. 

\medskip

\noindent {\bf Claim 2:} $u$ has an available path in $F$, i.e. $C(u) \neq \emptyset$. \\

\textit{Proof of Claim 2}. Suppose that $Q_3(u)$ intersects $F$ and pick $w \in Q_3(u) \cap F$. Observe that  $w\in S_{SW}(u)$. Therefore, by {\bf Claim 1} there is a directed path from $w$ to $u$, and moreover, $(w,u)$ is an edge in  $E(H)$.  Then, by {\bf Rule 2} there is an edge connecting $u$ with all the particles in $V_1(u)$. Since $u$ has out-degree zero, necessarily $V_1(u) \cap F$ is empty. We deduce that $Q_1(u)$ is an available path for $u$ and then $1 \in C(u)$. We conclude that $\{1, 3\} \cap C(u) \neq \emptyset$.

\medskip

Claims {\bf 1} and {\bf 2} imply that $u$ satisfies the first two conditions of the definition of removable particle. It remains to show that the dependency graph of $F\setminus \{u\}$ is acyclic. In the following we call $E(H-u)$ the dependency graph of $F\setminus \{u\}$. \\

\noindent {\bf Claim 3:}  $E(H)\setminus (F \times \{u\})$ is a subset of $E(H-u)$.  \\

\textit{Proof of Claim 3}. Suppose that $E(H)\setminus (F \times \{u\})$ contains an edge $e = (w_1, w_2)$ not in $E(H-u)$. From all possible choices, let us pick the one that is created in the earliest iteration of Algorithm \ref{algo:dependency}. Observe that $e$ can not be created by {\bf Rule 1}. Indeed, if  $R_\ell(w_2)$ does not contains $u$ then necessarily $e$ is also an edge of $E(H-u)$, and if $R_\ell(w_2)$ contains $u$ then necessarily $u$ also belongs to $D^+(w_1)$ ($u$ must be different than $w_1$ because $u$ has out-degree zero). In both cases we have that if $e$ is created by {\bf Rule 1}, then $e$ must be an edge of $E(H-u)$. Then necessarily $e$ is created with {\bf Rule 2}. However, as $u$ has out-degree zero, it cannot block any path of another node. We deduce that $e$ can not exist and then $E(H)\setminus (F \times \{u\}) \subseteq E(H-u)$. \\

%

\noindent {\bf Claim 4:}  Let $e = (w_1, w_2)$ be an edge of $E(H-u)\setminus E(H)$. Then $w_1 \neq w_2$ (i.e. $e$ is not a self-loop), and  $w_1$ belongs to $S_{NE}(u)$. Moreover, if $w_2$ also belongs to $S_{NE}(u)$, then $w_2$ is closer to $u$ than $w_1$ in Manhattan distance, i.e. $|\row(w_2) - \row(u)| + |\col(w_2)-\col(u)| < |\row(w_1) - \row(u)| + |\col(w_1)-\col(u)|$.  \\

\textit{Proof of Claim 4}. Let $\ell_1< \dots < \ell_t$ be the iterations of Algorithm \ref{algo:dependency} in which an edge of $E(H-u)\setminus E(H)$ is created. If $e  = (w_1, w_2)$ is created in iteration $\ell_1$, then necessarily $e$ is created by {\bf Rule 1} (as $u$ has out-degree $0$ it can not block any path of  another node). Then necessarily $w_2$ is a neighbor of $u$ such that $D^+(w_2) = \emptyset$. Observe that $R_{\ell_1}(w_2) \setminus \{u\} \neq \emptyset$ because otherwise $(w_2, u)$ is an edge of $E(H)$. Then $w_1\neq w_2$ and $w_2$ is closer to $u$ than $w_1$ in Manhattan distance.  

Now suppose that the claim is true for all edges created in steps up to $\ell_i$, with $i\in \{1, \dots, t-1\}$, and suppose that $e=(w_1, w_2)$ is created in iteration $\ell_{i+1}$. 

 If $(w_1, w_2)$ is created by {\bf Rule 1}, then necessarily $w_2$ has an out-neighbor $w_3$ not present in $E(H)$, that was created in iterations $\ell_1, \dots, \ell_{i}$. Then $w_2$ belongs to $S_{NE}(u)$ by induction hypothesis. Observe that $R_{\ell_{i+1}}(w_2) \setminus w_3 \neq  \emptyset$, otherwise $(w_2,w_3)$ is an edge of $E(H)$. Then $w_1 \neq w_2$. Suppose by contradiction that $w_1$ is not contained in $S_{NE}(u)$. Then necessarily $w_1$ belongs to $S_{SE}(u)$, with $w_1 = (\row(w_2)+1, \col(w_2))$. Moreover, $w_3$ belongs to $S_{NE}(u)$. Then, by induction hypothesis, $w_3$ is closer to $u$ than $w_2$, meaning that $w_3 = (\row(w_2), \col(w_2)-1)$. This implies that $(w_1, w_3)$ is an edge of $E(H)$ and then $(w_1, w_2)$ must be also an edge of $E(H)$. We deduce that $w_1$ is contained in $S_{NE}(u)$. As the distance from $w_2$ to $w_3$ is fewer than the distance of $w_1$ to $w_3$ we deduce by the induction hypothesis that the distance from $w_2$ to $u$ is fewer than the distance from $w_1$ to $u$.  

 If $(w_1, w_2)$ is created by {\bf Rule 2}. Then necessarily $w_1$ has an incoming edge from a node $w_0 \in V(w_1)$ such that $(w_0, w_1)$ is not in $E(H)$ and was created in iterations $\ell_1, \dots, \ell_{i}$. From the induction hypothesis we know that $w_0$ belongs to $S_{NE}(u)$ and is closer to $u$ than $w_1$. Then $w_0$ belongs to $V_1(w_1)$. If $(w_0, w_1)$ is created with {\bf Rule 1}, then there is a path $P$ from $w_1$ to one of the adjacent particles of $u$ where all the edges in that path are edges in $E(H-u)\setminus E(H)$ created by {\bf Rule 1} (this is proven in the previous paragraph). Then $P$ either contains $w_2$ or contains a node $w' $ such that. $\row(w') = \row(w_2)$ and $\col(w') > \col(w_2)$. Indeed, in the other possibility is that $(v_0, v_1)$ and $(v_1, v_2)$ are edges of $P$, with $v_0 = (\row(w_2)-1, \col(w_2)+1)$, $v_1 = (\row(w_2)-1, \col(w_2))$ and $v_2= (\row(w_2)-1, \col(w_2)-1)$. However, this is impossible because in that case $(w_2, v_1)$ must be an edge of $E(H)$ and then $(v_0, v_1)$ is also an edge of $E(H)$. Then, there is a directed path of \emph{rule 1 edges} connecting $w_1$ and a neighbor of $u$. This implies that $w_2$ is ether in $S_{NE}(u)$ or in $S_{SW}(u)$. In any case we have that $(w_1, w_2)$ satisfies the conditions of the claim. It remains the case when $(w_0, w_1)$ is created with {\bf Rule 2} in some iteration $\ell_1, \dots, \ell_{i}$. Let $P' = v_0, \dots, v_k$ be the longest directed path  of edges in $E(H-u)\setminus E(H)$ that are created with {\bf Rule 2} in iterations $\ell_1, \dots, \ell_{i}$, such that $v_k = w_0$. Then $v_0$ must have an incoming neighbor $z_0$ such that $(z_0, w_0)$ is an edge of $E(H-u)\setminus E(H)$ created by {\bf Rule 1}, and such that $z_0 = (\row(v_0)-1, \col(v_0))$.  Following the same argument than in the previous case (i.e. taking the path of edges of $E(H-u)\setminus E(H)$ created by {\bf Rule 1} connecting $w_0$ with a particle adjacent to $u$) we deduce that $(w_1, w_2)$ satisfies the conditions of the claim. This finishes the proof of {\bf Claim 4}.

 Observe that {\bf Claim 4} implies that the edges of $E(H-u)\setminus E(H)$ form a directed acyclic graph.  On the other hand, we know that $H$ is acyclic. Finally, it is impossible to have a cycle containing a combination of edges in $E(H-u)\setminus E(H)$ and $E(H)$, because all the out-going neighbors of nodes in $E(H-u)\setminus E(H)$ also belong to $E(H-u)\setminus E(H)$. We deduce that the dependency graph of $F \setminus \{u\}$ is acyclic. We conclude that $u$ is a removable particle of $F$.
 
\end{proof}

We remark that the third condition of Definition \ref{def:dependency} is not necessarily deduced from the first two conditions, in the sense that there are figures with particles $u$ satisfying that $C(u)\neq \emptyset$ and $D^+(u) = \emptyset$, but that are not removable. For instance, the realizable figure depicted in Figure \ref{fig:dependecy_example2} contains a particle $u$ such that  $C(u)\neq \emptyset$, $D^+(u) = \emptyset$ and $F\setminus\{u\}$ is not realizable.

\begin{figure}
\begin{center}
\includegraphics[width=0.4\textwidth]{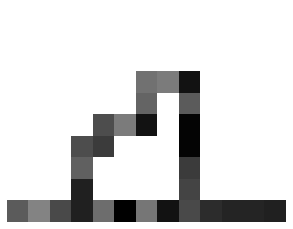}
\includegraphics[width=0.5\textwidth]{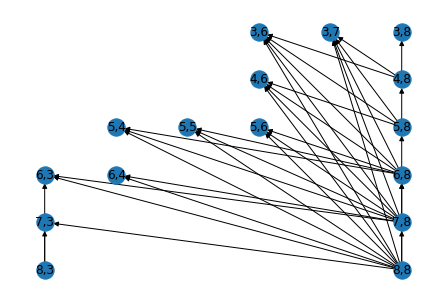}
\includegraphics[width=0.4\textwidth]{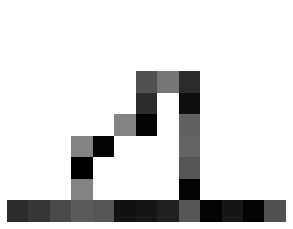}
\includegraphics[width=0.5\textwidth]{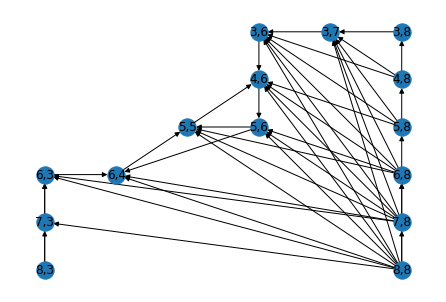}
\includegraphics[width=0.4\textwidth]{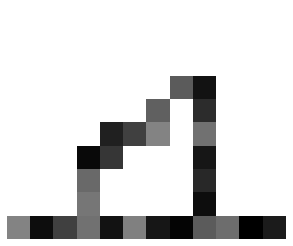}
\includegraphics[width=0.5\textwidth]{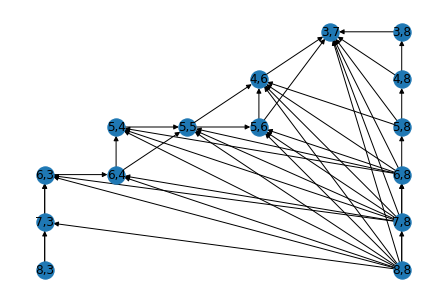}
\end{center}
\caption{An example of a figure containing a node of coordinates $u = (5,4)$ such that $C(u)\neq \emptyset$ (as $Q_1(u)$ is available), $D^+(u) = \emptyset$ but $F \setminus \{u\} $ contains a cycle. In the first row is depicted the figure $F$ and its dependency graph. And in the second row is depicted $F \setminus \{u\}$ and its corresponding dependency graph. Nodes $(5,5), (4,6), (5,6)$ form a cycle in this graph. Observe that $F$ does contain a removable particle, for instance $v = (3,6)$. In the third row is depicted $F \setminus \{ v\}$ and its corresponding (acyclic) dependency graph. }
\label{fig:dependecy_example2}
\end{figure}

We are now ready to give the main result of this section.

\begin{theorem}
A figure $F$ is 2DLA-realizable if and only its dependency graph $H$ is acyclic.
\end{theorem}

\begin{proof}
Let us assume that $H$ contains a cycle, and let $u$ be a node in a cycle of $H$. Lemma \ref{lem:dependency} implies  that $u$ must be placed before himself on any canonical realization. Therefore, $F$ does not admit a canonical realization. From Lemma \ref{canon} we deduce that $F$ is not realizable.

Conversely, suppose that $H = (F,E)$ is acyclic. Then, Lemma \ref{lem:tech1} imply that we can successively decompose $F$ finding removable nodes. The obtained canonical realization of $F$ is the output of Algorithm \ref{algo:construction}.
\end{proof}

\begin{corollary}
 2-\DReal~ is solvable in polynomial time. 
\end{corollary}

\begin{remark}
An implementation of the algorithm solving $2$-\DReal\ can be found in \url{https://github.com/pedromontealegre/2DLA-Realization}.
\end{remark}

\section{Conclusion}
The introduction of restrictions to discrete dynamical systems, coming from statistical physics, has been shown to change the computational complexity of its associated decision problem. We have once again observed this phenomenon with the changes in computational complexity when restricting the directions a particle can move in the DLA model. By adapting the \PC \ proof of the 4-{\DPred} \ we showed that both 3-{\DPred} and 2-{\DPred} are \PC.  For 1-{\DPred}, \ we tackled the generalized problem, \textsc{Ballistic Deposition} and showed that by exploiting the commutativity exhibited by the dynamics of the system, we created a non-deterministic log-space algorithm to solve the problem, and show that \Pred \ is \textbf{\textsc{NL}}-Complete. What is interesting to note is that the algorithm does not depend on the topological properties of the model, it exclusively works on the input word (the sequence of particle throws in this case). 

We finally showed that characterizing the shapes that are obtainable through the dynamics is an interesting problem, and exhibited that the computational problem associated is in $\Lspace$ for the one-directional dynamics, and in $\PP$ for the two-directional one. This is interesting due to the fact that the figures produced when executing the probabilistic versions of the dynamics, as shown in Figure \ref{fig:dynex}, the computational complexity of examining these shapes is efficient.

\subsection{Future Work}

An interesting extension to the presented problem is the one of determining, given a figure, what is the minimum amount of directions necessary to produce it, if it is achievable at all. We have shown that figures generated by the \Pred \ model are characterizable in $\Lspace$, and those of the two-directional model in $\PP$. It remains to see if the latter can be improved into an $\NC$ algorithm, or the problem is in fact \PC. The definition of the dependency graph seems intrinsically sequential, and we believe that is unlikely that the problem could be solved by a fast-parallel algorithm

In addition, the complexity classification of realization problem for three and four directions remains open. We believe that the definition of a dependency graph for these cases is plausible, but several assumptions (such as the existence of canonical realizations) must be redefined. 

A related problem, not treated in this article, consists in the \emph{reachability}, i.e. given a figure $F$ and a subfigure $F'$, decide if it is possible to construct $F$ given that $F'$ is already fixed. This problem can be also studied restricting the number of directions the particles are allowed to move in. Our belief is that there is room to create more complex gadgets, indicating that these problems are possibly \textbf{\textsc{NP}}-Complete at least for three or more movement directions.

\section*{Acknowledgments}
We would like to thank the anonymous reviewers for letting us know about the Ballistic Deposition model, which we had no previous knowledge~of. 

This work has been partially supported by ANID, via PAI + Convocatoria Nacional Subvenci\'on a la Incorporaci\'on en la Academia A\~no 2017 + PAI77170068 (P.M.), FONDECYT 11190482 (P.M.) FONDECYT 1200006 (E.G. and P.M.) and via CONICYT-PFCHA/Mag\'isterNacional/2019 - 22190497 (N.B.).

 \bibliographystyle{siamplain}
\bibliography{FSG}
\end{document}